\documentclass[12pt, titlepage]{article}

\usepackage[utf8]{inputenc}
\usepackage{remreset}
\usepackage[T1]{fontenc}
\usepackage[amsmath]{}
\usepackage[ngerman]{babel}
\usepackage{amsmath}
\usepackage{amsthm}
\usepackage{amssymb}
\usepackage{amsfonts}
\usepackage{amssymb}
\usepackage{graphicx}
\usepackage{color}
\usepackage[a4paper,lmargin={4cm},rmargin={2cm},tmargin={2.5cm},bmargin={2.5cm}]{geometry}
\usepackage{stmaryrd}

\newcommand{\uproman}[1]{\uppercase\expandafter{\romannumeral#1}}

\newtheorem{definition}{Definition}[section]
\newtheorem{proposition}{Satz}[section]
\newtheorem{theorem}{Theorem}[section]
\newtheorem{lemma}{Lemma}[section]
\newtheorem{remark}{Bemerkung}[section]
\newtheorem{example}{Beispiel}[section]

\begin{document}

\pagestyle{empty}

\begin{titlepage}
\title{Quantenlogische Systeme und Tensorprodukträume}
\date{2022}
\author{Tobias Starke}
\maketitle
\end{titlepage}

\thispagestyle{empty}

\newpage 
\thispagestyle{empty}
\quad 
\newpage

\begin{center}
     Danksagung
\end{center}
An dieser Stelle möchte ich mich kurz bei all denjenigen bedanken, die mich während der Anfertigung dieser Arbeit unterstützten. 

Als Erstes möchte ich mich bei Herr Professor Michael Breuß, der meine Bachelorarbeit betreute, für die vielen Anregungen und die konstruktive Kritik bedanken. Auch möchte ich mich bei Herr Marvin Kahra für die vielen hilfreichen Gespräche bedanken, mit deren Hilfe ich viele Probleme in Bezug auf diese Arbeit beseitigen konnte. Ebenfalls möchte ich auch meinem Bruder Tim Starke für seine Hilfsbereitschaft danken, mir bei bestimmten konzeptionellen Fragen weitergeholfen zu haben. 

Abschließend möchte ich mich auch bei meiner Mutter und meinem Großvater für die viele Unterstützung während des Studiums bedanken.

\newpage

\newpage 
\thispagestyle{empty}
\quad 
\newpage

\pagestyle{plain}
\thispagestyle{empty}

\tableofcontents

\newpage

\newpage 
\thispagestyle{empty}
\quad 
\newpage

\setcounter{page}{1}



\section{Einleitung}

Die etwa um das Jahr 1900 entstandene Theorie der Quantenmechanik, die zur Lösung von Problemen wie der Ultraviolettkatastrophe \cite{ehrenfest1911zuge} ersonnen wurde, ergänzte die bis dahin als gültig angenommene klassische Mechanik, was das physikalische Weltbild nachhaltig veränderte.

Im Gegensatz zur klassischen Mechanik wurden nun (massebehaftete, sich nicht relativistisch bewegende) Teilchen wie Elektronen im quantenmechanischen Formalismus über sogenannte Wellenfunktionen beschrieben, die der von Erwin Schrödinger Ende 1925 entdeckten Schrödinger Gleichung genügen \cite{schrodinger1926quantisierung}.

Mit der Entwicklung der Quantenmechanik im frühen 20. Jahrhundert stellte sich irgendwann die Frage, wie ein logisches Aussagensystem an ein physikalisches System aussähe, welches den Prinzipien der Quantenmechanik gerecht wird und inwiefern sich ein derartiges System von einem Logikmodell unterscheide, welches lediglich eine logische Beschreibung klassischer physikalischer Systeme bereitstellt. Damit war die sogenannte Quantenlogik geboren.

Anfang der 1930er wurden erste Überlegungen zur Konstruktion einer Quantenlogik unternommen. So bemerkte John von Neumann in seinem 1932 veröffentlichten Werk \textit{Mathematische Grundlagen der Quantenmechanik} \cite{von2013mathematische}, dass man Projektionsoperatoren auf einem Hilbertraum als Aussagen oder Propositionen an ein quantenmechanisches System verstehen kann. Mit Garrett Birkhoff arbeitete John von Neumann diese Idee weiter aus. Zusammen veröffentlichten sie im Jahre 1936 eine Arbeit \cite{birkhoff1936logic}, die eine erste Quantenlogik enthielt. Mit dieser war es möglich, die Propositionen wie in einem gewöhnlichen Logikmodell mittels spezieller logischer Operationen zu manipulieren. Fast 30 Jahre später, im Jahre 1963, präsentierte dann der amerikanische Mathematiker George Mackey eine Axiomatisierung dieser Quantenlogik mittels orthokomplementierter Verbände \cite{mackey2013mathematical}. Ab 1963 wurde dann die Forschung im Bereich der Quantenlogik vornehmlich von Constantin Piron und Josef-Maria Jauch in der Schweiz fortgesetzt.

Wir wollen in dieser Arbeit auf anschauliche Weise das Axiomensystem von George Mackey zur Konstruktion einer Quantenlogik motivieren und diese mit einem Logikmodell der klassischen Mechanik, wie in \cite{aerts1978physical} präsentiert, vergleichen. Ziel dieser Arbeit wird es sein, die Resultate der von Diederik Aerts und Ingrid Daubechies veröffentlichten Arbeit \textit{Physical justification for using the tensor product to describe two quantum systems as one joint system} \cite{aerts1978physical} ausführlich zu präsentieren, d.h. es soll gezeigt werden, wie spezielle zusammengesetzte Systeme in der klassischen Mechanik und der Quantenmechanik logisch beschrieben werden. 

Dazu werden wir wie in \cite{aerts1978physical} eine Klasse bestimmter zusammengesetzter physikalischer Systeme betrachten und diese axiomatisch definieren, um so mittels einiger Resultate der Verbands- und c-Morphismus-Theorie (\cite{piron1976foundations} und \cite{aerts1979structure}) zu zeigen, dass im quantenmechanischen Fall die Beschreibung der betrachteten zusammengesetzten Systeme notwendigerweise über Tensorprodukt-räume erfolgen muss.
\newpage



\newpage 
\quad 
\newpage

\section{Klassische Logik}

Um sich dem Begriff der sogenannten klassischen logischen Systeme \cite{shapiro2018classical} zu nähern betrachten wir sogenannte Elementaraussagen. Das sind logische Aussage, denen man einen Wahrheitswert zuordnen kann und die sich nicht in 'kleinere' Aussagen zerlegen lassen. Für klassische logische Systeme kann der Wahrheitswert dabei die Werte 0 oder 1 annehmen, wobei 0 für \textit{falsch} und 1 für \textit{wahr} steht. Mit $E$ bezeichnen wir im folgenden eine Menge von Elementaraussagen.

Zwischen diesen Elementaraussagen definieren wir logische Operationen, um Elementaraussagen zu neuen logischen Aussagen zusammenzusetzen. Die logischen Operationen die hier relevant sind, sind die Adjunktion $\wedge$ (Und-Verknüpfung), die Konjunktion $\vee$ (Oder-Verknüpfung), die Negation $\neg$ (Verneinung) und die Subjunktion $\to$ (Implikation). Auch diese zusammengesetzten Aussagen sollen nun über diese logische Operationen miteianander verknüpfbar sein. Mit $L(E)$ bezeichnen wir die Menge aller Aussagen, die durch fortgesetztes Anwenden der logischen Operationen entstehen, wobei die Elemetaraussagen Aussagen aus $E$ sind. Die Elemente aus $L(E)$ nennen wir logische Aussagen. 

Unter einem klassischen logischen Systeme versteht man nun eine Menge von logischen Aussagen, auf denen obige logische Verknüpfungen definiert sind, sodass Folgendes erfüllt ist:
\begin{itemize}
    \item[\textit{i)}] Es gilt das Bivalenzprinzip, d.h. jede Aussage ist entweder \textit{wahr} oder \textit{falsch}. In diesem Sinne ist das Logiksystem zweiwertig. 
    \item[\textit{ii)}] Es gilt das Prinzip der Extensionalität, d.h. der Wahrheitswert einer zusammengesetzten Aussagen ist eindeutig bestimmbar aus den einzelnen Wahrheitswerten derjenigen Elementaraussagen, aus dem die Aussage zusammengesetzt ist.
\end{itemize}
Man beachte, dass es sich bei den beiden Bedingungen, die wir an ein klassisches Logiksystem stellen, um sehr starke Restriktionen handelt, was die Menge an Aussagen angeht, die wir mittels solcher Systeme untersuchen können. 

So können wir beispielsweise Aussagen wie \textit{„Es regnet morgen“} nicht im Rahmen klassischer Logiksysteme untersuchen, da wir der Aussage zum jetzigen Zeitpunkt keinen Wahrheitswert im obigem Sinne zuordnen können, da wir uns heute noch nicht sicher sein können, ob es morgen regnet.

Ebenso sind Aussagen wie \textit{„Es ist möglich, dass Gauß kein Mathematiker ist“} und \textit{„Es ist möglich, dass eckige Kreise existieren“} nicht durch klassische Logiksysteme erfassbar, denn in beiden Fällen enthalten die Aussagen eine falsche Teilaussage (\textit{„Gauß ist kein Mathematiker“} und \textit{„Eckige Kreise existieren“}) und sind sonst identisch, unterschieden sich jedoch in ihrem Wahrheitswert. Denn während wir uns noch vorstellen können, dass Gauß nie zur Mathematik gefunden hat, so sind eckige Kreise nicht vorstellbar, was zeigt, dass solche Aussagen das Prinzip der Extensionalität verletzen. 

Um nun ein Gefühl für die zu betrachtenden Aussagen und für die Bedeutung der einzelnen logischen Operationen zu bekommen, wollen wir einige Beispiele betrachten:
\begin{example}
Wir betrachten ein klassisches Teilchen der Masse $m$ und wählen einen geeigneten Satz an Koordinaten, um dieses zu beschreiben. Die $i$-te Ortskoordinate bezeichnen wir dabei mit $x^{i}$, wobei $i = 1,2,3$ ist. 

Wir betrachten nun die logischen Aussagen $a,b$ und $c$ an das System. Angenommen, die Aussage $a$ sei $$\textit{„Das Teilchen habe zum Zeitpunkt $t$ den Impuls $p$“},$$ $b$ sei $$\textit{„Das Teilchen befinde sich zum Zeitpunkt $t$ am Ort $x = (x^1,x^2,x^3)$“}$$ und die Aussage $c$ sei $$\textit{„Das Teilchen habe zum Zeitpunkt $t$ die kinetische Energie $T$“}.$$ Wir bemerken zuerst, das man jeder dieser Aussagen auf eindeutige Weise den Wahrheitswert \textit{wahr} oder \textit{falsch} zuordnen kann, denn wenn wir beispielsweise den Impuls des Teilchens zum Zeitpunkt $t$ messen, so entspricht er entweder dem Impuls $p$ aus der Aussage $a$ oder eben nicht. Analog argumentiert man für die Aussagen $b$ und $c$.

Die Aussagen $a,b$ und $c$ sind elementare Aussagen an den Zustand des Teilchens, d.h. in diesem Beispiel gilt $a,b,c \in E$, wobei $E$ hier die Menge aller elementaren Aussagen an das zu beschreibende Teilchen sind. Wir wollen nun mittels obiger eingeführter logischer Operationen zusammengesetzte Aussagen konstruieren und deren Bedeutung erklären: 

Die Aussage $a \wedge b$ bedeute $$\textit{„Zum Zeitpunkt $t$ sei das Teilchen am Ort x \textbf{und} habe den Impuls p“}.$$ Falls die Aussage $a$ oder $b$ oder beide \textit{falsch} sein sollten, so folgt daraus natürlich, dass auch die Aussage $a \wedge b$ \textit{falsch} sein muss, da mindestens ein Teil der Aussage nicht zutreffend wäre. Nur wenn die Aussagen $a$ und $b$ beide \textit{wahr} sind, folgt, dass auch die Aussage $a \wedge b$ \textit{wahr} ist.

Die Aussage $a \vee b$ bedeute $$\textit{„Zum Zeitpunkt $t$ sei das Teilchen am Ort x \textbf{oder} habe den Impuls p“}.$$ Auch hier ist schnell ersichtlich, dass die Aussage $a \vee b$ genau dann \textit{wahr} ist, sobald die Aussage $a$ oder $b$ oder beide \textit{wahr} sind, da in diesem Fall lediglich gefordert wird, dass mindestens eine Teilaussage von $a \vee b$ \textit{wahr} sein soll. Nur wenn die Aussage $a$ und $b$ beide \textit{falsch} sind, folgt, dass auch $a \vee b$ \textit{falsch} ist.

Die Aussage $\neg c$ bedeute $$\textit{„Das Teilchen habe zum Zeitpunkt $t$ \textbf{nicht} die kinetische Energie $T$“}.$$ Da es sich bei der Negation einer Aussage lediglich um eine Verneinung dieser handelt, wird der Wahrheitswert der Aussage $\neg c$ aus der Umkehrung des Wahrheitswertes von $c$ bestimmt und umgedreht. Das heißt, wenn z.B. die Aussage $c$ \textit{wahr} ist, so ist $\neg c$ \textit{falsch} und umgedreht, und wenn $c$ \textit{falsch} ist, so ist $\neg c$ \textit{wahr} und umgedreht.

Die Aussage $a \to c$ bedeute 
\begin{align*}
    &\textit{„Zum Zeitpunkt $t$ gilt: Hat das Teilchen den Impuls p,} \\
    &\textit{so \textbf{folgt}, dass das Teilchen die kinetische Energie T besitzt“}.
\end{align*}
Dabei ist die Aussage $a \to b$ nur dann \textit{falsch}, wenn $a$ \textit{wahr} und $b$ \textit{falsch} ist, da man auch unter falschen Vorraussetzungen zu richtigen Schlussfolgerungen kommen kann.

Man beachte, dass der Wahrheitswert jeder oben beschriebenen zusammengesetzten Aussage sich aufgrund der Definition der hier verwendeten logischen Operationen aus den Wahrheitswerten der einzelnen Teilaussagen bestimmen lässt, d.h. die obigen zusammengesetzten Aussagen erfüllen das Prinzip der Extensionalität.
\end{example}

Mittels Wahrheitstabellen lässt sich außerdem leicht nachprüfen, dass im Rahmen klassischer Logiksysteme die Subjunktion keine elementare Operation auf der Menge der logischen Aussagen sein kann, denn es gilt für $a,b \in L(E)$
\begin{align}
    a \to b  \iff  \neg a \vee b.
\end{align}

\newpage



\newpage 
\quad 
\newpage

\section{Logikmodell der klassischen Mechanik}

Die folgenden Ausführungen beziehen sich auf \cite{herbert2002classical} und \cite{arnol2013mathematical}.


\subsection{Klassische Mechanik}

Es gibt viele verschiedene Formalismen, um eine Theorie der klassischen Mechanik zu formulieren. Wir wollen hier den Lagrange- und Hamiltonformalismus kurz präsentieren, um darauf aufbauend ein Logikmodell der klassischen Mechanik gemäß \cite{aerts1978physical} zu konstruieren.

Wir beschränken uns zunächst auf die Beschreibung von Systemen, die ein einzelnes massebehaftetes Teilchen enthalten, kurz Einteilchensysteme. Wir werden dabei ab jetzt das zu beschreibende klassische Einteilchensystem mit $\sigma$ und den Ort des darin enthaltenen Teilchens zum Zeitpunkt $t$ mit $x(t) = (x^1(t),x^2(t),x^3(t))$ bezeichnen. Unterliegt das System $m \leq 3$ Zwangsbedingungen, d.h. Nebenbedingungen, die die Bewegungsfreiheit und damit die mögliche geometrische Bahn des Teilchens einschränken, so lassen sich die einzelnen Ortskoordinaten $x^{i}$, $i \in \{1,2,3\}$, durch $3-m$ unabhängige Koordinaten $q^{j}$, $j \in \{1,...,3-m\}$, ausdrücken. Diese unabhängige Koordinaten bezeichnen wir als generalisierte Koordinaten.

Das sogenannte Hamiltonsche Integralprinzip sagt nun aus, dass sich die Dynamik des Systems $\sigma$ für $t \in [t_i,t_e]$ aus dem Extremum des sogenannten Wirkungsfunktionals
\begin{align}
    S[q^1(t),...,q^{3-m}(t)] = \int_{t_i}^{t_e} L(q^1(t),...,q^{3-m}(t),\dot{q}^1(t),...,\dot{q}^{3-m}(t),t) \,dt
\end{align}
ergibt. Dabei bezeichnet $\dot{q}^{i}$ die Zeitableitung von $q^{i}$ für alle $i \in \{1,...,3-m\}$ und $L$ ist die sogenannte Lagrange-Funktion, die im Zentrum des Lagrange-Formalismus steht. 

Die Forderung, dass die tatsächlich realisierte Bahn eines Teilchens die Wirkung extremal werden lässt, ist äquivalent zur Forderung
\begin{align}
    \frac{d}{dt}\biggl(\frac{\partial L}{\partial \dot{q}^{i}}\biggr) - \biggl( \frac{\partial L}{\partial q^{i}} \biggr) = 0, \qquad \forall i \in \{1,...,3-m\}
\end{align}
Diese Gleichungen werden Lagrange-Gleichungen zweiter Art genannt. Für die klassischen physikalischen Systeme, die wir hier betrachten, ist die Lagrange-Funktion dabei gegeben durch
\begin{align}
    L = T - V,
\end{align}
wobei $T$ die kinetische Energie und $V$ die potentielle Energie des Teilchens bezeichnet. 

Mittels der Lagrange-Funktion lassen sich nun kanonische bzw. generalisierte Impulse definieren:
\begin{align}
    p_i = \frac{\partial L}{\partial \dot{q}^{i}}, \qquad \forall i \in \{1,...,3-m\}
\end{align}
Um nun vom Lagrange-Formalismus zum Hamilton-Formalismus überzugehen, machen wir den Übergang 
\begin{align}
    (q^1,...,q^{3-m},\dot{q}^1,...,\dot{q}^{3-m},t) \longrightarrow (q^1,...,q^{3-m},p_1,...,p_{3-m},t)
\end{align}
Mittels einer Legendr\'{e}-Transformation erhalten wir dabei aus der Lagrange-Funktion $L = L(q^1,...,q^{3-m},\dot{q}^1,...,\dot{q}^{3-m},t)$ die sogenannte Hamilton-Funktion $H = H(q^1,...,q^{3-m},p_1,...,p_{3-m},t)$:
\begin{align}
    H = \sum_{i=1}^{3-m}p_i \dot{q}^{i} - L
\end{align}
Aus den Lagrange-Gleichungen zweiter Art folgen die Hamiltonischen Bewegungsgleichungen:
\begin{align}
    \frac{d}{dt}q^{i} = \frac{\partial H}{\partial p_{i}}, \qquad \frac{d}{dt}p_{i} = - \frac{\partial H}{\partial q^{i}}, \qquad \forall i \in \{1,2,3\}
\end{align}

Ein Teilchen wird nun also im Hamiltonformalismus der klassischen Mechanik zu jeder Zeit $t \in \mathbb{R}$ vollständig durch die Angabe seines generalisierten Ortes $(q^1(t), q^2(t), q^3(t))$ und seines generalisierten Impulses $ p = (p_1(t), p_2(t), p_3(t))$ beschrieben. Zum Zeitpunkt $t$ kann der Zustand des Teilchens daher als Punkt im sogenannten Phasenraum $\Omega$ verstanden werden, der salopp als Menge aller generalisierter Orte und Impulse aufgefasst werden kann, d.h. 
\begin{align}
    \Omega := \{ (q^1, q^2, q^3, p_1, p_2, p_3) : q^{i} \in \mathbb{R} \: \forall i = 1,2,3;\: p_j \in \mathbb{R} \: \forall j = 1,2,3 \}.
\end{align}
Der Phasenraum stellt also den Zustandsraum im Hamiltonformalismus dar. Ein $\textbf{x} \in \Omega$ heißt Zustandsvektor vom System $\sigma$. Die Dynamik des Teilchens im System $\sigma$ für die Zeiten $t \in [t_i,t_e]$ wird dann durch die Phasenraum-Kurve $\gamma$, mit $\gamma(t) = (q^1(t),..., q^{3-m}(t), p_1(t),...,p_{3-m}(t))$ erfasst.

Nach dieser eher abstrakten Darstellung der Theorie wollen wir nun ein kleines Beispiel betrachten:

\begin{example}
Wir wollen jetzt im Folgenden den eindimensionalen harmonischen Oszillator berechnen:
\\ \\ Wir bemerken zuerst, dass das System keiner Zwangsbedingung unterliegt und wir somit genau einen Freiheitsgrad im eindimensionalen vorliegen haben. Wir bezeichnen mit $x$ die eine Raumkoordinate und identifizieren diese mit der generalisierte Koordinate $q$. Das Teilchen, welches wir beschreiben wollen, habe die Masse $m$. Die kinetische Energie $T$ des Teilchens ist gegeben durch
\begin{align}
    T = \frac{1}{2}m \dot{x}^2
\end{align}
und die potentielle Energie ist gegeben durch
\begin{align}
    V = \frac{1}{2}Dx^2,
\end{align}
wobei D die sogenannte Federkonstante ist. Die Langrange-Funktion ist damit gegeben durch
\begin{align}
    L(x,\dot{x}) = T - V = \frac{1}{2}m \dot{x}^2 - \frac{1}{2}Dx^2
\end{align}
und der zu $x$ gehörige generalisierte Impuls ist
\begin{align}
    p = \frac{\partial L}{\partial \dot{x}} = m \dot{x} \qquad \Longleftrightarrow \qquad \dot{x} = \frac{p}{m}
\end{align}
Die Hamilton-Funktion $H$ ist dann 
\begin{align}
    H =& \: \dot{x}p - L(x,\dot{x}) \nonumber\\ =& \: \frac{p^2}{m} - \frac{1}{2}m\biggl(\frac{p}{m}\biggr)^2 + \frac{1}{2}Dx^2 \nonumber\\ =& \: \frac{p^2}{2m} + \frac{1}{2}Dx^2 \nonumber\\ =& \: H(x,p).
\end{align}
Eingesetzt in die Hamiltonschen Bewegungsgleichungen ergibt sich 
\begin{align}
    \dot{x} = \frac{\partial H}{\partial p} = \frac{p}{m}, \qquad \dot{p} = -\frac{\partial H}{\partial x} = -Dx 
\end{align}
Differenzieren des ersten Ausdrucks nach der Zeit liefert
\begin{align}
    \ddot{x} = \frac{\dot{p}}{m} = -\frac{D}{m}x \qquad \Longleftrightarrow& \: \qquad m\ddot{x} = -Dx \nonumber\\ \Longleftrightarrow& \: \qquad \ddot{x} = -\omega_0^2x,
\end{align}
wobei $\omega_0^2 = \frac{D}{m}$ ist. Mittels eines Exponentialansatzes erhält man
\begin{align}
    x(t) = A \: sin(\omega_0 t + \phi),
\end{align}
wobei sich die beiden Konstanten $A$ und $\phi$ aus Anfangsbedingungen ergeben. Daraus folgt
\begin{align}
    p(t) = \omega_0 m A \: cos(\omega_0 t + \phi)
\end{align}
Damit wird die Dynamik des Systems durch die Phasenraum-Kurve 
\begin{align}
    \gamma(t) = (x(t), p(t)) = (A \: sin(\omega_0 t + \phi),\omega_0 m A \: cos(\omega_0 t + \phi))
\end{align}
beschrieben, wobei $\mathcal{X} = \{\gamma(t): t \in \mathbb{R}\} \subset \Omega$.
\end{example}


\subsection{Klassisches Logiksystem}

Zu diesem System betrachten wir nun eine Menge von logischen Aussagen $L$ an das System $\sigma$, bei denen es sich um Ja/Nein-Experimente handelt, d.h. die Aussagen beziehen sich auf den Zustand von $\sigma$, der nach einer Messung am System zur Zeit $t$ in Erfahrung gebracht wird. Diese Aussagen können nur mit \textit{Ja} oder \textit{Nein} beantwortet werden. 

Auf $L$ definieren wir wie oben die logischen Operationen, mit deren Hilfe sich unter anderem Aussagen miteinander verknüpfen lassen. Der Zustand vom System $\sigma$ lässt sich, wie oben erwähnt, zu jedem Zeitpunkt $t$ eindeutig über die Angabe des Ortes und des Impulses des Teilchens zu dieser Zeit bestimmen. Da die logischen Aussagen an das System $\sigma$ Aussagen über den Zustand des Systems sind, lassen sich diese Aussagen im folgenden Sinne als Teilmenge von $\Omega$ verstehen: Zu jedem Zeitpunkt $t$ ist der Zustand des Systems eindeutig ableitbar aus den sechs Einträgen des Zustandsvektors $(q^1(t),..., q^{3-m}(t), p_1(t),..., p_{3-m}(t))$. Demzufolge kann jede Aussage über den Zustand des Systems zur Zeit $t$ und damit jedes Ja/Nein Experiment aus $L$ als Aussage über den Wertebereich der einzelnen Einträgen des Zustandsvektors angesehen werden. Werden bei einer Aussage über das System gewisse Einträge des Zustandsvektors nicht näher spezifiziert, so werden diese als frei wählbar angenommen. Auf diese Weise korrespondiert zu jedem Ja/Nein-Experiment aus $L$ und damit zu jeder Aussage über den Zustand von $\sigma$ eine Teilmenge des Phasenraums $\Omega$. Daraus folgt, dass die Menge $L$ mit der Potenzmenge $\mathcal{P}(\Omega)$ vom Phasenraum $\Omega$ identifiziert werden kann. 

Betrachten wir das Beispiel $3.1$, so wäre z.B. eine mögliche Aussage an das System $\sigma$ \textit{„Der Ort und der Impuls des Teilchen ist für einen festen Zeitpunkt $t \in \mathbb{R}$ gegeben durch $x(t) = A \: sin(\omega_0 t + \phi)$ und $p(t) = \omega_0 m A \: cos(\omega_0 t + \phi)$“}. Diese Aussage würde dann für jeden festen Zeitpunkt $t$ der einelementigen Menge $\{\gamma(t) = (x(t), p(t))\}$ entsprechen.

Eine andere mögliche Aussage, wieder bezogen auf das Beispiel $3.1$, wäre \textit{„Das Teilchen befinde sich zur Zeit $t$ am Ort $x(t)$“}. Diese Aussage würde dann der Menge $\{(x(t),p) : p \in \mathbb{R}\}$ entsprechen.

Es gilt weiter, dass die logischen Operationen auf $L$ in einer natürlichen Beziehung mit den auf $\mathcal{P}(\Omega)$ definierbaren Mengenoperationen $\cup$ (Mengenvereinigung), $\cap$ (Mengendurchschnitt) und $\cdot^C$ (Komplement-Operation) stehen, denn sei $F$ diejenige Abbildung, die jeder Aussage an das System $\sigma$ ihre im obigem Sinne korrespondierende Teilmenge im Phasenraum $\Omega$ zuordnet, so gilt  
\begin{align}
    &F(a \vee b) = F(a) \cup F(b) \\ &F(a \wedge b) = F(a) \cap F(b) \\ &F(\neg a) = F(a)^C \\ &F(a \to b) = F(a)^C \cup F(b)
\end{align}

Die algebraische Struktur $(\mathcal{P}(\Omega), \cup, \cap, \cdot^C)$, die wir im Folgenden Propositionensystem (der klassischen Mechanik) nennen wollen, stellt also im Rahmen der klassischen Mechanik ein geeignetes Logikmodell dar, um Aussagen an das klassische System $\sigma$ zu studieren. Im folgenden werden wir sehen, dass dieses Propositionensystem ein Beispiel eines sogenannten vollständigen, orthokomplementierten, distributiven und atomaren Verbandes ist. 

\newpage

\subsection{Verbandsstruktur eines klassischen Logiksystems}

Eine wichtige und in dieser Arbeit sehr zentrale mathematische Struktur, die im Folgenden immer wieder auftauchen wird, ist die des Verbandes \cite{berghammer2012ordnungen}:

\begin{definition}
Ein Verband ist gegeben durch ein Tripel $\mathbb{V} = (V, \sqcup, \sqcap)$, bestehend aus einer Grundmenge $V$ und zwei binären Verknüpfungen $\sqcup, \sqcap : V \times V \longrightarrow V$ mit den folgenden Eigenschaften:
\begin{itemize}
    \item [\textit{i)}] $a \sqcup b = b \sqcup a$ und $a \sqcap b = b \sqcap a$, \: $\forall a,b \in V$
    \item [\textit{ii)}] $a \sqcup (b \sqcup c) = (a \sqcup b) \sqcup c$ und $a \sqcap (b \sqcap c) = (a \sqcap b) \sqcap c$, \: $\forall a,b,c \in V$
    \item [\textit{iii)}] $a \sqcup (a \sqcap b) = a$ und $a \sqcap (a \sqcup b) = a$, \: $\forall a,b \in V$
\end{itemize}
\end{definition} 

Im Falle obigen Propositionensystems wählen wir $V = \mathcal{P}(\Omega), \sqcup = \cup$ und $\sqcap = \cap$. Es ist leicht einzusehen, dass damit tatsächlich ein Verband definiert wird.\\

Es lässt sich nun auf jedem Verband eine partielle Ordnung, d.h. eine reflexive, transitive und antisymmetrische zweistellige Relation $\leq$, durch 
\begin{align}
    a \leq b \: :\Longleftrightarrow  \: a \sqcup b = b
\end{align}
definieren. Es lässt sich leicht einsehen, dass im Falle unseres obigen Propositionensystems $(\mathcal{P}(\Omega),\cup, \cap)$ und bezüglich der eben definierten $\leq$ Relation, die hier mit der Teilmengenbeziehung $\subseteq$ zusammenfällt, es zu jeder beliebigen Familie $(a_i)_{i \in I}$ aus $\mathcal{P}(\Omega)$ eine kleinste obere Schranke $\cup_{i \in I}b_i$ und eine größte untere Schranke $\cap_{i \in I}b_i$ gibt. Verbände mit dieser Eigenschaft nennt man vollständig. 

Da wir Verbände als Modelle von Logiken verstehen wollen, benötigen auf diesen eine Art Negation bzw. eine Operation, die wir später als eine logische Negation identifizieren wollen:

\begin{definition}
Ein Verband $\mathbb{V}$ heißt orthokomplementiert, falls auf diesem eine Orthokomplementation $'$ definiert ist. Dabei handelt es sich um eine bijektive Selbstabbildung auf $V$, sodass $\forall a,b \in V$ gilt:
\begin{itemize}
    \item [\textit{i)}] $(a')' = a$
    \item [\textit{ii)}] $a \leq b \rightarrow b' \leq a'$
    \item [\textit{iii)}] $a \sqcup a' = V =: \mathbf{I}$ und $a \sqcap a' = \mathbf{0}$, wobei $\mathbf{0} := \bigsqcap_{b \in V}b$ 
\end{itemize}
Wir nennen $\mathbf{I}$ Identitätselement.
\end{definition}

Für unser konkretes Beispiel sehen wir, dass die Komplement-Operation $\cdot^C$ eine Orthokomplementierung für den Verband $(\mathcal{P}(\Omega),\cup, \cap)$ darstellt.\\

Als Nächstes definieren wir eine wichtige Klasse von Verbänden:

\begin{definition}
Ein Verband $\mathbb{V}$ heißt distributiv, falls $\forall a,b,c \in V$ gilt, dass $a \sqcup (b \sqcap c) = (a \sqcup b) \sqcap (a \sqcup c)$.
\end{definition}

Auch hier ist durch unsere obige Setzung $\sqcup = \cup$ und $\sqcap = \cap$ leicht einzusehen, dass der Verband $(\mathcal{P}(\Omega),\cup, \cap)$ ein distributiver Verband ist. 

Wie wir später noch sehen werden, ist ein Verband, welcher ein Logiksystem der Quantenmechanik darstellt, im Gegensatz zu unserem eben definierten Logikverband der klassischen Mechanik, ein nicht distributiver Verband, was einen wichtigen Unterschied zwischen klassischen und quantenmechanischen Logikverbänden darstellt.

Zum Abschluss dieses Abschnittes noch eine weitere Definition, mit deren Hilfe wir später die klassischen und quantenmechanischen Logikverbände besser miteinander vergleichen können. 

\begin{definition}
Sei $\mathbb{V}$ ein Verband. Ein Atom $p \in V$ ist ein Element aus $V \setminus \{0\}$ mit der Eigenschaft, dass für alle $a \in V$ mit $a \leq p$ folgt, dass $a = \mathbf{0}$ oder $a = p$ ist. Ein Verband $\mathbb{V}$ heißt atomar, falls für jedes Element $a \in V$ mit $a \neq \mathbf{0}$ ein Atom $p \in V$ existiert, sodass $p \leq a$ ist.
\end{definition}

Im Verband $(\mathcal{P}(\Omega),\cup, \cap)$ stellen die Singletons, also die einelementigen Mengen aus $\mathcal{P}(\Omega)$, die Atome dar. Auch hier sieht man leicht, dass der Verband atomar ist.\\

Zusammengefasst ist damit ein Logiksystem der klassischen Mechanik gegeben durch den vollständigen, orthokomplementierten, distributiven und atomaren Verband $(\mathcal{P}(\Omega), \cup, \cap, \cdot^C)$. Dabei stellt $\Omega$ die Vereinigung aller Atome des Verbandes dar. Allgemein sagt man, dass jeder vollständige, orthokomplementierte, distributive und atomare Verband ein Logikmodell der klassischen Mechanik darstellt.

\newpage



\section{Logikmodell der Quantenmechanik}

Im Folgenden beziehen wir uns auf \cite{cohen2007quantenmechanik} und \cite{mackey2013mathematical}. Für Erklärungen der Begriffe Messbarkeit, Wahrscheinlichkeitsmaß und Wahrscheinlichkeitsdichte, siehe \cite{bauer2019wahrscheinlichkeitstheorie} und \cite{elstrodt1996mass}.


\subsection{Quantenmechanik}

Im Gegensatz zur klassischen Mechanik wird die Dynamik eines (nichtrelativistischen) quantenmechanischen Einteilchen-Systems $\Sigma$ mit Potential $V$ durch die sogenannte Schrödinger Gleichung beschrieben: 
\begin{align}
    i \hbar \frac{\partial}{\partial t} \psi (r,t) = \biggl(-\frac{\hbar^2}{2m} \Delta + V(r,t) \biggl) \psi (r,t) \label{eq:S}
\end{align}
Dabei ist $i$ die imaginäre Einheit, $\hbar$ das reduzierte Planksche Wirkungsquantum, $m$ die Masse, $\Delta$ der Laplaceoperator, $r = (x^1, x^2, x^3)$ der Ort mit den kartesischen Koordinaten $x^{j}$ und $\psi$ eine messbare Funktion, die sogenannte Wellenfunktion, die das physikalische System beschreibt. Den Ausdruck
\begin{align*}
    \biggl(-\frac{\hbar^2}{2m} \Delta + V(r,t) \biggl) =: \hat{H}
\end{align*}
auf der rechten Seite von \eqref{eq:S} nennt man Hamiltonoperator.

Für zeitunabhängige Probleme, bei denen sich also das Potential mit der Zeit nicht ändert, kann man mittels des Ansatzes 
\begin{align}
    \psi (r,t) = \tilde{\psi}(r)\exp{\bigg(\frac{-iE}{\hbar}t\bigg)}
\end{align}
die Zeitabhängigkeit von \eqref{eq:S} beseitigen und gelangt zur zeitunabhängigen Schrödingergleichung 
\begin{align}
    \hat{H}\tilde{\psi} = E \tilde{\psi}, \label{eq:Z}
\end{align}
wobei $E$ die Energie des Systems ist.

Die Wellenfunktion $\psi$ muss dabei zu jeder Zeit $t$ die Bedingung der quadratintegrabilität erfüllen, d.h. in unserem Kontext, dass 
\begin{align}
    \int_{\mathbb{R}^3} \psi(r,t)^{*} \psi(r,t)\,d \mu (r) = \int_{\mathbb{R}^3} |\psi(r,t)|^2 \,d \mu (r) < \infty
\end{align}
für alle Zeiten $t \in \mathbb{R}$ gilt. Dabei ist $\mu$ das Lebesgue-Maß im $\mathbb{R}^3$, d.h. obiges Integral ist als Lebesgue-Integral zu verstehen, und $^{*}$ ist die komplexe Konjugation. 

Die Menge aller solcher (messbaren) quadratintegrabler Funktionen für festgehaltenes $t$ wird mit $\mathcal{L}^2(\mu)$ bezeichnet. Der Grund für diese Forderung ist der Folgende: Die Wellenfunktion $\psi$ ist eine allgemein komplexwertige Funktion, die das physikalische System $\Sigma$ modelliert. Es kann sich bei $\psi$ also nicht um eine direkt beobachtbare, d.h. messbare Größe von $\Sigma$, einer sogenannten Observablen von $\Sigma$, handeln. Stattdessen soll über das Betragsquadrat der Wellenfunktion $\psi (r,t)$ zu jeder festen Zeit $t$ eine Aufenthaltswahrscheinlichkeitsdichte für das Teilchen in $\Sigma$ definiert werden, sodass 
\begin{align}
    \int_{\mathbb{R}^3} |\psi(r,t)|^2 \,d \mu (r) = 1 \label{eq:N}
\end{align}
ist. Man fordert also die Normiertheit von $\psi$, damit über \eqref{eq:N} eine Wahrscheinlichkeitsdichte definiert wird (Man beachte, dass für zeitunabhängige Probleme die Wellenfunktionen $\psi$ und $\tilde{\psi}$ die selbe Wahrscheinlichkeitsdichte definieren.). 

Dadurch wird auch ein weiterer Unterschied zur klassischen Mechanik klar: Im Gegensatz zur klassischen Physik sind in den Modellen der Quantenphysik nur wahrscheinlichkeitstheoretische Aussagen möglich! Die Funktion 
\begin{align}
    \mathbb{P}_t^{\psi}(B) = \int_{B} |\psi(r,t)|^2 \,d \mu (r)
\end{align}
soll also zu jeder festen Zeit $t$ für alle (messbaren) $B \subseteq \mathbb{R}^3$ ein Wahrscheinlichkeitsmaß definieren. Aus diesem Grund wird auch klar, warum wir das Lebesgue-Integral nutzen: Damit $\mathbb{P}_t^{\psi}$ für festes $t$ ein Wahrscheinlichkeitsmaß ist, müssen für dieses die Kolmogorov-Axiome erfüllt sein und im Gegensatz zum herkömmlichen Riemann-Integral ist das Lebesgue-Integral Sigma-Additiv!  

Eine wichtige Beobachtung ist nun die folgende: Wie wir oben bereits bemerkt haben, ist die Wellenfunktion $\psi$ nicht direkt beobachtbar. Mit dem Experiment vergleichbare physikalische Vorhersagen über das Teilchen erhalten wir erst einmal nur über das Wahrscheinlichkeitsmaß $\mathbb{P}_t^{\psi}$. Wird ein quantenmechanisches System nun einmal über die Wellenfunktion $\psi_1$ und einmal über die Wellenfunktion $\psi_2$ modelliert, und $\psi_1$ und $\psi_2$ unterscheiden sich nur auf Lebesgue-Nullmengen (Mengen vom Lebensgue-Maß Null) voneinander, so lässt sich leicht nachprüfen, das zu jeder festen Zeit $t$ die wie oben aus $\psi_1$ und aus $\psi_2$ definierten Wahrscheinlichkeitsmaße $\mathbb{P}_t^{\psi_1}$ und $\mathbb{P}_t^{\psi_2}$ bezüglich beliebiger (messbarer) Eingabe $B \subseteq \mathbb{R}^3$ immer dieselbe Ausgabe liefern. Das heißt, physikalisch kann man nicht zwischen $\psi_1$ und $\psi_2$ unterscheiden, was nahe legt, das man alle Wellenfunktionen, die sich nur auf Lebesgue-Nullmengen voneinander unterscheiden, miteinander identifizieren sollte. Konkret bedeutet das, dass man folgende Äquivalenzrelation einführt: 
\begin{align}
    \psi_1 \sim \psi_2 \: :\Longleftrightarrow \: \psi_1 - \psi_2 \in \mathcal{N} = \{ f  \: \textrm{messbar} : f = 0 \: \mu-\textrm{fast überall} \}
\end{align}
Weiter wissen wir, dass die Schrödingergleichung eine lineare, homogene, partielle Differentialgleichung der Ordnung Zwei ist. Definieren wir auf der Menge $\mathcal{L}^2$ eine punkweise Vektoraddition und Skalarmultiplikation, so sehen wir, dass damit der Lösungsraum der Schrödingergleichung eine Vektorraumstruktur aufweist. Zusammengefasst definieren wir damit den Raum aller Zustandsvektoren eines quantenmechanischen Systems zu einem Zeitpunkt $t$ als den Quotientenvektorraum 
\begin{align}
    L^2(\mu) = \mathcal{L}^2(\mu) / \mathcal{N}
\end{align}
Auf dem $L^2(\mu)$ lässt sich mittels 
\begin{align}
    \bigl\langle f , g \bigr\rangle := \int_{\mathbb{R}^3} f^*(r)  g(r) \,d \mu (r)
\end{align}
ein komplexes Skalarprodukt, d.h. eine positiv definite, hermetische Sesquilinearform, definieren, so dass für die induzierte Norm $\| f \| := \sqrt{\langle f , f \rangle}$ gilt, dass 
\begin{align}
    \| f \|^2 = \int_{\mathbb{R}^3} |f(r)|^2 \,d \mu (r).
\end{align}
Damit reproduziert die vom Skalarprodukt induzierte Norm für jeden festen Zeitpunkt $t$ die quadratintegrabilitäts-Bedingung für die Wellenfunktionen. Man kann zeigen \cite{brezis2011functional}, dass der $L^2(\mu)$ zusammen mit $\langle \cdot , \cdot \rangle$ ein separabler Hilbertraum ist, was bedeutet, dass der $L^2(\mu)$ eine abzählbare Basis besitzt und das alle Cauchy-Folgen bezüglich der vom komplexen Skalarprodukt induzierten Norm konvergieren.  


\subsection{Projektive Hilberträume als Zustandsräume}

Aufgrund obiger Normierungsbedingung, die wichtig für die wahrscheinlichkeitstheoretische Interpretation der Quantenmechanik ist, kann man sich die Fragen stellen, ob Zustände in der Quantenmechanik, die zu jeder Zeit $t$ als Element des $L^2(\mu)$ aufgefasst werden, aus physikalischer Sicht notwendigerweise über normierte Zustandsvektoren beschrieben werden müssen. 

Wie sich zeigt, handelt es sich dabei tatsächlich nur um eine bequeme Konvention: Nehmen wir an, dass der Zustandvektor $\psi \in L^2(\mu)$ eines Systems $\Sigma$ normiert ist, so liefert das Betragsquadrat per Definition eine Wahrscheinlichkeitsdichte und mittels $\mathbb{P}_t^{\psi}$ lassen sich dann wie oben die Wahrscheinlichkeiten einer Ortsmessung in einem (messbaren) Raumvolumen ausrechnen. Physikalisch äquivalent lässt sich das System aber auch durch einen Zustandsvektor $\phi \in L^2(\mu)$ beschreiben, sofern gilt, dass $\psi = \lambda \phi$ ist, wobei $\lambda \in \mathbb{C}$. Es ist nun zwar nicht mehr möglich, über das Betragsquadrat von $\phi$ direkt eine Wahrscheinlichkeitsdichte zu definieren, mittels 
\begin{align}
    \mathbb{P}_t^{\phi}(B) = \frac{\int_{B} |\phi(r,t)|^2 \,d \mu (r)}{\sqrt{\int_{\mathbb{R}^3} |\phi(r,t)|^2 \,d \mu (r) \int_{\mathbb{R}^3} |\phi(r,t)|^2 \,d \mu (r)}}
\end{align}
sind wir aber in der Lage, dennoch die selben physikalischen Vorhersagen aus $\phi$ wie aus $\psi$ zu extrahieren. 

Damit gelangen wir zu einer verallgemeinerten Normierungsbedingung für quantenmechanische Zustände: 
\begin{align}
    \mathbb{P}_t^{\phi}(\mathbb{R}^3) = \frac{\| \phi \|^2}{\sqrt{\bigl\langle \phi , \phi \bigr\rangle \bigl\langle \phi , \phi \bigr\rangle}} = 1
\end{align}
Aus dieser Beobachtung schließen wir, dass alle Elemente $\phi_1 \neq 0$ und $\phi_2 \neq 0$ aus $L^2(\mu)$ physikalisch miteinander identifiziert werden sollten, sofern diese sich nur um eine Konstante $\lambda \in \mathbb{C}$ unterscheiden. Mathematisch führt das erneut auf eine Äquivalenzrelation: 
\begin{align}
    \phi_1 \sim \phi_2 \: :\Longleftrightarrow \: \exists \lambda \in \mathbb{C} \setminus \{0\} : \phi_1 = \lambda \phi_2
\end{align}
Definieren wir diese Äquivalenzrelation auf dem komplexen Hilbertraum $L^2(\mu)$, so erhalten wir den projektiven Hilbertraum $\mathbf{P}(L^2(\mu))$ (Beachte, dass es sich bei $\mathbf{P}(L^2(\mu))$ nicht mehr um einen Vektorraum handelt!). Die Punkte in $\mathbf{P}(L^2(\mu))$, auch Strahlen genannt, entsprechen dabei salopp den eindimensionalen komplexen Unterräumen des $L^2(\mu)$. Man bezeichnet den $\mathbf{P}(L^2(\mu))$ als Zustandsraum. 

Man beachte, dass man immer strikt zwischen Zustandsvektoren (Elemente aus dem $L^2(\mu)$) und Zuständen (Elemente aus $\mathbf{P}(L^2(\mu))$ unterscheiden sollte, denn um konkrete Berechnungen durchzuführen, wählt man immer nur einen beliebigen Repräsentanten aus denjenigen Strahl, der das System $\Sigma$ beschreibt. Man rechnet aber nicht mit den Zuständen selbst. 

Die Tatsache, dass nun der Zustandsraum der Quantenmechanik aus physikalischer Sicht als der projektive Raum $\mathbf{P}(L^2(\mu))$ verstanden werden kann, wird im weiteren Verlauf für die Konstruktion eines quantenmechanischen Logikmodells von großer Bedeutung sein. 


\subsection{Observablen, Superpositionen und Übergangswahrscheinlichkeiten in der Quantenmechanik}

Für Erklärungen der im Folgenden verwendeten Begriffe Selbstadjungiert, Spektrum, Eigenwerte und Eigenvektoren, siehe \cite{borthwick2020spectral} und \cite{fischer2011lernbuch}. \\ \\
Befassen wir uns nun noch kurz mit beobachtbaren Größen, sogenannte Observablen, des Systems $\Sigma$. Ein Postulat der Quantenmechanik ist es, dass jede Observable $\mathcal{V}$ eines quantenmechanischen Systems durch einen zu dieser Observable assoziierten nichtkonstanten selbsadjungierten Operator $\mathcal{O}_\mathcal{V}$ auf dem $L^2(\mu)$ beschrieben wird. Dabei wird gefordert, dass die möglichen Messwerte der Observable $\mathcal{V}$ den Elementen im Spektrums des Operators $\mathcal{O}_\mathcal{V}$ entsprechen. Im Falle selbstadjungierter Operatoren setzt sich das wegen der selbstadjungiertheit reellwertige Spektrum des Operators $\mathcal{O}_\mathcal{V}$ aus der disjunkten Vereinigung des Punktspektrums und des kontinuierlichen Spektrums zusammen. Das Punktspektrum ist dabei die Menge aller Eigenwerte, d.h. die Menge aller $\lambda$, für die der Operator
\begin{align}
    \lambda \mathbf{id} - \mathcal{O}_\mathcal{V} \label{Operator}
\end{align}
einen nicht trivialen Kern hat, d.h. es gilt
\begin{align}
    (\lambda \mathbf{id} - \mathcal{O}_\mathcal{V})(v) = 0, \qquad v \in L^2(\mu)
\end{align}
für $v \neq 0$. Das kontinuierliche Spektrum wiederum ist die Menge aller $\lambda$, für welche der Operator aus \eqref{Operator} injektiv, aber nicht surjektiv ist. Wie bereits erwähnt repräsentieren die Elemente des Spektrums die möglichen Messwerte der Observablen $\mathcal{V}$. Im Folgenden betrachten wir nur mögliche Messwerte, die einem Element im Punktspektrum des Operators $\mathcal{O}_\mathcal{V}$ entsprechen, denn für diese existiert ein zugehöriger Eigenvektor $v \in L^2(\mu)$, den man Eigenzustandsvektoren nennt. Ein Eigenzustandvektor des Observablenoperators $\mathcal{O}_\mathcal{V}$ beschreibt dabei das System, wenn bei Messung der Observable $\mathcal{V}$ der zum Eigenzustandvektor zugehörige Eigenwert $\lambda$ gemessen wird. 

Ein Beispiel eines solchen zu einer Observablen assoziierten Operators haben wir schon kennengelernt: So ist der Hamiltonoperator für zeitunabhängige Probleme wie dem des quantenmechanischen harmonischen Oszillators derjenige Operator, der der Observable Energie zugeordnet wird (siehe \eqref{eq:Z}), d.h. die Eigenwerte des Hamiltonoperators entsprechen den möglichen Energien des Systems.

Wie bereits erwähnt sichert die Eigenschaft der Selbstadjungiertheit, dass u.a. alle Eigenwerte des Operators reelwertig sind (wie wir es von möglichen Messwerten auch erwarten würden), denn sei $\mathcal{O_\mathcal{V}}$ ein selbstadjungierter Operator auf einer (dichten) Teilmenge des $L^2(\mu)$ und $\psi \in L^2(\mu)$ sei ein Eigenvektor zum Eigenwert $\lambda$, dann gilt
\begin{align}
    \lambda \bigl\langle \psi , \psi \bigr\rangle = \bigl\langle \psi , \lambda \psi \bigr\rangle = \bigl\langle \psi , \mathcal{O}_\mathcal{V} \psi \bigr\rangle = \bigl\langle \mathcal{O}_\mathcal{V} \psi , \psi \bigr\rangle = \bigl\langle \lambda \psi , \psi \bigr\rangle = \lambda^{*} \bigl\langle \psi , \psi \bigr\rangle,
\end{align}
woraus folgt, dass $\lambda = \lambda^{*}$ ist, was wiederum impliziert, dass der Imaginärteil von $\lambda$ verschwindet. 

Man kann sich nun die Frage stellen, warum wir fordern, dass die Operatoren, die wir mit Observablen assoziieren, nicht konstant sein dürfen. Der Grund ist der, dass wir in dem kommenden Abschnitten, in welchen wir ein Modell der Quantenlogik motivieren wollen, mittels dieser technischen Forderung leicht ein Modell der Negation einführen können. Im Prinzip werden sich die kommenden Überlegungen auf die folgende Beobachtung stützen: Sei $\mathcal{O}_\mathcal{V}$ wieder der (nichtkonstante) selbstadjungierte Operator zu einer Observablen $\mathcal{V}$ und $\psi_i$ ein Eigenvektor zu $\lambda_i$ und $\psi_j$ ein Eigenvektor zu $\lambda_j$ mit $\lambda_i \neq \lambda_j$, so gilt 
\begin{align}
    \lambda_j \bigl\langle \psi_i , \psi_j \bigr\rangle =& \: \bigl\langle \psi_i , \lambda_j \psi_j \bigr\rangle = \bigl\langle \psi_i , \mathcal{O}_\mathcal{V} \psi_j \bigr\rangle \nonumber\\ =& \: \bigl\langle \mathcal{O}_\mathcal{V} \psi_i , \psi_j \bigr\rangle = \bigl\langle \lambda_i \psi_i , \psi_j \bigr\rangle = \lambda_i \bigl\langle \psi_i , \psi_j \bigr\rangle
\end{align}
was impliziert, dass $\langle \psi_i , \psi_j \rangle = 0$ ist. D.h. die Eigenzustandsvektoren zu verschiedenen Eigenwerten stehen orthogonal aufeinander. 

Physikalisch argumentiert schließen wir dabei konstante selbstadjungierte Operatoren als Observablen deshalb aus, da für diese Operatoren jeder beliebige Zustand ein Eigenvektor zu einem Eigenwert $\lambda \in \mathbb{R}$ ist. Streng genommen müssen Messapparaturen, die eine solche Observable messen, welche von einem konstantem Operator repräsentiert wird, sich nicht einmal auf das eigentliche System $\Sigma$ beziehen, da die Messung so oder so immer den Messwert $\lambda$ liefert. Es erscheint daher sinnvoll derartige Operatoren aus physikalischer Sicht auszuschließen.

Nun wollen wir noch angeben, wie ein Zustand im Allgemeinen in der Quantenmechanik aussieht. Für alles weitere nehmen wir vereinfachend an, dass der Observablenoperator $\mathcal{O}_\mathcal{V}$ ein leeres kontinuierliches Spektrum und ein nichtleeres Punktspektrum besitzt. Sei nun weiter der Zustand des Systems der Einfachheit halber durch einen normierten Zustandvektor gegeben. Betrachtet man nun die Observable $\mathcal{V}$, so lässt sich im allgemeinen der normierte Zustandsvektor $\Psi$ vom System $\Sigma$ als Linearkombination der normierten Eigenzustandsvektoren $\psi_n$ des Operators $\mathcal{O}_\mathcal{V}$ schreiben, dass heißt es gilt 
\begin{align}
    \Psi = \sum_{n = 1}^{\infty} c_n \psi_n, \qquad c_n \in \mathbb{C} \: \: \forall n \in \mathbb{N}.
\end{align}
Obigen Ausdruck nennt man auch oft Superposition von Eigenzustandsvektoren.

Bei einer tatsächlichen Messung der Observable $\mathcal{V}$ des Systems $\Sigma$ sollen nun in unserer Vereinfachung die möglichen Messwerte gegeben sein durch die Eigenwerte des Operators $\mathcal{O}_\mathcal{V}$. Da bei einer Messung immer nur ein Messwert gemessen werden kann, muss durch die Messung des Eigenwertes $\lambda_n$ der obige Zustandsvektor $\Psi$ nach obiger Bemerkung in den zum Eigenwert $\lambda_n$ gehörigen normierten Eigenzustand $\psi_n$ übergehen. Diesen Vorgang nennt man oft Kollaps der Wellenfunktion. 

Zum Abschluss dieses Abschnittes lässt sich nun noch die Frage stellen, mit welcher Wahrscheinlichkeit bei einer Messung der Observablen $\mathcal{V}$ der Messwert $\lambda_n$ gemessen wird: Sei dazu das System unmittelbar vor der Messung gegeben durch den normierten Zustandsvektor $\Psi = \sum_{n = 1}^{\infty} c_n \psi_n$, wobei die $\psi_n$ wieder die normierten Eigenvektoren des selbstadjungierten Operators $\mathcal{O}_\mathcal{V}$ sind, der die Observable $\mathcal{V}$ beschreibt. Dann ergibt sich die Wahrscheinlichkeit dafür, bei Messung der Observable $\mathcal{V}$, den zum normierten Eigenvektor $\psi_n$ zugehörigen Eigenwert $\lambda_n$ zu messen, durch den Ausdruck 
\begin{align}
    |\bigr\langle \psi_n , \Psi \bigr\rangle|^2 = |c_n|^2 = c_n^*c_n.
\end{align}

Nach dieser wieder eher abstrakten Abhandlung wichtiger Begriffe der Quantenmechanik wollen wir uns noch ein kleines Beispiel ansehen, um ein besseres Gefühl für die eben eingeführten Begriffe zu bekommen:

\begin{example}
Als Beispiel bietet sich der eindimensionale quantenmechanische harmonische Oszillator (\cite{griffiths2018introduction}, Seite 40) an, da wir bereits in Beispiel 3.1 den klassischen harmonischen Oszillator behandelt hatten. Das Potential ist wie im klassischen Fall gegeben durch
\begin{align}
    V(x) = \frac{1}{2}Dx^2 = \frac{1}{2}m\omega_0 x^2
\end{align}
Damit ergibt sich der Hamiltonoperator zu 
\begin{align}
    \hat{H} = -\frac{\hbar^2}{2m}\Delta + \frac{1}{2}m\omega_0 x^2
\end{align}
Wie bereits erwähnt stellt der Hamiltonoperator in diesem Beispiel denjenigen Operator dar, der der Observable Energie zugeordnet wird, d.h. die Eigenwerte $E_n$ von $\hat{H}$ sind die möglichen Energiemesswerte.

Durch die Einführung der sogenannten Leiteroperatoren 
\begin{align}
    \hat{a} = \sqrt{\frac{m\omega_0}{2 \hbar}}\bigg(x + \frac{i}{m\omega_0}\hat{p}\bigg), \\ \hat{a}^{\dag} = \sqrt{\frac{m\omega_0}{2 \hbar}}\bigg(x - \frac{i}{m\omega_0}\hat{p}\bigg),
\end{align}
wobei $\hat{p} = - i \hbar \frac{\partial}{\partial x}$ der sogenannte Impulsoperator ist, lässt sich der Hamiltonoperator umschreiben zu 
\begin{align}
    \hat{H} = \hbar \omega_0 \bigg(\hat{a}^\dag \hat{a} + \frac{1}{2}\bigg).
\end{align}
Der Operator 
\begin{align}
    \hat{a}^\dag \hat{a} =: \hat{N}
\end{align}
wird manchmal auch Besetzungszahloperator genannt. Es lässt sich zeigen, dass die Eigenwerte von $\hat{N}$ der Menge $\mathbb{N}_0 = \mathbb{N}\cup\{0\}$ entsprechen. Die möglichen Energien des Systems sind damit gegeben durch 
\begin{align}
    E_n = \hbar \omega_0 \bigg(n + \frac{1}{2}\bigg), \qquad n \in \mathbb{N}_0.
\end{align}
Es lässt sich zeigen, dass der normierte Eigenvektor $\tilde{\psi}_n$ zum Eigenwert $E_n$ gegeben ist durch
\begin{align}
    \tilde{\psi}_n(x) = \bigg(\frac{m \omega_0}{\pi \hbar}\bigg)^{\frac{1}{4}}\frac{1}{\sqrt{2^n n!}} \bigg( \sqrt{\frac{m \omega_0}{\hbar}}x - \sqrt{\frac{\hbar}{m \omega_0}}\frac{\partial}{\partial x} \bigg)^n \exp{\bigg(\frac{-m \omega_0}{2 \hbar}x^2\bigg)}. \label{eq:Eig}
\end{align}
Ein allgemeiner Zustand des Systems ist dann z.B. gegeben durch den normierten Zustand 
\begin{align}
    \tilde{\Psi} = \sum\limits_{n = 0}^{\infty} c_n \tilde{\psi}_n
\end{align}
Befindet sich beispielsweise das System zu $75\%$ im Grundzustand $\tilde{\psi}_0$ und zu $25\%$ im angeregten Zustand $\tilde{\psi}_1$, so gilt
\begin{align}
    \tilde{\Psi} = \sqrt{\frac{3}{4}}\tilde{\psi}_0 + \sqrt{\frac{1}{4}}\tilde{\psi}_1,
\end{align}
denn es gilt 
\begin{align}
    |\langle \tilde{\psi}_0 , \tilde{\Psi} \rangle|^2 = \frac{3}{4} \equiv 75\%, \\
    |\langle \tilde{\psi}_1 , \tilde{\Psi} \rangle|^2 = \frac{1}{4} \equiv 25\%.
\end{align}
\end{example}


\subsection{Quantenlogik}

Sei $\Sigma$ wieder ein quantenmechanisches System und $\mathcal{V}$ eine beliebige aber fest gewählte Observable des Systems. Ähnlich wie im klassischen Fall betrachten wir nun zusätzlich eine Menge an Ja/Nein-Experimenten bezüglich der Observable $\mathcal{V}$ an das System $\Sigma$, dass heißt, eine Menge an Aussagen über den Zustand von $\Sigma$, welche nach einer Messung von $\mathcal{V}$ entweder mit \textit{Ja} oder \textit{Nein} beantwortet werden können. Wir betrachten dabei erneut nur Observablen, deren mögliche Messwerte im Punktspektrum des zugehörigen Observablenoperators $\mathcal{O}_\mathcal{V}$ liegen. 

Im vorherigen Abschnitt haben wir gesehen, dass der Zustandsraum eines quantenmechanischen Systems als projektiver Hilbertraum $\mathbf{P}(L^2(\mu))$ verstanden werden kann. Mittels dieser Vorarbeit lässt sich nun leicht ein quantenmechanischen Logikmodell für Zustände des Systems motivieren. Da die Ja/Nein-Experimente Aussagen über den Zustand des Systems $\Sigma$ bzgl. $\mathcal{V}$ sein sollen, können wir Aussagen über einen konkreten Zustand des Systems als eindimensionale komplexen Unterräume des $L^2(\mu)$ auffassen, die von denjenigen Eigenzustandsvektor aufgespannt werden, die in der Äquivalenzklasse liegen, die dem Zustand aus der logischen Aussage entsprechen.

Alle weiteren logischen Aussagen über $\mathcal{V}$ bzgl. des Systems $\Sigma$, die sich nicht auf einen konkreten Zustand, d.h. einem bestimmten Strahl aus $\mathbf{P}(L^2(\mu))$, beziehen, sagen stattdessen etwas über die Zugehörigkeit des das System $\Sigma$ beschreibenden Zustandsvektors zu einem abgeschlossenen komplexen Untervektorraum des $L^2(\mu)$ aus. Daher sagt man, dass die Ja/Nein-Experimente an ein quantenmechanischen System bzgl. der Observable $\mathcal{V}$ assoziiert werden mit derjenigen Menge aller abgeschlossener Unterräume des $L^2(\mu)$, welche aufgespannt werden von den Eigenzustandsvektoren des Operators $\mathcal{O}_\mathcal{V}$. 

Allgemeiner bezeichne die Menge $\mathcal{L}$ nun alle abgeschlossenen Unterräume des komplexen Hilbertraumes $L^2(\mu)$. Man ordnet nun jeder beliebigen Aussage an das System, welche sich auf Messwerte beziehen, die im Punktspektrum des zugehörigen Observablenoperators liegen, einen abgeschlossenen Unterraum des komplexen Hilbertraumes zu und umgekehrt. Es ist dabei zu beachten, dass nicht alle solche Aussagen an das System zu einem Ja/Nein-Experiement der Observable $\mathcal{V}$ korrespondieren (siehe kommendes Beispiel 4.2). Es macht aber Sinn, auch solche Aussagen als einen Unterraum des Hilbertrams zu modellieren, da sich jede solche Aussage an das System als Aussage über den Zustand des Systems verstehen lässt, welcher wiederrum über Zustandsvektoren beschreiben wird. Außerdem lässt sich jeder beliebige Vektor aus dem $L^2(\mu)$ als Eigenvektor eines bestimmten selbstadjungierten Operators (z.B. eines speziellen Projektionsoperators) verstehen. Da selbstadjungierte Operatoren wiederum Observablen repräsentieren, lässt sich somit jede Aussage, die wir über einen abgeschlossenen komplexen Unterraum des $L^2(\mu)$ modellieren, als Ja/Nein-Experiment einer speziellen Observable $\mathcal{W}$ interpretieren.

Wir bezeichnen die Aussagen an das System $\Sigma$ (und manchmal auch die abgeschlossenen Unterräume) als Propositionen. Die Forderung nach der Abgeschlossenheit der den Aussagen zugeordneten Untervektorräume hat unter anderem technische Gründe, die mit der Einführung einer Orthokomplementierung auf der Menge $\mathcal{L}$ als Modell der logischen Negation zu tun haben und im nächsten Abschnitt besprochen werden.

Es ist hierbei aber zu beachten, dass wir tatsächlich in unser Logikmodell nur diejenigen Propositionen aufnehmen, die sich auf diejenigen Messwerte einer Observablen $\mathcal{V}$ beziehen, welche durch Elemente im Punktspektrum dargestellt werden. Der Grund ist der, dass nur in diesem Fall die Proposition bzgl. $\mathcal{V}$ mit einem komplexen abgeschlossenen Unterraum assoziiert werden können, da nur in diesem Fall zu den möglichen Messwerten zugehörige Eigenzustandsvektoren existieren, aus denen man die komplexen abgeschlossenen Unterräume konstruieren kann, die man dann mit der Proposition assoziiert. D.h. wir sind primär nur an Propositionen interessiet, die man mit abgeschlossenen Unterräumen des $L^2(\mu)$ assoziieren kann. Begründet wird das damit, dass wir am Ende dieser Arbeit ein Theorem beweisen wollen, welches sich auf die Zustandsvektorräume eines zusammengesetzten Systems und seiner Teilsysteme bezieht. Demzufolge interessieren uns nur Propositionensysteme, welche Informationen über die beteiligten Zustandsvektorräume enthalten.\\

Betrachten wir wieder ein kleines Beispiel, um uns vertraut mit den eben eingeführten Konzept zu machen:

\begin{example}
Betrachten wir wieder den eindimensionalen quantenmechanischen harmonischen Oszillator aus Beispiel 4.1. Wie wir dort gesehen haben waren die Eigenzustände $\tilde{\psi}$ des Systems gegeben durch \eqref{eq:Eig}. Die Propositionen, die sich auf einen konkreten Eigenzustand des Systems beziehen, wie z.B. \textit{„Das System befindet sich im Zustand $\tilde{\psi_i}$“}, sind dann gegeben durch die komplexen Unterräume
\begin{align}
    span_\mathbb{C}(\{\tilde{\psi_i}\}),
\end{align}
während z.B. Propositionen wie \textit{„Das System befindet sich im Eigenzustand $\tilde{\psi_i}$ oder im Eigenzustand $\tilde{\psi_i}$“} durch den komplexen Unterraum
\begin{align}
    span_\mathbb{C}(\{\tilde{\psi_i, \psi_j}\})
\end{align}
gegeben sind. Natürlich gibt es auch Propositionen, die sich nicht direkt auf konkrete Eigenzustände des Hamiltonoperators beziehen. So ist die Proposition \textit{„Das System befindet sich zu $75\%$ im Eigenzustand $\tilde{\psi}_0$ und zu $25\%$ im Eigenzustand $\tilde{\psi}_1$“} gegeben durch 
\begin{align}
    span_\mathbb{C}(\{\tilde{\Psi} = \sqrt{3/4}\tilde{\psi}_0 + \sqrt{1/4}\tilde{\psi}_1\}).
\end{align}
Es ist dabei aber zu beachten, dass letztere Proposition zu keinem Ja/Nein-Experiment bezüglich der Observable Energie korrespondiert, da nach einer Messung des Systems nicht sofort entschieden werden kann, ob die Aussage mit \textit{Ja} oder \textit{Nein} beantwortet werden kann.
\end{example}

Um kurz noch den Unterschied zwischen Propositionen, die als abgeschlossene komplexe Unterräume verstanden werden können, und Propositionen, die nicht als solche aufgefasst werden können, etwas klarer zu gestalten, können wir erneut den Hamiltonoperator des eindimensionalen quantenmechanischen harmonischen Oszillators und den Ortoperator $\hat{x}$ betrachten: In Beispiel 4.2 haben wir gesehen, wie wir Propositionen an die Observable Energie, die beschrieben wird über den Hamiltonoperator $\hat{H}$, assoziieren können mit abgeschlossenen komplexen Unterräumen in $L^2(\mu)$, sofern die Energiemesswerte in der Proposition sich auf Messwerte beziehen, die im Punktspektrum von $\hat{H}$ liegen. Betrachten wir nun den Ortoperator $\hat{x}$. Im eindimensionalen ist dieser definiert durch 
\begin{align}
    \hat{x} \psi(x) = x \cdot \psi(x), \qquad \psi \in L^2(\mu).
\end{align}
Wegen der Tatsache, dass der $L^2(\mu)$ ein Quotientenvektorraum ist, dessen Elemente Äquivalenzklassen darstellen, die all diejenigen Funktionen enthalten, die sich nur auf einer Nullmenge voneinander unterscheiden, folgt, dass der Ortoperator $\hat{x}$ keine Eigenvektoren und damit keine Eigenwerte besitzt. D.h. das Punktspektrum von $\hat{x}$ ist leer! Daraus folgt, dass alle möglichen Ortsmesswerte im kontinuierlichen Spektrum von $\hat{x}$ liegen. Es existiert also in $L^2(\mu)$ kein Zustandsvektor, der das System beschreiben würde, wenn bei Messung des Ortes mit $100\%$ Wahrscheinlichkeit der Ort x gemessen wird. D.h., dass nach obiger Philosophie kein abgeschlossener Unterraum in $L^2(\mu)$ existiert, der mit der Proposition \textit{„Das Teilchen befindet sich am Ort x“} assoziiert werden kann. Folglich wird eine derartige Proposition nicht durch die Menge $\mathcal{L}$ erfasst, weshalb wir derartige Propositionen aus unserer Betrachtung ausschließen werden. Immerhin enthält eine solche Proposition keine Informationen über den Zustandsraum $L^2(\mu)$ im obigem Sinne.

Für den ein oder anderen mag nun $\mathcal{L}$ als spezielles Propositionensystem an das System $\Sigma$ unbefriedigend erscheinen, da viele Propositionen, die z.B. Aussagen über den Impuls oder den Ort des Teilchens in $\Sigma$ machen, nicht von diesem erfasst werden. Im Abschnitt 6 wird daher eine Idee diskutiert, Propositionen in dieses Bild mit aufzunehmen, welche sich nicht auf Messwerte beziehen, die im kontinuierlichen Spektrum liegen.

Immer wenn wir im Folgenden von logischen Aussagen oder der Menge aller Propositionen sprechen, beziehen wir uns dabei immer nur auf die Menge aller logischen Aussagen an das System $\Sigma$, welche mit einem abgeschlossenen komplexen Unterraum im obigem Sinne assoziiert werden können!\\

Genau wie im klassischen Fall müssen wir nun eine algebraische Struktur auf der Menge $\mathcal{L}$ definieren, mit der wir die logischen Operationen auf der Menge der logischen Aussagen auf die Menge $\mathcal{L}$ übertragen können. Eine naheliegende Möglichkeit ist die Folgende: Seien $a$ und $b$ logische Aussagen an $\Sigma$ und $p$ derjenige abgeschlossenen Unterraum, der der Aussage $a$ und $q$ derjenige abgeschlossene Unterraum, der der Aussage $b$ zugeordnet wird. Dann ordnen wir der Aussage $a \wedge b$ den komplexen Unterraum $p \cap q$ zu. Man beachte, dass der Schnitt zweier abgeschlossener Unterräume wieder ein abgeschlossener Unterraum ist. Der Aussage $a \vee b$ ordnen wir hingegen den topologischen Abschluss vom Unterraum $span_\mathbb{C}(p \cup q)$ zu, da die lineare Hülle zweier abgeschlossener Unterräume zwar wieder ein Unterraum ist, dieser aber nicht notwendigerweise abgeschlossen sein muss. Zu beachten ist, dass im Gegensatz zum obigen klassischen Logiksystem, die Aussage $a \vee b$ für ein quantenmechanisches System $\Sigma$ in einem gewissen Sinne wahr sein kann, auch wenn $a$ und $b$ jeweils nicht auf $\Sigma$ zutreffen. Etwas schwieriger ist die Frage nach dem abgeschlossenen Unterraum, der der Aussage $\neg a$ zugeordnet werden soll. Mittels der Beobachtungen im vorigen Abschnitt liegt es jedoch nahe, der Aussage $\neg a$ den Unterraum 
\begin{align}
    p^{\perp} := \{ \phi \in L^2(\mu) : \langle \psi , \phi \rangle = 0 \: \textrm{mit} \: \psi \in p \}
\end{align}
zuzuordnen, den wir orthogonales Komplement der Menge $p$ nennen. Der Grund dafür lässt sich mittels obiger Beobachtung folgendermaßen motivieren: Angenommen wir betrachten eine Observable $\mathcal{V}$ von $\Sigma$, die über den selbstadjungierten Operator $\mathcal{O}_\mathcal{V}$ beschreiben wird und sei $\lambda_j$ der $j$-te Eigenwert von $\mathcal{O}_\mathcal{V}$ zum $j$-ten Eigenvektor $\psi_j$. Gehen wir davon aus, dass bei der Messung der Observable $\mathcal{V}$ zu $100 \%$ der Eigenwert $\lambda_j$ gemessen wird, so liegt das System bereits vor der Messung in einem Zustand vor, der vollständig durch den Zustandsvektor $\psi_j$ beschrieben wird. Gehen wir nun davon aus, dass wir das selbe System so präparieren, dass bei einer erneuten zweiten Messung von $\mathcal{V}$ zu $0 \%$ der Messwert $\lambda_j$ gemessen wird, so muss das System unmittelbar vor der Messung in einer Superposition der Form $\Psi = \sum_{i = 1, i \neq j}^{\infty} c_i \psi_i$, mit $c_i \in \mathbb{C}$, vorgelegen haben. Aus dem vorherigen Abschnitt wissen wir, dass $\langle \psi_i , \psi_j \rangle = 0$ für alle $i \in \mathbb{N} \setminus \{j\}$ und damit auch $\langle \psi_i , \Psi \rangle = 0$ ist. In diesem Sinne liegt der Zustandsvektor bei dieser zweiten Messung, die als 'negierte erste Messung' verstanden werden kann, im orthogonalen Komplement von $p = span_\mathbb{C}(\{\psi_j\})$. Mittels der Stetigkeit des Skalarproduktes von $L^2(\mu)$ lässt sich nun noch leicht nachweisen, dass $p^{\perp}$ ein abgeschlossener Unterraum ist, denn sei $(\phi_i)_{i \in \mathbb{N}}$ eine Cauchy-Folge in $p^{\perp}$, die für $i \rightarrow \infty$ gegen $\phi \in L^2(\mu)$ konvergiert und $\psi \in p$, so gilt 
\begin{align}
    \bigl\langle \psi , \phi \bigr\rangle = \bigl\langle \psi , lim_{i \rightarrow \infty} \phi_i \bigr\rangle = lim_{i \rightarrow \infty} \bigl\langle \psi , \phi_i \bigr\rangle
\end{align}
Da $\langle \psi , \phi_i \rangle = 0$ für alle $i \in \mathbb{N}$ ist, folgt, dass $\phi \in p^{\perp}$. Damit entspricht $p^{\perp}$, genau wie wir es auch erwarten würden, nach unserer bisherigen Philosophie wieder einem Ja/Nein-Experiment der Observable $\mathcal{O}_{\mathcal{V}}$ an das System $\Sigma$ und damit insbesondere einem Element in $\mathcal{L}$. 

Zusammengefasst wollen wir damit das Quadrupel $(\mathcal{L}, \sqcup, \sqcap, \cdot^{\perp})$ mit $p \sqcap q = p \cap q$ und $p \sqcup q = \overline{span_\mathbb{C}(p \cup q)}$ mit $p,q \in \mathcal{L}$ als Quantenlogik bezeichnen. 

Falls die Menge $\mathcal{L}$ die Menge aller abgeschlossenen Unterräume eines allgemeinen Hilbertraumes $\mathcal{H}$ meint, nennt man $(\mathcal{L}, \sqcup, \sqcap, \cdot^{\perp})$ manchmal auch einfach Hilbert-Verband (des Hilbertraumes $\mathcal{H}$).  Genauso wie im klassischen Fall ist $(\mathcal{L}, \sqcup, \sqcap, \cdot^{\perp})$ ein Verband, der sich aber, wie wir im nächsten Abschnitt sehen werden, hinsichtlich seiner Eigenschaften von $(\mathcal{P}(\Omega), \cup, \cap, \cdot^C)$ unterscheiden wird. 


\subsection{Verbandsstruktur der Quantenlogik}

Genau wie im klassischen Fall ist ohne Weiteres sofort ersichtlich, dass es sich bei $(\mathcal{L}, \sqcup, \sqcap)$ mit $p \sqcap q = p \cap q$ und $p \sqcup q = \overline{span_\mathbb{C}(p \cup q)}$ um einen mathematischen Verband handelt. Denn seien $p,q \in \mathcal{L}$ Propositionen, so gilt
\begin{align}
    p \sqcup (p  \sqcap q) = p,
\end{align}
da $p$ abgeschlossen und $p  \sqcap q \subseteq p$ ist. Ferner gilt 
\begin{align}
    p \sqcap (p  \sqcup q) = p,
\end{align}
da $p \subseteq p  \sqcup q$ ist. Die Kommutativität und Assoziativität der Operationen $\sqcap$ und $\sqcup$ ist klar.

Genauso wie für das klassische System ergibt sich, dass der Verband $(\mathcal{L}, \sqcup, \sqcap)$ ein vollständiger atomarer Verband ist. Dabei ist die Relation $\leq$ wieder gegeben durch die Teilmengenrelation $\subseteq$ und die Atome sind demzufolge die eindimensionalen komplexen Unterräume des Hilbertraumes. 

Man sieht nun aber relativ schnell, dass es sich bei $(\mathcal{L}, \sqcup, \sqcap)$, im Gegensatz zu $(\mathcal{P}(\Omega), \cup, \cap)$, nicht um einen distributiven Verband handeln kann! Seien nämlich $p_1 = span_\mathbb{C}(\{\psi_1\})$ und $p_2 = span_\mathbb{C}(\{\psi_2\})$ zwei eindimensionale komplexe Unterräume eines Hilbertraumes $\mathcal{H}$, die sich nur im Nullvektor schneiden, und sei $p_3 = span_\mathbb{C}(\{\psi_2, \psi_3\}) \subseteq \mathcal{H}$ ein zweidimensionaler komplexer Unterraum. Per Konstruktion gilt $p_2 \subseteq p_3$. Es gilt \begin{align}
    p_3 \sqcup (p_1 \sqcap p_2) = p_3 \sqcup \{0\} = \textrm{span}_\mathbb{C}(\{\psi_2, \psi_3\}).
\end{align}
Dabei ist zu beachten, dass endlichdimensionale Vektorräume immer abgeschlossen sind, weshalb der topologische Abschluss im letzten Ausdruck nicht mehr explizit auftaucht. Andererseits gilt nun aber auch 
\begin{align}
    (p_3 \sqcup p_1) \sqcap (p_1 \sqcup p_2) =& \: \textrm{span}_\mathbb{C}(\{\psi_1, \psi_2, \psi_3\}) \sqcap \textrm{span}_\mathbb{C}(\{\psi_1, \psi_2\}) \nonumber\\ =& \: \textrm{span}_\mathbb{C}(\{\psi_1, \psi_2\}).
\end{align}
Daraus folgt, dass im Allgemeinen gilt, dass $p \sqcup (q \sqcap r) \neq (p \sqcup q) \sqcap (p \sqcup r)$ für $p,q,r \in \mathcal{L}$ ist. Dadurch sehen wir, dass es sich bei Quantenlogiken im Allgemeinen nicht um distributive Verbände handelt.

Man kann die Tatsache, dass ein quantenlogischer Verband im Allgemeinen nicht distributiv ist, auch physikalisch auf Ebene der Aussagen an das physikalische System $\Sigma$, statt über die Elemente des Verbandes $\mathcal{L}$, erklären \cite{floridi1998routledge}: Wir betrachten dazu ein eindimensionales physikalisches System $\Sigma$, welches ein Teilchen beschreibt, das sich am Ort $x \in [\alpha,\beta]$ aufhält. Wir betrachten nun drei Aussagen an das System $\Sigma$:
\begin{itemize}
    \item [\textit{i)}] Die Aussage $s$ sei: \textit{„Das Teilchen besitzt die Geschwindigkeit $v_x \in [v_1,v_2]$}. D.h., wenn man die Geschwindigkeit des Teilchens messen würde, so würde man eine Zahl aus dem Intervall $[v_1,v_2]$ erhalten, sofern die Aussage $s$ wahr ist.
    \item [\textit{ii)}] Die Aussage $u$ sei: \textit{„Das Teilchen befindet sich im Intervall $[\alpha,\gamma]$}. D.h., wenn man den Ort des Teilchens messen würde, so würde man eine Zahl aus dem Intervall $[\alpha,\gamma]$ erhalten, sofern die Aussage $u$ wahr ist. $\gamma$ bezeichne dabei die Mitte des Intervalls $[\alpha,\beta]$.
    \item [\textit{iii)}] Die Aussage $m$ sei: \textit{„Das Teilchen befindet sich im Intervall $[\gamma,\beta]$}. D.h., wenn man den Ort des Teilchens messen würde, so würde man eine Zahl aus dem Intervall $[\gamma,\beta]$ erhalten, sofern die Aussage $m$ wahr ist.
\end{itemize}
Wir betrachten die Aussage $a := s \wedge (u \vee m)$ und $b := (s \wedge u) \vee (s \wedge m)$. Würde die Distributivität gelten, wären beide Aussagen wahrheitstechnisch äquivalent. Es ist aber ein bekannter Fakt, dass man in der Quantenmechanik den Ort $x$ und die Geschwindikeit $v_x$ bzw. den Impuls $p_x = mv_x$ nicht gleichzeitig beliebig messen kann. Man bezeichnet dieses Phänomen als Orts-Impuls-Unschärfe und die folgt aus der Nichtkommutativität des Ortsoperators $\hat{x}$ und des Impulsoperators $\hat{p}$ \cite{robertson1929uncertainty}, d.h. es gilt 
\begin{align}
    [\hat{p},\hat{x}] := \hat{p}\hat{x} - \hat{x}\hat{p} \neq 0.
\end{align}
Daraus folgt, dass, wenn die Aussage $a$ wahr ist im Allgemeinen Aussage $b$ wegen der Orts-Impuls-Unschärfe nicht wahr ist und damit das Distributivgesetz im Allgemeinen nicht gilt. 

Wie im obigen Logiksystem der klassischen Mechanik, können wir uns nun um die Einführung einer Orthokomplementation $'$ auf $\mathcal{L}$ bemühen. Ein naheliegender Kandidat ist natürlich $\cdot^{\perp}$, da wir $\cdot^{\perp}$ ja bereits als Modell der Negation eingeführt haben. Formal müssen wir natürlich überprüfen, ob es sich bei $\cdot^{\perp}$ tatsächlich um eine Orthokomplementation handelt. Dazu stellen wir zuerst fest, dass der sogenannte Projektionssatz \cite{meise2013einfuhrung} sicherstellt, dass in einem Hilbertraum $\mathcal{H}$ für einen abgeschlossenen Unterraum $U \subseteq \mathcal{H}$ ein eindeutiger Orthogonalprojektor auf den Unterraum $U$ existiert. Bei einem Orthogonalprojektor handelt es sich dabei um einen Operator $P_U : \mathcal{H} \rightarrow \mathcal{H}$ mit den Eigenschaften, dass 
\begin{itemize}
    \item [\textit{i)}] das Bild des Operators $U$ ist, d.h. $\textrm{im}(P_U) = U$,
    \item [\textit{ii)}] der Kern des Operators dem orthogonalem Komplement von $U$ entspricht, d.h. $\textrm{ker}(P_U) = U^{\perp}$.
\end{itemize}
Aus diesem Resultat folgt, dass es zu jedem Vektor $x \in \mathcal{H}$ eindeutige Vektoren $u_1 \in U$ und $u_2 \in U^{\perp}$ gibt, so dass $x = u_1 + u_2$ ist. In diesem Sinne bilden $U$ und $U^{\perp}$ eine orthogonale Zerlegung des Hilbertraumes $\mathcal{H}$ und es gilt damit auch, dass $\mathcal{H} = U \sqcup U^{\perp}$ ist. Unmittelbar aus der Definition des orthogonalen Komplements ergibt sich, dass $U \sqcap U^{\perp} = \{0\}$ ist und das für $V \subseteq U$ folgt, dass $U^{\perp} \subseteq V^{\perp}$. 

Zum Schluss stellen wir fest, dass für $U \subseteq \mathcal{H}$ gilt, dass $U = U^{\perp \perp}$ ist, sofern $U$ ein abgeschlossener Unterraum ist. Die Richtung $U \subseteq U^{\perp \perp}$ ist trivial, da per Definition $U^{\perp \perp}$ abgeschlossen ist und jedes Element aus $U$ auch in $U^{\perp \perp}$ enthalten sein muss. Für die Richtung $U^{\perp \perp}$ zerlegen wir den Vektor $x \in U^{\perp \perp}$ wie oben, d.h. $x = u_1 + u_2$ mit $u_1 \in U$ und $u_2 \in U^{\perp}$. Es folgt $\|u_2\|^2 = \langle u_2 , u_2 \rangle = \langle x - u_1 , u_2 \rangle = \langle x , u_2 \rangle - \langle u_1, u_2 \rangle = 0$, woraus folgt, dass $x = u_1 \in U$. Damit ist verfiziert, dass es sich bei $^{\perp}$ auf $\mathcal{L}$ um eine Orthokomplementation handelt.

Nun wollen wir der Vollständigkeitshalber noch wichtige Eigenschaften von $(\mathcal{L}, \sqcup, \sqcap, \cdot^{\perp})$ aufzählen, da in der Literatur eine Quantenlogik oft abstrakt als Verband eingeführt wird, der ganz speziellen Eigenschaften genügt. 

Eine erste Beobachtung, die sich aus den obigen Ergebnissen ergibt, ist, das trotz der Tatsache, dass das Distributivgesetz in der Quantenlogik nicht gilt, zumindest eine Abschwächung des Distributivgesetztes, das sogenannte modulare Gesetz, in $(\mathcal{L}, \sqcup, \sqcap, \cdot^{\perp})$ 
\begin{align}
    p \subseteq q : \: \: \: q = p \sqcup (q \sqcap p^{\perp}) \qquad p,q \in \mathcal{L}
\end{align}
gilt. Man nennt Verbände, die orthokomplementiert sind und für die das modulare Gesetz gilt, orthomodulare Verbände.  

Desweiteren lässt sich zeigen, dass es sich bei $(\mathcal{L}, \sqcup, \sqcap, \cdot^{\perp})$ um einen schwach modularen und irreduziblen Verband handelt, der dem Überdeckungsgesetz genügt (Beweis siehe \cite{piron1976foundations} und \cite{piron1964axiomatique}). Dabei heißt schwach modular, dass Unterverbände von $\mathcal{L}$ generiert durch $\{a,a^{\perp},b,b^{\perp}\}$ mit $a \subseteq b$ distributiv sind. Mit dem Überdeckungsgesetz ist folgende Eigenschaft gemeint: Sei $p \in \mathcal{L}$ ein Atom und $a,b \in \mathcal{L}$ so gewählt, dass $a \sqcap p = \{0\}$ und $a \subseteq b \subseteq a \sqcup p$ ist, so folgt $b = a$ oder $b = a \sqcup p$. Um die Eigenschaft der Irreduzibilität zu verstehen benötigen wir noch die folgende Definition:

\begin{definition}
Seien $a,b \in \mathcal{L}$. $a$ und $b$ heißen kompatibel miteinander, falls der Unterverband, der von $\{a,a^{\perp},b,b^{\perp}\}$ generiert wird, ein distributiver ist. Wir schreiben in diesem Fall $a \leftrightarrow b$.
\end{definition}

Der Verband $(\mathcal{L}, \sqcup, \sqcap, \cdot^{\perp})$ heißt nun in dem Sinne irreduzibel, das nur die Elemente $\{0\} = \mathbf{0}$ und $L^2(\mu) = \mathbf{1}$ kompatibel mit allen anderen Elementen aus $\mathcal{L}$ sind.  

Zusammengefasst ist damit eine Quantenlogik $(\mathcal{L}, \sqcup, \sqcap, \cdot^{\perp})$ gegeben durch einen vollständigen, orthomodularen, irreduziblen, schwach modularen, atomaren Verband, für den das Überdeckungsgesetz gilt. Wir nennen $(\mathcal{L}, \sqcup, \sqcap, \cdot^{\perp})$ oder kurz $\mathcal{L}$ (quantenmechanisches) Propositionensystem.Auf diese Art werden für gewöhnlich in der Literatur Quantenlogiken abstrakt definiert.\\

Am Ende dieses Abschnittes wollen wir nun noch kurz folgendes Lemma beweisen, mit deren Hilfe wir später leichter Propositionen manipulieren können:

\begin{lemma}
Seien $p_i \in \mathcal{L}$ für alle $i$ aus einer beliebigen Indexmenge $I$. Dann gilt, dass
\begin{itemize}
    \item [\textit{i)}] $( \cap_{i \in I} p^{\perp}_i ) = ( \cup_{i \in I} p_i )^{\perp}$
    \item [\textit{ii)}] $\overline{\textrm{span}_\mathbb{C}( \cup_{i \in I} p^{\perp}_i )} = ( \sqcup_{i \in I} p^{\perp}_i) = ( \cap_{i \in I} p_i )^{\perp}$
\end{itemize}
\end{lemma}

\begin{proof}
Um die beiden Mengengleichheiten zu zeigen, müssen wir jeweils zeigen, dass die eine Menge in der anderen enthalten ist und umgekehrt.
\begin{itemize}
    \item [\textit{i)}] Sei $\phi \in ( \cup_{i \in I} p_i )^{\perp}$. Dann gilt
    \begin{align}
    \phi \in ( \cup_{i \in I} p_i)^{\perp} \Longleftrightarrow&  \: \: \phi \perp  ( \cup_{i \in I} p_i ) \nonumber\\ \Longleftrightarrow&  \: \: \phi \perp p_i \: \: \: \forall i \in I \nonumber\\ \Longleftrightarrow& \: \: \phi \in p^{\perp}_i \: \: \: \forall i \in I \nonumber\\ \Longleftrightarrow& \: \: \phi \in ( \cap_{i \in I} p^{\perp}_i )
    \end{align}
    \item [\textit{ii)}] Sei $\phi \in ( \cup_{i \in I} p^{\perp}_i )$. Dann gilt
    \begin{align}
    \phi \in ( \cup_{i \in I} p^{\perp}_i ) \Longrightarrow& \: \: \exists j \in I \: \: \textrm{mit} \: \: \phi \in p^{\perp}_j \nonumber\\ \Longrightarrow& \: \: \phi \perp p_j \nonumber\\ \Longrightarrow& \: \: \phi \perp ( \cap_{i \in I} p_i ) \nonumber\\ \Longrightarrow& \: \: \phi \in ( \cap_{i \in I} p_i )^{\perp}
    \end{align}
    Daraus folgt $( \cup_{i \in I} p^{\perp}_i ) \subseteq ( \cap_{i \in I} p_i )^{\perp}$ und da $( \cap_{i \in I} p_i )^{\perp}$ ein Vektorraum ist folgt $\textrm{span}_\mathbb{C}( \cup_{i \in I} p^{\perp}_i) \subseteq ( \cap_{i \in I} p_i)^{\perp}$. Da weiter $( \cap_{i \in I} p_i)^{\perp}$ ein abgeschlossener Unterraum ist, folgt schließlich $\overline{\textrm{span}_\mathbb{C}( \cup_{i \in I} p^{\perp}_i )} \subseteq ( \cap_{i \in I} p_i )^{\perp}$. Es bleibt die Umkehrung zu zeigen:
    \begin{align}
    \phi \in ( \cap_{i \in I} p_i )^{\perp} \Longrightarrow& \: \: \phi \in ( \cap_{i \in I} p_i) \nonumber\\ \Longrightarrow& \: \: \exists j \in I \: \: \textrm{mit} \: \: \phi \perp p_j \nonumber\\ \Longrightarrow& \: \: \phi \in p^{\perp}_j \nonumber\\ \Longrightarrow& \: \: \phi \in ( \cup_{i \in I} p^{\perp}_i) \nonumber\\ \Longrightarrow& \: \: \phi \in \textrm{span}_\mathbb{C}( \cup_{i \in I} p^{\perp}_i) \nonumber\\ \Longrightarrow& \: \: \phi \in \overline{\textrm{span}_\mathbb{C}( \cup_{i \in I} p^{\perp}_i)}
    \end{align}
    Daraus folgt $( \cap_{i \in I} p_i )^{\perp} \subseteq \overline{\textrm{span}_\mathbb{C}( \cup_{i \in I} p^{\perp}_i)}$ und damit obige Gleichheit.
\end{itemize}
\end{proof}


\subsection{Wahrheitswerte in der Quantenlogik}

Wir betrachten nun ein quantenmechanisches System $\Sigma$. Dieses präparieren wir so, dass das System durch den normierten Zustandsvektor $\psi \in L^2(\mu)$ beschrieben wird. Wir interessieren uns nun für Aussagen der Form $q = span_\mathbb{C}(\{ \phi_i \: : \: i \in I\})$, wobei die $\phi_i$ Eigenzustände einer beliebigen aber festen Observable $\mathcal{V}$ sind und $I$ eine beliebige Indexmenge ist. Die Aussagen die wir betrachten stellen also Ja/Nein-Experiemente dar, könnten nach einer Messung von $\mathcal{V}$ also eindeutig mit \textit{Ja} oder \textit{Nein} beantwortet werden. Wir wollen nun jedem solchem Element $q \in \mathcal{L}$ einen Wahrheitswert bezüglich unseres im Zustand $\psi$ präparierten Systems zuordnen. Im Gegensatz zu unseren vorherigen Ausführungen wollen wir nun also in der Lage sein, den Aussagen $q$ an das präparierte System einen Wahrheitswert zuzuordnen, ohne das wir eine konkrete Messung am System durchführen. Anders ausgedrückt wollen wir also der Aussage $q$ einen Wahrheitswert zuordnen, auch wenn das System $\Sigma$ sich zum jetzigen Zeitpunkt in keinem Eigenzustand befindet. 

Wir erklären dazu zu jedem solchem $q \in \mathcal{L}$ einen Orthogonalprojektor $P_q$, der jeden Vektor $\phi \in L^2(\mu)$ auf den abgeschlossenen Unterraum $q$ projeziert. Das heißt es gilt
\begin{itemize}
\item [\textit{i)}] $P_q: L^2(\mu) \longrightarrow L^2(\mu),$
    \item [\textit{ii)}] $im(P_q) = q,$
    \item [\textit{iii)}] $ker(P) = q^{\perp}.$
\end{itemize}
Wir stellen uns nun die Frage, was passiert, wenn wir den zu einem Element $q \in \mathcal{L}$ zugehörigen orthogonalen Operator $P_q$ auf den Zustand $\psi$, in dem das System $\Sigma$ präpariert ist, anwenden. 

Ein Fall der eintreten kann ist, dass $\psi$ ein Eigenvektor von $P_q$ zum Eigenwert $0$ ist,
\begin{align}
    P_q(\psi) = 0 \cdot \psi = 0,
\end{align}
d.h. $\psi \in q^{\perp}$. In diesem Fall ist die Wahrscheinlichkeit dafür, dass man einen zum Eigenvektor $\phi_i \in q$ gehörenden Eigenwert $\lambda_i$ bei einer Messung von $\mathcal{V}$ am System $\Sigma$ misst, gleich Null. Aus diesem Grund kann man in einem derartigen Fall der Aussage $q$ den Wahrheitswert $0$ bzw. \textit{falsch} zuordnen. 

Ein weiterer Fall der eintreten kann ist, dass $\psi$ ein Eigenvektor von $P_q$ zum Eigenwert $1$ ist,
\begin{align}
    P_q(\psi) = 1 \cdot \psi = \psi,
\end{align}
d.h. $\psi \in q$. In diesem Fall ist die Wahrscheinlichkeit dafür, dass man einen zum Eigenvektor $\phi_i \in q$ gehörenden Eigenwert $\lambda_i$ bei einer Messung von $\mathcal{V}$ am System $\Sigma$ misst, gleich Eins. Aus diesem Grund kann man in einem derartigen Fall der Aussage $q$ den Wahrheitswert $1$ bzw. \textit{wahr} zuordnen.

In den beiden eben besprochenen Fällen lässt sich also der Wahrheitswert der Aussage $q$ bzgl. des in dem Zustand $\psi$ präparierten Systems $\Sigma$ in den beiden Eigenwerten des zum Element $q$ gehörigen orthonalen Projektors codieren. In diesem Fall lässt sich die Aussage als eindeutig \textit{wahr} oder \textit{falsch} bezeichnen.

Was passiert nun aber, wenn $\psi$ kein Eigenvektor von $P_q$ ist? In diesem Fall gilt
\begin{align}
    P_q(\psi) = \theta,
\end{align}
mit $\theta \neq \psi$ und $\theta \neq 0$. Es folgt, dass 
\begin{align}
    \langle P_q(\psi),P_q(\psi) \rangle = \langle \theta,\theta \rangle = g < 1
\end{align}
mit $g > 0$ ist, d.h. die Projektionswahrscheinlichkeit von $\psi$ auf den abgeschlossenen Unterraum $q$ ist kleiner als Eins und ungleich Null. Das wiederum bedeutet, dass mit einer gewissen Wahrhscheinlichkeit, die der Projektionswahrscheinlichkeit entspricht, die Wellenfunktion $\psi$ bei Messung des Systems $\Sigma$ in einen Eigenzustand aus $q$ kollabiert. In diesem Sinne können wir in einem derartigen Fall der Aussage $q$ als Wahrheitswert die Projektionswahrscheinlichkeit als Zahl zwischen $0$ und $1$ zuordnen. Diese Zahl liefert uns eine gewisse Gewissheit darüber, mit welcher Wahrscheinlichkeit die Aussage $q$ \textit{wahr} für das präparierte System $\Sigma$ ist. 

In diesem Sinne kann man die Quantenlogik daher als eine Art mehrwertige Logik bezüglich obiger Aussagen $q$ auffassen, in der also jeder darartigen Aussagen an das präparierte System $\Sigma$ eine Zahl aus $[0,1]$ zugeordnet wird. Die Zahl $0$ steht dabei für \textit{falsch} und die Zahl $1$ für \textit{wahr}. Die Zahlen zwischen $0$ und $1$ stellen den Grad der Gewissheit darüber da, ob eine solche Aussage für das System $\Sigma$ wahr ist.

Am Ende möchten wir wieder ein kleines Beispiel angeben:

\begin{example}
Wir beziehen uns wieder auf das Beispiel 4.1, d.h. unser physikalisches System ist wieder durch den eindimensionalen harmonischen Oszillator gegeben und die Observable die wir betrachten ist die Energie. Damit sind all die Propositionen, denen wir nun einen Wahrheitswert vor einer Messung des Systems zuordnen wollen, diejenigen komplexen Unterräume, die von den normierten Eigenvektoren $\tilde{\psi_i}$ des Hamiltonoperators $\hat{H}$ aufgespannt werden. Die orthogonalen Projektoren sind dabei gegeben durch 
\begin{align}
    P_q = \sum\limits_{i \in I} \langle \tilde{\psi_i} , \cdot \rangle \tilde{\psi_i},
\end{align}
wobei $I$ eine Indexmenge ist und $\{\tilde{\psi_i} \: : \: i \in I\}$ eine Basis von $q$. Sei $\tilde{\psi}$ wieder 
\begin{align}
    \tilde{\Psi} = \sqrt{\frac{3}{4}}\tilde{\psi}_0 + \sqrt{\frac{1}{4}}\tilde{\psi}_1,
\end{align}
und sei $q$ gegeben durch die Menge $\{\phi \: : \: \phi = \mu \tilde{\psi_0}, \mu \in \mathbb{C}\}$. Dann ist der orthogonale Projektor von $q$ 
\begin{align}
     P_q = \langle \tilde{\psi_0} , \cdot \rangle \tilde{\psi_0}
\end{align}
und es gilt 
\begin{align}
    P_q(\tilde{\psi}) =& \: \langle \tilde{\psi_0} , \tilde{\psi} \rangle \tilde{\psi_0} \nonumber\\=& \: \big\langle \tilde{\psi_0} , \sqrt{3/4}\tilde{\psi}_0 + \sqrt{1/4}\tilde{\psi}_1 \big\rangle \tilde{\psi_0} \nonumber\\ =& \: \big\langle \tilde{\psi_0} , \sqrt{3/4}\tilde{\psi}_0 \big\rangle \tilde{\psi_0} + \big\langle \tilde{\psi_0} , \sqrt{1/4}\tilde{\psi}_1 \big\rangle \tilde{\psi_0} \nonumber\\ =& \: \sqrt{3/4} \big\langle \tilde{\psi_0} , \tilde{\psi}_0 \big\rangle \tilde{\psi_0} + \sqrt{1/4} \big\langle \tilde{\psi_0} , \tilde{\psi}_1 \big\rangle \tilde{\psi_0} 
\end{align}
Da die $\tilde{\psi_i}$ eine Orthonormalbasis bilden gilt $\langle \tilde{\psi_0} , \tilde{\psi}_0 \rangle = 1$  und $\langle \tilde{\psi_0} , \tilde{\psi}_1 \rangle = 0$. Daraus folgt
\begin{align}
    P_q(\tilde{\psi}) = \sqrt{3/4} \tilde{\psi_0}
\end{align}
und 
\begin{align*}
    \langle  P_q(\tilde{\psi}) ,  P_q(\tilde{\psi}) \rangle = \frac{3}{4} \equiv 75\%
\end{align*}
\end{example}

\newpage



\section{Zusammengesetzte physikalische Systeme}

Seien $\Gamma_1$ und $\Gamma_2$ zwei physikalische Systeme, wobei vorerst nicht näher spezifiziert wird, ob es sich dabei um klassische oder quantenmechanische Systeme handelt. Im Folgenden sei nur festgelegt, dass entweder beide klassische Systeme oder beide quantenmechanische Systeme darstellen. 

Zu $\Gamma_1$ bzw. $\Gamma_2$ gehören jeweils wieder die Propositionensysteme $\mathcal{L}_1$ bzw. $\mathcal{L}_2$. Ob es sich dabei um ein klassisches oder quantenmechanisches Propositionensystem handelt, hängt nun natürlich davon ab, ob es sich bei den Systeme $\Gamma_1, \Gamma_2$ um klassische oder quantenmechanische Systeme handelt. Wir wollen nun $\Gamma_1$ und $\Gamma_2$ als die beiden Teilsysteme des Systems $\Gamma$ verstehen, d.h. $\Gamma$ stellt ein aus $\Gamma_1$ und $\Gamma_2$ zusammengesetztes System dar. 

Auch zu $\Gamma$ erklären wir ein zugehöriges Propositionensystem, dass wir mit $\mathcal{L}$ bezeichnen. Die Frage, die sich nun stellt, ist, welche physikalischen Forderungen wir an $\Gamma_1, \Gamma_2$ und $\Gamma$ stellen sollten, um eine solche Beschreibung sinnvoll zu gewährleisten. Im Folgenden orientieren wir uns an \cite{aerts1978physical} und fordern die Erfüllung folgender Bedingungen:
\begin{itemize}
    \item[\textit{1.}] Jede Eigenschaft von $\Gamma_1$ und jede Eigenschaft von $\Gamma_2$ sei eine Eigenschaft vom Gesamtsystem $\Gamma$. Mathematisch heißt das, dass wir Abbildungen 
    \begin{align}
        h_1: \mathcal{L}_1 \longrightarrow \mathcal{L} \\ h_2: \mathcal{L}_2 \longrightarrow \mathcal{L}
    \end{align}
    erklären, die anschaulich jeder Eigenschaft von $\Gamma_1$ und $\Gamma_2$, zu denen jeweils korrespondierende Propositionen in $\mathcal{L}_1$ und $\mathcal{L}_2$ existieren, eine  Eigenschaft in $\Gamma$ zuordnen, die wieder durch eine Proposition in $\mathcal{L}$ ausgedrückt wird.
    \item[\textit{2.}] Die physikalische Struktur von $\Gamma_1$ und $\Gamma_2$ sollen erhalten bleiben, wenn diese jeweils als Teile des Systems $\Gamma$ gesehen werden. Diese Forderung drücken wir dadurch aus, dass die Abbildungen $h_1$ und $h_2$ jeweils strukturerhaltend sein sollen, d.h. für $i \in \{1,2\}$ soll gelten, dass 
    \begin{itemize}
    \item[\textit{i)}] $p, q \in \mathcal{L}_i$ mit $p \leftrightarrow q \: \: \Longrightarrow \: \: h_i(p) \leftrightarrow h_i(q)$
    \item[\textit{ii)}] $(p_k)_{k \in I} \subseteq \mathcal{L}_i$ mit $I$ Indexmenge $\: \: \Longrightarrow \: \: h_i\bigl(\bigsqcup_{k \in I} p_k\bigr) = \bigsqcup_{k \in I} h_i(p_k)$
    \end{itemize}
    Daraus folgt, dass $h_i(\mathbf{0}_{\mathcal{L}_i}) = \mathbf{0}_{\mathcal{L}}$ und $h_i(p') = h_i(p)' \sqcap h_i(\mathbf{I}_{\mathcal{L}_i})$ für $i=1,2$ ist, denn
    \begin{align}
        h_i(p_i) =&\: p \nonumber\\=&\: h_i(p_i \sqcup \mathbf{0}_{\mathcal{L}_i}) \nonumber\\=&\: h_i(p_i) \sqcup h_i(\mathbf{0}_{\mathcal{L}_i}) \nonumber\\=&\: p \sqcup q, \qquad p_i \in \mathcal{L}_i, p,q \in \mathcal{L}\\ &\Longrightarrow p = p \sqcup q \qquad \forall p\in \mathcal{L}\\&\Longrightarrow\: q = h_i(\mathbf{0}_{\mathcal{L}_i}) = \mathbf{0}_\mathcal{L}
    \end{align}
    und 
    \begin{align}
        h_i(\mathbf{I}_{\mathcal{L}_i}) =&\: h_1(p_i \sqcup p_i') \nonumber\\=&\: h_i(p_i) \sqcup h_i(p_i'), \qquad \forall p_i \in \mathcal{L}_i \\\Longrightarrow (h_i(p_i))' \sqcap h_i(\mathbf{I}_{\mathcal{L}_i}) =&\: ((h_i(p_i))' \sqcap h_i(p_i)) \sqcup h_i(p_i') \nonumber\\=&\: \mathbf{0}_\mathcal{L} \sqcup h_i(p_i') \nonumber\\=&\: h_i(p_i'). \label{eq: 83}
    \end{align}
    Wir haben dabei in \eqref{eq: 83} ausgenutzt, dass $h_i(p_i') \leq h_i(p_i)'$ ist. In klassischen Logikverbänden lässt sich dann mit der Distributivität und in Quantenlogiken mit dem modularen Gesetz argumentieren. Den Ausdruck $h(\mathbf{I}_{\mathcal{L}_i})$ interpretieren wir nun als ein Ja/Nein-Experiment auf $\Gamma$, welches immer dann wahr ist, wenn das System $\Gamma_i$ existiert. Da wir die Existenz der beiden Systeme $\Gamma_1$ und $\Gamma_2$ vorrausgesetzt haben, folgt, dass $h_i(\mathbf{I}_{\mathcal{L}_i})$ für $i = 1,2$ immer wahr ist. Daraus folgt, dass $h_1(\mathbf{I}_{\mathcal{L}_1}) = h_2(\mathbf{I}_{\mathcal{L}_2}) = \mathbf{I}_\mathcal{L}$ ist, d.h. $h_1$ und $h_2$ bilden jeweils die Identitäten aus $\mathcal{L}_1$ und $\mathcal{L}_2$ auf die Identität in $\mathcal{L}$ ab. Mittels Lemma 1 und der eben gemachten Beobachtung folgt, dass $h_i(\bigsqcap_{k \in I} p_k) = \bigsqcap_{k \in I} h_i(p_k)$ ist.
    
    Eine Abbildung zwischen Propositionensystemen, die den Bedingungen \textit{i)} und \textit{ii)} genügen, nennt man c-Morphismus. Bildet ein c-Morphismus zuzätzlich die Identität auf die Identität ab, heißt er unitär. $h_1$ und $h_2$ sind beides unitäre c-Morphismen. 
    \item[\textit{3.}] Die Kopplung der Systeme $\Gamma_1$ und $\Gamma_2$ sei derart gegeben, dass kein Experiment an $\Gamma_1$ das System $\Gamma_2$ stört und umgedreht. Mathematisch formuliert: 
    \begin{align}
        p_1 \in \mathcal{L}_1, \: p_2 \in \mathcal{L}_2 \: \: \Longrightarrow \: \: h_1(p_1) \leftrightarrow h_2(p_2)
    \end{align}

    \item[\textit{4.}] Durch die Kopplung der Systeme $\Gamma_1$ und $\Gamma_2$ verlieren wir keine Information über die beiden Teilsysteme. Führen wir Messungen an $\Gamma_1$ und $\Gamma_2$ durch die uns jeweils maximale Information über diese Systeme liefern, so liefert uns das maximale Information über das System $\Gamma$. Da uns die Atome des Propositionensystems maximale Information über den Zustand des zugehörigen physikalischen System bereitstellen, lässt sich diese Forderung folgendermaßen mathematisch übersetzen: 
    
    Sei $p_1$ ein Atom in $\mathcal{L}_1$ und sei $p_2$ ein Atom in $\mathcal{L}_2$, dann ist $h_1(p_1) \sqcap h_2(p_2)$ ein Atom in $\mathcal{L}$.
\end{itemize}


\subsection{Zusammengesetzte klassische Systeme}

Im folgenden seien die Systeme $\Gamma_1 = \sigma_1$, $\Gamma_2 = \sigma_2$ und $\Gamma = \sigma$ jeweils klassischer Natur. Jedes System besitzt dabei jeweils seinen ganz eigenen Phasenraum: Zu $\sigma_1$ gehört der Phasenraum $\Omega_1$, zu $\sigma_2$ gehört $\Omega_2$ und zu $\sigma$ gehört $\Omega$. Schließlich konstruieren wir wie oben zu den gegeben Phasenräumen die jeweils zugehörigen Propositionensysteme $\mathcal{P}(\Omega_1)$, $\mathcal{P}(\Omega_2)$ und $\mathcal{P}(\Omega)$.

\begin{theorem}
Existieren zwei unitäre c-Morphismen 
\begin{align}
    h_1: \mathcal{P}(\Omega_1) \longrightarrow \mathcal{P}(\Omega), \\ h_2: \mathcal{P}(\Omega_2) \longrightarrow \mathcal{P}(\Omega),
\end{align}
sodass für Atome $p_1$ aus $\mathcal{P}(\Omega_1)$ und für Atome $p_2$ aus $\mathcal{P}(\Omega_2)$ gilt, dass der Ausdruck $h_1(p_1) \cap h_2(p_2) \in \mathcal{P}(\Omega)$ wieder ein Atom ist, so ist $\mathcal{P}(\Omega)$ isomorph zu $\mathcal{P}(\Omega_1 \times \Omega_2)$. Dabei bezeichnet $\Omega_1 \times \Omega_2$ das kartesische Produkt zwischen den Phasenräumen $\Omega_1$ und $\Omega_2$.
\end{theorem}

\begin{proof}
Sei $\mathcal{A} := \{\{(x_1,x_2)\} : x_1 \in \Omega_1 \: \: \textrm{und} \: \: x_2 \in \Omega_2\}$. Offensichtlicherweise ist $\mathcal{A}$ bijektiv zu $\Omega_1 \times \Omega_2$. Desweiteren betrachten wir eine Menge \\ $\mathcal{B} := \{h_1(p_1) \cap h_2(p_2) : p_1 = \{x_1\}, \: x_2 = \{x_2\}\}$. Wir definieren die Abbildung 
\begin{align}
    \zeta : \mathcal{A} \longrightarrow \mathcal{B}
\end{align}
durch 
\begin{align}
    \mathcal{A} \ni \{(x_1, x_2)\} \longmapsto h_1(\{x_1\}) \cap h_2(\{x_2\}) \in \mathcal{B} \subseteq \mathcal{P}(\Omega).
\end{align}
Die Abbildung $\zeta$ ist injektiv, denn seien $\{x_1\} = p_1$, $\{y_1\} = q_1$, $\{x_2\} = p_2$ und $\{y_2\} = q_2$ Atome aus $\mathcal{P}(\Omega_1)$ und $\mathcal{P}(\Omega_2)$, wobei $x_1, y_1 \in \Omega_1$ und $x_2, y_2 \in \Omega_2$. Dann gilt
\begin{align}
    h_1(p_1) \cap h_2(p_2) =& \: \:  h_1(q_1) \cap h_2(q_2) \nonumber\\ =& \: \: (h_1(q_1) \cap h_2(q_2)) \cap (h_1(p_1) \cap h_2(p_2)) \nonumber\\ =& \: \: (h_1(p_1) \cap h_1(q_1)) \cap (h_2(p_2) \cap h_2(q_2)) \nonumber\\ =& \: \: h_1(p_1 \cap q_2) \cap h_2(p_2 \cap q_2)
\end{align}
Angenommen $p_1 \neq q_1$, so folgt, dass $x_1 \neq y_1$. Daraus folgt offensichtlich, dass $p_1 \cap q_1 = \{0\} = \mathbf{0}_{\Omega_1}$ ist. Da $h_1$ ein c-Morphismus ist, folgt demnach, dass $h_1(p_1 \cap q_1) = \mathbf{0}_\Omega$ ist. Daraus folgt aber, dass $h_1(p_1) \cap h_2(p_2) = \mathbf{O}_\Omega$ ist, was im Widerspruch zur Annahme steht. Demnach ist $\zeta$ injektiv. \\ Per Definition gilt ebenso, dass $\zeta$ surjektiv ist. \\ Wir machen folgende Beobachtung:
\begin{align}
    \Omega =& \: \: h_1(\Omega_1) \cap h_2(\Omega_2) \nonumber\\ =& \: \: h_1 ( \cup_{x_1 \in \Omega_1} \{x_1\}) \cap h_2 ( \cup_{x_2 \in \Omega_2} \{x_2\}) \nonumber\\ =& \: \: ( \cup_{x_1 \in \Omega_1} h_1(\{x_1\})) \cap ( \cup_{x_2 \in \Omega_2} h_2(\{x_2\})) \nonumber\\ =& \: \: \cup_{x_1 \in \Omega_1} \cup_{x_2 \in \Omega_2} ( h_1(\{x_1\}) \cap h_2(\{x_2\}) ) \nonumber\\ =& \: \: \{h_1(\{x_1\}) \cap h_2(\{x_2\}) : x_1 \in \Omega_1, \: x_2 \in \Omega_2\} \nonumber\\ =& \: \: \{\zeta(\{(x_1,x_2)\}) : (x_1, x_2) \in \Omega_1 \times \Omega_2\}
\end{align}
Aus dieser Beobachtung folgt, dass die Bildmenge von $\zeta$ den Phasenraum $\Omega$ ergibt und demnach induziert $\zeta$ eine Bijektion zwischen $\Omega$ und $\Omega_1 \times \Omega_2$. \\ Mittels der Abbildung 
\begin{align}
    \eta : \mathcal{P}(\Omega) \longrightarrow \mathcal{P}(\Omega_1 \times \Omega_2),
\end{align}
gegeben durch 
\begin{align}
    \mathcal{P}(\Omega) \ni A \longmapsto \eta(A) = \{(x_1, x_2) : \zeta(\{(x_1, x_2)\}) \subseteq A\},
\end{align}
erhalten wir schließlich mittels der Bijektivität von $\zeta$ einen Isomorphismus zwischen $\mathcal{P}(\Omega)$ und $\mathcal{P}(\Omega_1 \times \Omega_2)$.
\end{proof}

Am Ende dieses Abschnittes wollen wir uns dieses Resultat noch an einem kleinen Beispiel veranschaulichen:

\begin{example}
Betrachten wir ein System zweier geladener Teilchen der Massen $m_1 > 0$ und $m_2 > 0$ und Ladungen $q_1$ und $q_2$ im dreidimensionalen Raum, die keinen Zwangsbedingungen unterliegen. Die Lagrange-Funktion $L$ des Gesamtsystems ist dann gegeben durch 
\begin{align}
    L = \sum\limits_{i =1}^{2} \frac{1}{2}m_i \dot{x_i}^2 - \frac{1}{4 \pi \epsilon_0}\frac{q_1 q_2}{\|x_1 - x_2\|},
\end{align}
wobei $\epsilon_0$ die Permittivität des Vakuums, $x_1$ bzw. $x_2$ der Ortsvektor des ersten bzw. des zweiten Teilchens im dreidimensionalen euklidischen Raum und $\|x_1-x_2\|$ der euklidischen Abstand der beiden Teilchen ist. Die Hamilton-Funktion ergibt sich dann zu
\begin{align}
    H = \sum\limits_{i=1}^{2}\frac{p_i}{2m_i} + \frac{1}{4 \pi \epsilon_0}\frac{q_1 q_2}{\|x_1 - x_2\|},
\end{align}
wobei $p_1$ bzw. $p_2$ der Impuls des ersten bzw. des zweiten Teilchens ist. Die Hamilton-Funktion ist damit gegeben als Funktion der zwölf unabhängigen Koordinaten $x^1_{1},x^2_{1},x^3_{1},x^1_{2},x^2_{2},x^3_{2},p_{1,1},p_{2,1},p_{3,1},p_{1,2},p_{2,2},p_{3,2}$. Dabei bezeichnet $x^j_i$ die $j$-te Komponente des $i$-ten Ortsvektors und $p_{j,i}$ die $j$-te Komponente des $i$-ten Impulses für $i=1,2$. Die Dynamik des Systems wird damit in einem $12$ dimensionalen Phasenraum $\Omega$ beschrieben, welcher von diesen $12$ unabhängigen Koordinaten aufgespannt wird und sich als das kartesische Produkt der beiden Phasenräume 
\begin{align}
    \Omega_1 = \{(x^1_{1},x^2_{1},x^3_{1},p_{1,1},p_{2,1},p_{3,1}) \: : \: x^{j}_1 \in \mathbb{R}; \: p_{j,1} \in \mathbb{R} \: \forall j = 1,2,3\}
\end{align}
und 
\begin{align}
    \Omega_2 = \{(x^1_{2},x^2_{2},x^3_{2},p_{1,2},p_{2,2},p_{3,2}) \: : \: x^{j}_2 \in \mathbb{R}; \: p_{j,2} \in \mathbb{R} \: \forall j = 1,2,3\}
\end{align}
ergibt.
\end{example}


\subsection{Zusammengesetzte quantenmechanische Systeme}

In diesem Abschnitt seien die Systeme $\Gamma_1 = \Sigma_1$, $\Gamma_2 = \Sigma_2$ und $\Gamma = \Sigma$ quantenmechanischer Natur. $\Sigma_1$ und $\Sigma_2$ seien dabei nun beliebige quantenmechanische Systeme. Wie im Unterabschnitt 4.1 diskutiert ordnen wir dem System $\Sigma_1$ einen (separablen) Hilbertraum $\mathcal{H}_1$, dem System $\Sigma_2$ einen (separablen) Hilbertraum $\mathcal{H}_2$ und dem System $\Sigma$ einen (separablen) Hilbertraum $\mathcal{H}$ zu. Die konkrete Gestalt der einzelnen Hilberträume hängt jeweils von den Systemen $\Sigma_1$ und $\Sigma_2$ ab. Wir nehmen außerdem mit Blick auf den noch kommenden Abschnitt über Tensorprodukträume zusätzlich an, dass die Hilberträume $\mathcal{H}_1$ und $\mathcal{H}_2$ isometrisch isomorph sind, d.h. es existiert ein stetige, bijektive, lineare Abbildung zwischen $\mathcal{H}_1$ und $\mathcal{H}_2$, die das Skalarprodukt erhält. Genau wie im klassischen Fall ordnen wir jedem System das aus dem jeweiligen Zustandsvektorraum konstruierte (quantenmechanische) Propositionensystem zu. Dabei sei $\mathcal{L}(\mathcal{H}_1) := \mathcal{L}_1$ das Propositionensystem zu $\Sigma_1$, $\mathcal{L}(\mathcal{H}_2) := \mathcal{L}_2$ das Propositionensystem zu $\Sigma_2$ und $\mathcal{L}(\mathcal{H}) := \mathcal{L}$ das Propositionensystem zu $\Sigma$.

Ziel dieses Abschnittes wird es sein, ein zu Abschnitt 5.1 analoges Resultat für zusammengesetzte quantenmechanische Systeme abzuleiten. Wir werden dabei zeigen, dass das Propositionensystem $\mathcal{L}(\mathcal{H})$ u.a. isomorph zu $\mathcal{L}(\mathcal{H}_1 \otimes \mathcal{H}_2)$ ist. Dabei ist $\otimes$ das Tensorprodukt und $\mathcal{L}(\mathcal{H}_1 \otimes \mathcal{H}_2)$ die Menge aller komplexer, abgeschlossener Unterräume von $\mathcal{H}_1 \otimes \mathcal{H}_2$. Die Bedeutung dieses Resultats werden wir kurz am Ende des Abschnittes diskutieren.  

Im folgenden nutzen wir wichtige Resultate aus der Verbandstheorie \cite{piron1976foundations} und der Theorie der strukturerhaltenden Abbildungen quantenmechanischer Propositionensysteme \cite{aerts1979structure}, die wir der Vollständigkeit halber noch auflisten möchten: \\ \\  Die folgenden Lemmata gelten für alle schwach modularen, vollständigen, orthokomplementierten Verbände $\mathbb{V} =(V, \sqcup, \sqcap, ')$.

\begin{lemma}
Sei $\mathbb{V}$ ein schwach modularer, vollständiger orthokomplementierter Verband und $I$ eine beliebige Indexmenge. Seien $a_i, b \in V$ für alle $i$ aus $I$. Dann gilt: 
\begin{align}
    b \leftrightarrow a_i \: \: \: \: \forall i \in I \: \: \Longrightarrow \: \: \sqcup_{i \in I} (b \sqcap a_i) = b \sqcap (\sqcup_{i \in I} a_i)
\end{align}
\end{lemma}

\begin{lemma}
Sei $\mathbb{V}$ ein schwach modularer, orthokomplementierter Verband. Es existieren zwei Kriterien, um zu bestimmen, ob die Elemente $a, b \in V$ miteinander kompatibel sind: 
\begin{align}
    a \leftrightarrow b \: \: : \Longleftrightarrow& \: \: (a \sqcap b) \sqcup (a' \sqcap b) = b \nonumber\\ : \Longleftrightarrow& \: \: (a \sqcup b') \sqcap b = a \sqcap b
\end{align}
\end{lemma}

\begin{lemma}
Sei $\mathbb{V}$ ein schwach modularer, orthokomplementierter Verband. Dann ist für $a, b, c \in V$ das Tripel $(a,b,c)$ distributiv, sobald ein Element aus dem Tripel mit den jeweils anderen beiden kompatibel ist.
\end{lemma} 

Im Folgenden bezeichne nun wieder $\mathcal{L}(\mathcal{H})$ die Menge aller komplexer abgeschlossener Unterräume eines (separablen) komplexen Hilbertraumes $\mathcal{H}$. Für $x \in \mathcal{H}$ bezeichne $\langle x \rangle$ den von $x$ erzeugten, komplex eindimensionalen Unterraum. Der Ausdruck $p^{\perp}$ bezeichne wieder das orthogonale Komplement von $p \in \mathcal{L}(\mathcal{H})$ und das Symbol $\oplus$ bezeichne die direkte Summe zwischen Vektorräumen. 

Auf der Menge $\mathcal{L}(\mathcal{H})$ lässt sich wie oben eine partielle Ordnung $\leq$ definieren, die durch die Teilmengenrelation $\subseteq$ gegeben ist. 

Abschließend nutzen wir wieder folgende Kurzschreibweisen: Für $p,q \in \mathcal{L}(\mathcal{H})$ sei $p \sqcup q := \overline{\textrm{span}_\mathbf{C} (p \cup q)}$ und $p \sqcap q := p \cap q$. Den Verband $(\mathcal{L}(\mathcal{H}), \sqcup, \sqcap, \cdot^{\perp})$ kürzen wir wieder mit $\mathcal{L}(\mathcal{H})$ ab.

Als Nächstes wollen wir nun den für diese Arbeit zentralen Begriff des sogenannten m-Morphismus definieren. Dazu zuerst folgende Definition:

\begin{definition}
 Sei $\mathcal{H}$ ein komplexer Hilbertraum und $\mathcal{L}(\mathcal{H})$ der zugehörige Hilbert-Verband und seien $p,q \in \mathcal{L}(\mathcal{H})$. Das Paar $(p,q)$ heißt modulares Paar, falls $\forall r \in \mathcal{L}(\mathcal{H})$ mit $r \leq q$ folgt, dass $(p \sqcup r) \sqcap q = (p \sqcap q) \sqcup r$ ist.
\end{definition}

Mittels dieser Definition ist ein m-Morphismus wie folgt definiert:

\begin{definition}
 Seien $\mathcal{H}$ und $\mathcal{G}$ zwei komplexe Hilberträume und $f$ ein c-Morphismus von $\mathcal{L}(\mathcal{H})$ nach $\mathcal{L}(\mathcal{G})$, mit der Eigenschaft, dass $f$ modulare Paare in $\mathcal{L}(\mathcal{H})$ auf modulare Paare in $\mathcal{L}(\mathcal{G})$ abbildet. In diesem Fall nennt man $f$ m-Morphismus.
\end{definition}

Der nächste Satz liefert ein einfaches Verfahren, um nachzuweisen, dass ein gegebener c-Morphismus ein m-Morphismus ist.

\begin{proposition}
Seien $\mathcal{H}$ und $\mathcal{G}$ zwei komplexe Hilberträume und $f$ ein c-Morphismus von $\mathcal{L}(\mathcal{H})$ nach $\mathcal{L}(\mathcal{G})$. Dann ist $f$ ein m-Morphismus, falls für alle $x,y \in \mathcal{H}$ mit $x \neq y$ gilt, dass 
\begin{align}
    f(\langle x - y \rangle) \subseteq f(\langle x \rangle) \oplus f(\langle y \rangle).
\end{align}
\end{proposition}

Man kann sich nun berechtigterweise die Frage stellen, warum der abstrakt daherkommende Begriff des m-Morphismus nützlich für diese Arbeit sein sollte. Den Grund liefert folgendes Theorem: 

\begin{theorem}
Seien $\mathcal{H}$ und $\mathcal{G}$ zwei komplexe Hilberträume mit $\textrm{dim}(\mathcal{H}) \geq 3$ und $\textrm{dim}(\mathcal{G}) \geq 3$ und seien $x,y,z \in \mathcal{H}$. Sei $f: \mathcal{L}(\mathcal{H}) \rightarrow \mathcal{L}(\mathcal{G})$ ein m-Morphismus. Dann existiert zu jedem Paar $(x,y)$ mit $x,y \neq 0$ eine bijektive, beschränkte, lineare Abbildung $F_{y,x} : f(\langle x \rangle) \rightarrow f(\langle y \rangle)$ mit den folgenden Eigenschaften:
\begin{itemize}
    \item [\textit{i)}] $F_{x,x} = id_{f(\langle x \rangle)}$, wobei $id_{f(\langle x \rangle)}$ die Identitätsabbildung auf $f(\langle x \rangle)$ ist.
    \item [\textit{ii)}] $F_{x,y} = (F_{y,x})^{-1}$.
    \item [\textit{iii)}] $F_{z,y} \circ F_{y,x} = F_{z,x}$, wobei $\circ$ die gewöhnliche Operatorkomposition ist.
    \item[\textit{iv)}] $F_{y + z,x} = F_{y,x} + F_{z,x}$.
    \item[\textit{v)}] $F_{\lambda x, \lambda y} = F_{x,y}$, \: \: $\forall \lambda \in \mathbb{C}$.
    \item[\textit{vi)}] Für $\|x\|_\mathcal{H} = \|y\|_\mathcal{H}$ folgt $\langle F_{y,x}(x_1), F_{y,x}(x_2) \rangle_{\mathcal{G}} = \langle x_1, x_2 \rangle_{\mathcal{G}}$  $\forall x_1, x_2 \in f(\langle x \rangle)$, d.h. $F_{y,x}$ ist in diesem Falle eine Isometrie. Dabei ist $\langle \cdot, \cdot \rangle_\mathcal{H}$ bzw. $\langle \cdot, \cdot \rangle_\mathcal{G}$ das Skalarprodukt auf $\mathcal{H}$ bzw. $\mathcal{G}$ und $\| \cdot \|_\mathcal{H} := \sqrt{\langle \cdot, \cdot \rangle_\mathcal{H}}$.
\end{itemize}
Für alle $x,y \in \mathcal{H}$ mit $x \neq 0$ existieren darüber hinaus orthogonale Projektoren $P_1^{\langle x \rangle}$ und $P_2^{\langle x \rangle}$ auf $f(\langle x \rangle)$ mit den Eigenschaften:
\begin{itemize}
    \item [\textit{vii)}] $P_1^{\langle x \rangle} \circ P_2^{\langle x \rangle} = 0$, wobei $0$ der Nulloperator auf $f(\langle x \rangle)$ ist.
    \item [\textit{viii)}] $P_1^{\langle x \rangle} + P_2^{\langle x \rangle} = id_{f(\langle x \rangle)}$.
    \item [\textit{ix)}] $P_i^{\langle y \rangle} = F_{y,x} \circ P_i^{\langle x \rangle} \circ F_{x,y}$, \: \: $i \in \{1,2\}$.
    \item[\textit{x)}] $F_{\lambda x,x} = \lambda P_1^{\langle x \rangle} + \lambda^* P_2^{\langle x \rangle}$, \: \: $\forall \lambda \in \mathbb{C}$.
\end{itemize}
\end{theorem}

Die Nützlichkeit dieses Theorems wird darin bestehen, dass wir uns mittels der bijektiven Operatoren $F_{y,x}$ Abbildungen konstruieren wollen, welche später u.a. die Isomorphie zwischen $\mathcal{L}(\mathcal{H})$ und $\mathcal{L}(\mathcal{H}_1 \otimes \mathcal{H}_2)$ vermitteln sollen.

Desweiteren lassen sich m-Morphismen mittels dieses Theorems auf folgende einfache Art klassifizieren: 

\newpage

\begin{definition}
Seien $\mathcal{H}$ bzw. $\mathcal{G}$ komplexe Hilberträume und $\mathcal{L}(\mathcal{H})$ bzw. $\mathcal{L}(\mathcal{G})$ die zugehörigen Hilbert-Verbände, dann heißt ein m-Morphismus $f$ 
\begin{itemize}
    \item [\textit{i)}] linear, falls $F_{\lambda x,x} = \lambda id_{f(\langle x \rangle)}$ ist, 
    \item [\textit{ii)}] antilinear, falls $F_{\lambda x,x} = \lambda^* id_{f(\langle x \rangle)}$ ist,
    \item [\textit{iii)}] gemischt, falls $f$ weder linear, noch antilinear ist.
\end{itemize}
\end{definition}

Das nächste Theorem liefert eine einfache Möglichkeit, um herauszufinden, ob ein unitärer c-Morphismus bijektiv ist.

\begin{theorem}
Seien $\mathcal{H}$ und $\mathcal{G}$ komplexe Hilberträume der Dimension größer oder gleich als Drei. Die Abbildung $f: \mathcal{L}(\mathcal{H}) \rightarrow \mathcal{L}(\mathcal{G})$ sei ein unitärer c-Morphismus, mit der Eigenschaft, dass ein Atom $p \in \mathcal{L}(\mathcal{H})$ existiert, sodass $f(p)$ ein Atom aus $\mathcal{L}(\mathcal{G})$ ist. Es folgt, dass $f$ ein Isomorphismus ist.
\end{theorem}

Wir betrachten nun zwei quantenmechanische Systeme $\Sigma_1$ und $\Sigma_2$ mit deren zugehörigen (quantenmechanischen) Propositionensystemen $\mathcal{L}(\mathcal{H}_1)$ und $\mathcal{L}(\mathcal{H}_2)$, die wir zu einem quantenmechanischen System $\Sigma$ mit Propositionensystem $\mathcal{L}(\mathcal{H})$ zusammensetzen wollen. Dabei sind $\mathcal{H}_1, \mathcal{H}_2$ und $\mathcal{H}$ wieder komplexe Hilberträume und wir fordern noch zusätzlich, dass $\textrm{dim}(\mathcal{H}_1), \textrm{dim}(\mathcal{H}_2) \geq 3$ ist. Wir nehmen nun an, dass die obigen Bedingungen $\textit{1}, \textit{2}, \textit{3}$ und $\textit{4}$ erfüllt seien, d.h. 
\begin{itemize}
    \item [\uproman{1}.] Es existieren zwei unitäre c-Morphismen 
    \begin{align}
        h_1 : \mathcal{L}(\mathcal{H}_1) \longrightarrow \mathcal{L}(\mathcal{H})
    \end{align}
    und 
    \begin{align}
        h_2 : \mathcal{L}(\mathcal{H}_2) \longrightarrow \mathcal{L}(\mathcal{H}).
    \end{align}
    \item [\uproman{2}.] Für alle $p_1 \in \mathcal{L}(\mathcal{H}_1)$ und alle $p_2 \in \mathcal{L}(\mathcal{H}_2)$ gilt, dass 
    \begin{align}
        h_1(p_1) \leftrightarrow h_2(p_2).
    \end{align}
    \item [\uproman{3}.] Für Atome $p_1 \in \mathcal{L}(\mathcal{H}_1)$ und Atome $p_2 \in \mathcal{L}(\mathcal{H}_2)$ folgt, dass $h_1(p_1) \sqcap h_2(p_2)$ ein Atom in $\mathcal{L}(\mathcal{H})$ ist.
\end{itemize}

Mittels der oben angegeben Lemmata und Theoreme beweisen wir nun einige Lemmata, mit deren Hilfe wir in der Lage sein werden, das angestrebte Ziel dieses Abschnittes zu erreichen. Die meisten folgenden Beweise stammen dabei im wesentlichen aus der Arbeit \cite{aerts1978physical}, wurden aber detaillierter ausgearbeitet und ergänzt.

\begin{lemma}
Seien $x_1 \in \mathcal{H}_1$ und $x_2 \in \mathcal{H}_2$. Wir definieren die Abbildungen 
\begin{align}
      u_{\langle x_2 \rangle} : \mathcal{L}(\mathcal{H}_1) \longrightarrow& \: \: \mathcal{L}(h_2(\langle x_2 \rangle)) \nonumber\\
     p_1 \longmapsto& \: \: h_1(p_1) \sqcap h_2(\langle x_2 \rangle) 
\end{align}
\begin{align}
      v_{\langle x_1 \rangle} : \mathcal{L}(\mathcal{H}_2) \longrightarrow& \: \: \mathcal{L}(h_1(\langle x_1 \rangle)) \nonumber\\
     p_2 \longmapsto& \: \: h_1(\langle x_1 \rangle) \sqcap h_2(p_2) 
\end{align}
$u_{\langle x_2 \rangle}$ und $v_{\langle x_1 \rangle}$ sind Isomorphismen.
\end{lemma}

\begin{proof}
Wir zeigen zuerst, dass es sich bei $u_{\langle x_2 \rangle}$ um einen c-Morphismus handelt. Seien dazu $x_1 \in \mathcal{H}_1$, $x_2 \in \mathcal{H}_2$, $I$ eine beliebige Indexmenge und $(p_{1,i})_{i \in I}$ eine Folge in $\mathcal{L}(\mathcal{H}_1)$. Es gilt 
\begin{align}
    u_{\langle x_2 \rangle}(\sqcup_{i \in I} p_{1,i}) =& \: \: h_1(\sqcup_{i \in I} p_{1,i}) \sqcap h_2(\langle x_2 \rangle) \nonumber\\ =& \: \: (\sqcup_{i \in I} h_1(p_{1,i})) \sqcap h_2(\langle x_2 \rangle)
\end{align}
Da für alle $i \in I$ gilt, dass $p_{1,i} \in \mathcal{L}(\mathcal{H}_1)$ und $\langle x_2 \rangle \in \mathcal{L}(\mathcal{H}_2)$ ist, folgt mittels \uproman{2}, dass $h_1(p_{1,i}) \leftrightarrow h_2(\langle x_2 \rangle)$ $\forall i \in I$ ist. Mittels Lemma 5.1 folgt damit 
\begin{align}
    u_{\langle x_2 \rangle}(\sqcup_{i \in I} p_{1,i}) =& \: \: \sqcup_{i \in I} (h_1(p_{1,i}) \sqcap h_2(\langle x_2 \rangle)) \nonumber\\ =& \: \: \sqcup_{i \in I} u_{\langle x_2 \rangle} (p_{1,i}).
\end{align}
Sei nun $p_1, q_1 \in \mathcal{L}(\mathcal{H}_1)$, sodass $p_1 \leftrightarrow q_1$ ist. Wir wollen nun zeigen, dass daraus $u_{\langle x_2 \rangle}(q_1) \leftrightarrow u_{\langle x_2 \rangle}(p_1)$ folgt. Dazu wollen wir Lemma 5.2 nutzen. Wir berechnen daher zuerst $u_{\langle x_2 \rangle}(p_1^{\perp})$:
\begin{align}
    u_{\langle x_2 \rangle}(p_1^{\perp}) =& \: \: h_1(p_1^{\perp}) \sqcap h_2(\langle x_2 \rangle) \nonumber\\ =& \: \: h_1(p_1)^{\perp} \sqcap h_2(\langle x_2 \rangle) \nonumber\\ =& \: \: (h_1(p_1)^{\perp} \sqcup h_2(\langle x_2 \rangle)^{\perp}) \sqcap h_2(\langle x_2 \rangle) \nonumber\\ =& \: \: (h_1(p_1) \sqcap h_2(\langle x_2 \rangle))^{\perp} \sqcap h_2(\langle x_2 \rangle) \nonumber\\ =& \: \: u_{\langle x_2 \rangle}(p_1)^{\perp} \sqcap h_2(\langle x_2 \rangle) 
\end{align}
Dabei wurde im dritten Schritt ausgenutzt, dass $h_1(p_1)^{\perp} \leftrightarrow h_2(\langle x_2 \rangle)$ ist und man demzufolge Lemma 5.2 anwenden kann. 

Es gilt nun 
\begin{align}
    (u_{\langle x_2 \rangle}(q_1) \sqcup (u_{\langle x_2 \rangle}(p_1)^{\perp} \sqcap& \: h_2(\langle x_2 \rangle))) \sqcap u_{\langle x_2 \rangle}(p_1) \nonumber\\ = \: &(u_{\langle x_2 \rangle}(q_1) \sqcup u_{\langle x_2 \rangle}(p_1^{\perp})) \sqcap u_{\langle x_2 \rangle}(p_1) \nonumber\\ = \:  &u_{\langle x_2 \rangle}((q_1 \sqcup p_1^{\perp}) \sqcap p_1) \nonumber\\ = \: &u_{\langle x_2 \rangle}(q_1 \sqcap p_1) \nonumber\\ = \: &u_{\langle x_2 \rangle}(q_1) \sqcap  u_{\langle x_2 \rangle}(p_1)
\end{align}
Damit folgt nach Lemma 5.2, dass $u_{\langle x_2 \rangle}(q_1) \leftrightarrow u_{\langle x_2 \rangle}(p_1)$ ist. Damit haben wir gezeigt, dass $u_{\langle x_2 \rangle}$ ein c-Morphismus ist. $u_{\langle x_2 \rangle}$ ist darüber hinaus unitär, da $u_{\langle x_2 \rangle}(\mathcal{H}_1) = h_2(\langle x_2 \rangle)$ ist. 

Sei nun $p_1 \in \mathcal{L}(\mathcal{H}_1)$ ein Atom, dann gilt $u_{\langle x_2 \rangle}(p_1) = h_1(p_1) \sqcap h_2(\langle x_2 \rangle)$. Da $\langle x_2 \rangle$ ein Atom aus $\mathcal{L}(\mathcal{H}_2)$ ist, folgt, dass $ h_1(p_1) \sqcap h_2(\langle x_2 \rangle)$ ein Atom in $\mathcal{L}(\mathcal{H})$ ist. Damit bildet $u_{\langle x_2 \rangle}$ Atome auf Atome ab und mit Theorem 5.3 folgt, dass $u_{\langle x_2 \rangle}$ ein Isomorphismus ist. Analog argumentiert man für $v_{\langle x_1 \rangle}$.
\end{proof}

Mittels Lemma 5.4 sind wir nun in der Lage zu beweisen, dass es sich bei den Abbildungen $h_1$ und $h_2$ um m-Morphismen handelt!

\begin{lemma}
Die Abbildungen $h_1$ und $h_2$ sind m-Morphismen
\end{lemma}

\begin{proof}
Seien $x_1, y_1 \in \mathcal{H}_1$ mit $x_1, y_1 \neq 0$ und sei $x \in h_1(\langle x_1 - y_1 \rangle)$. Aus Lemma 5.4 folgt, dass die Abbildung $v_{\langle x_1 - y_1 \rangle}: \mathcal{L}(\mathcal{H}_2) \rightarrow \mathcal{L}(h_1(\langle x_1 - y_1 \rangle))$ ein Isomorphismus ist. Aus der Definition von $v_{\langle x_1 - y_1 \rangle}$ folgt, dass Atome aus $\mathcal{L}(\mathcal{H}_2)$ auf Atome aus $\mathcal{L}(h_1(\langle x_1 - y_1 \rangle))$ abgebildet werden und mittels der Umkehrungabbildung von $v_{\langle x_1 - y_1 \rangle}$ auch umgedreht. Demzufolge existiert ein Atom $\langle x_2 \rangle \in \mathcal{L}(\mathcal{H}_2)$, sodass 
\begin{align}
    v_{\langle x_1 - y_1 \rangle}(\langle x_2 \rangle) = h_1(\langle x_1 - y_1 \rangle) \sqcap h_2(\langle x_2 \rangle) = \langle x \rangle
\end{align}
ist. Aus $\langle x_1 - y_1 \rangle \subseteq \langle x_1 \rangle \sqcup \langle y_1 \rangle$ folgt, dass 
\begin{align}
    h_1(\langle x_1 - y_1 \rangle) =& \: \: h_1(\langle x_1 - y_1 \rangle \sqcap (\langle x_1 \rangle \sqcup \langle y_1 \rangle)) \nonumber\\ =& \: \: h_1(\langle x_1 - y_1 \rangle) \sqcap h_1(\langle x_1 \rangle \sqcup \langle y_1 \rangle) \nonumber\\ \subseteq& \: \: h_1(\langle x_1 \rangle \sqcup \langle y_1 \rangle) \nonumber\\ =& \: \: h_1(\langle x_1 \rangle) \sqcup h_1(\langle y_1 \rangle).
\end{align}
Das impliziert
\begin{align}
    \langle x \rangle =& \: \: h_1(\langle x_1 - y_1 \rangle) \sqcap h_2(\langle x_2 \rangle) \nonumber\\ \subseteq& \: \: (h_1(\langle x_1 \rangle) \sqcup h_1(\langle y_1 \rangle)) \sqcap h_2(\langle x_2 \rangle) \nonumber\\ =& \: \: (h_1(\langle x_1 \rangle) \sqcap h_2(\langle x_2 \rangle)) \sqcup (h_1(\langle y_1 \rangle) \sqcap h_2(\langle x_2 \rangle)).
\end{align}
Im letzten Schritt haben wir ausgenutzt, dass wegen \uproman{2} folgt, dass $h_2(\langle x_2 \rangle)$ jeweils mit $h_1(\langle y_1 \rangle)$ und $h_1(\langle x_1 \rangle)$ kompatibel ist und damit nach Lemma 5.3 die Distributivität folgt. Man beachte nun, dass $h_1(\langle x_1 \rangle) \sqcap h_2(\langle x_2 \rangle)$ und $h_1(\langle y_1 \rangle) \sqcap h_2(\langle x_2 \rangle)$ nach \uproman{3} Atome und damit endlichdimensionale Vektorräume sind. Da jeder endlichdimensionale Vektorraum abgeschlossen ist und auch die lineare Hülle endlichdimensionaler Vektorräume endlichdimensional ist, folgt 
\begin{align}
    \langle x \rangle \subseteq (h_1(\langle x_1 \rangle) \sqcap h_2(\langle x_2 \rangle)) \oplus (h_1(\langle y_1 \rangle) \sqcap h_2(\langle x_2 \rangle)).
\end{align}
Es folgt damit die Darstellung 
\begin{align}
    x = v + w,
\end{align}
wobei $v \in h_1(\langle x_1 \rangle)$ und $w \in h_1(\langle y_1 \rangle)$ ist. Da das $x$ beliebig war folgt, dass 
\begin{align}
    h_1(\langle x_1 - y_1 \rangle) \subseteq h_1(\langle x_1 \rangle) \oplus h_1(\langle y_2 \rangle)
\end{align}
ist. Mittels Satz 5.1 folgt damit, dass $h_1$ ein m-Morphismus ist. Analog argumentiert man für $h_2$.
\end{proof}

Da es sich bei $h_1$ und $h_2$ um m-Morphismen handelt, existieren gemäß Theorem 5.2 die Abbildungen $F_{y_1,x_1}: h_1(\langle x_1 \rangle) \rightarrow h_1(\langle y_1 \rangle)$ für $h_1$ und die Abbildungen $K_{y_1,x_2}: h_2(\langle x_2 \rangle) \rightarrow h_2(\langle y_2 \rangle)$ für $h_2$.

Im Folgenden wollen wir uns einige wichtige Eigenschaften dieser Operatoren ansehen. Zuerst benötigen wir aber noch folgende hilfreiche Definition: 

\begin{definition}
Seien $x_1, y_1 \in \mathcal{H}_1$, $x_2,y_2 \in \mathcal{H}_2$ und $x \in \mathcal{H}$. Wir definieren die linearen Operatoren 
\begin{align}
      \tilde{F}_{y_1,x_1} : \mathcal{H} \longrightarrow& \: \: \mathcal{H} \nonumber\\
     x \longmapsto& \: \: \begin{cases} \tilde{F}_{y_1,x_1}(x) = x & \text{falls } x \in \mathcal{H}\setminus h_1(\langle x_1 \rangle) \\ \tilde{F}_{y_1,x_1}(x) = F_{y_1,x_1}(x) & \text{falls } x \in h_1(\langle x_1 \rangle) \end{cases}
\end{align}
und 
\begin{align}
     \tilde{K}_{y_2,x_2} : \mathcal{H} \longrightarrow& \: \: \mathcal{H} \nonumber\\
     x \longmapsto& \: \: \begin{cases} \tilde{K}_{y_2,x_2}(x) = x & \text{falls } x \in \mathcal{H}\setminus h_2(\langle x_2 \rangle) \\ \tilde{K}_{y_2,x_2}(x) = K_{y_2,x_2}(x) & \text{falls } x \in h_2(\langle x_2 \rangle) \end{cases}
\end{align}
\end{definition}

\begin{lemma}
Seien $x_1, y_1 \in \mathcal{H}_1$ und $x_2, y_2 \in \mathcal{H}_2$ jeweils linear unabhängig und ungleich dem Nullvektor in $\mathcal{H}_1$ bzw. $\mathcal{H}_2$. Sei $x \in h_1(\langle x_1 \rangle) \sqcap h_2(\langle x_2 \rangle)$. Dann gilt 
\begin{align}
    F_{y_1,x_1}(x) \in& \: h_1(\langle y_1 \rangle) \sqcap h_2(\langle x_2 \rangle), \\ K_{y_2,x_2}(x) \in& \: h_1(\langle x_1 \rangle) \sqcap h_2(\langle y_2 \rangle) 
\end{align}
und
\begin{align}
    (F_{y_1,x_1} \circ K_{y_2,x_2})(x) = (K_{y_2,x_2} \circ F_{y_1,x_1})(x). 
\end{align}
\end{lemma}

\begin{proof}
Sei $x \in h_1(\langle x_1 \rangle) \sqcap h_2(\langle x_2 \rangle)$. Es gilt per Definition 
\begin{align}
    x = K_{y_2,x_2}(x) + u
\end{align}
mit $K_{y_2,x_2}(x) \in h_2(\langle y_2 \rangle)$ und eindeutigem $u \in h_2(\langle x_2 - y_2 \rangle) \subseteq (h_2(\langle x_2 \rangle) \oplus h_2(\langle y_2 \rangle))$. Wir zeigen, dass $K_{y_2,x_2}(x) \in h_1(\langle x_1 \rangle)$ ist. Wir wissen, dass gilt 
\begin{align}
    h_1(\langle x_1 \rangle) \leftrightarrow& \: h_2(\langle y_2 \rangle), \\  h_1(\langle x_1 \rangle) \leftrightarrow& \: h_2(\langle x_2 - y_2 \rangle)
\end{align}
Aus Lemma 5.2 folgt, dass 
\begin{align}
    h_2(\langle y_2 \rangle) = (h_1(\langle x_1 \rangle) \sqcap h_2(\langle y_2 \rangle)) \sqcup (h_1(\langle x_1 \rangle)^{\perp} \sqcap h_2(\langle y_2 \rangle))
\end{align}
und 
\begin{align}
    &h_2(\langle x_2 - y_2 \rangle) \nonumber\\ = \: &(h_1(\langle x_1 \rangle) \sqcap h_2(\langle x_2 - y_2 \rangle)) \sqcup (h_1(\langle x_1 \rangle)^{\perp} \sqcap h_2(\langle x_2 - y_2 \rangle)).
\end{align}
Damit gilt 
\begin{align}
    K_{y_2,x_2}(x) = v + w,
\end{align}
mit
\begin{align}
    v \in& \: h_1(\langle x_1 \rangle) \sqcap h_2(\langle y_2 \rangle), \\ w \in& \: h_1(\langle x_1 \rangle)^{\perp} \sqcap h_2(\langle y_2 \rangle)
\end{align}
und 
\begin{align}
    u = z + d
\end{align}
mit 
\begin{align}
    z \in& \: h_1(\langle x_1 \rangle) \sqcap h_2(\langle x_2 - y_2 \rangle), \\ d \in& \: h_1(\langle x_1 \rangle)^{\perp} \sqcap h_2(\langle x_2 - y_2 \rangle). 
\end{align}
Es gilt nun 
\begin{align}
    x = v + w + z + d \: \: \Leftrightarrow \: \: x - v - z = w + d
\end{align}
Wegen $w + d \in (h_1(\langle x_1 \rangle)^{\perp} \sqcap h_2(\langle y_2 \rangle)) \sqcup (h_1(\langle x_1 \rangle)^{\perp} \sqcap h_2(\langle x_2 - y_2 \rangle))$ und \\ $x - v - z \in (h_1(\langle x_1 \rangle) \sqcap h_2(\langle x_2 \rangle)) \sqcup (h_1(\langle x_1 \rangle) \sqcap h_2(\langle y_2 \rangle)) \sqcup (h_1(\langle x_1 \rangle) \sqcap h_2(\langle x_2 - y_2 \rangle))$ folgt $w + d \in h_1(\langle x_1 \rangle) \sqcap h_1(\langle x_1 \rangle)^{\perp}$ und damit $w + d = 0$. Weiter gilt, dass $h_2(\langle x_2 - y_2 \rangle \sqcap \langle y_2 \rangle) = h_2(\langle x_2 - y_2 \rangle) \sqcap h_2(\langle y_2 \rangle) = 0$ ist, da $x_2$ und $y_2$ linear unabhängig in $\mathcal{H}_2$ sind. Da aus obigem ebenfalls folgt, dass die Summe $w + d \in h_2(\langle x_2 - y_2 \rangle) \sqcup h_2(\langle y_2 \rangle)$ ist, folgt zusammen mit $w + d = 0$, dass $w = d = 0$ ist. Daraus folgt nun aber, dass 
\begin{align}
    K_{y_2,x_2}(x) = v \in h_1(\langle x_1 \rangle) \sqcap h_2(\langle y_2 \rangle)
\end{align}
ist. Auf analogem Wege, mittels der eindeutigen Darstellung $x = F_{y_1,x_1}(x) + e$ mit $e \in h_1(\langle x_1 - y_1 \rangle)$, zeigt man, dass $F_{y_1,x_1}(x) \in h_1(\langle y_1 \rangle) \sqcap h_2(\langle x_2 \rangle)$ ist. Um nun die letzte Gleichung des Lemmas zu zeigen, betrachten wir wieder die obigen Darstellungen von $x \in h_1(\langle x_1 \rangle) \sqcap h_2(\langle x_2 \rangle)$: 
\begin{align}
    x = K_{y_2,x_2}(x) + u \Longleftrightarrow K_{y_2,x_2}(x) = x - u
\end{align}
bzw. 
\begin{align}
    x = F_{y_1,x_1}(x) + e \Longleftrightarrow F_{y_1,x_1}(x) = x - e.
\end{align}
Zuerst bemerken wir, dass 
\begin{align}
    \tilde{K}_{y_2,x_2}(x) = x - u \label{eq:K}
\end{align}
und
\begin{align}
    \tilde{F}_{y_1,x_1}(x) = x - e, \label{eq:F}
\end{align}
da $x \in h_1(\langle x_1 \rangle) \sqcap h_2(\langle x_2 \rangle)$ ist. Wir wenden nun auf den Ausdruck \eqref{eq:K} den Operator $\tilde{F}_{y_1,x_1}$ an:
\begin{align}
    (\tilde{F}_{y_1,x_1} \circ \tilde{K}_{y_2,x_2})(x) =& \: \tilde{F}_{y_1,x_1}(\tilde{K}_{y_2,x_2}(x)) \nonumber\\ =& \: \tilde{F}_{y_1,x_1}(K_{y_2,x_2}(x)) \nonumber\\ =& \: \tilde{F}_{y_1,x_1}(x - u) \nonumber\\ =& \: \tilde{F}_{y_1,x_1}(x) - \tilde{F}_{y_1,x_1}(u) \nonumber\\ =& \: F_{y_1,x_1}(x) - u \nonumber\\ =& \: x - e - u 
\end{align}
Anderseits gilt durch Anwendung von $\tilde{K}_{y_2,x_2}$ auf den Ausdruck \eqref{eq:F}:
\begin{align}
    (\tilde{K}_{y_2,x_2} \circ \tilde{F}_{y_1,x_1})(x) =& \: \tilde{K}_{y_2,x_2}(\tilde{F}_{y_1,x_1}(x)) \nonumber\\ =& \: \tilde{K}_{y_2,x_2}(x - e) \nonumber\\ =& \: \tilde{K}_{y_2,x_2}(x) - \tilde{K}_{y_2,x_2}(e) \nonumber\\ =& \: K_{y_2,x_2}(x) - e \nonumber\\ =& \: x - u - e
\end{align}
Also gilt
\begin{align}
    & \: (\tilde{K}_{y_2,x_2} \circ \tilde{F}_{y_1,x_1})(x) = (\tilde{F}_{y_1,x_1} \circ \tilde{K}_{y_2,x_2})(x) \nonumber\\ \Leftrightarrow& \: (\tilde{K}_{y_2,x_2} \circ F_{y_1,x_1})(x) = (\tilde{F}_{y_1,x_1} \circ K_{y_2,x_2})(x)
\end{align}
Da nach obigem $ F_{y_1,x_1}(x) \in h_1(\langle y_1 \rangle) \sqcap h_2(\langle x_2 \rangle)$ und \\ $K_{y_2,x_2}(x) \in h_1(\langle x_1 \rangle) \sqcap h_2(\langle y_2 \rangle)$ ist, folgt
\begin{align}
    (K_{y_2,x_2} \circ F_{y_1,x_1})(x) = (F_{y_1,x_1} \circ K_{y_2,x_2})(x)
\end{align}
\end{proof}

Die Aussage aus Lemma 5.6 lässt sich nun noch verallgemeinern: Seien dazu $0 \neq x_1,y_1 \in \mathcal{H}_1$ und $0 \neq x_2,y_2 \in \mathcal{H}_2$ jeweils linear unabhängig und für $y_1$ gelte, dass $y_1 = \tilde{y} + \lambda x_1$ ist. Sei $x \in h_1(\langle x_1 \rangle) \sqcap h_2(\langle x_2 \rangle)$. Mittels der Eigenschaften von $F_{y_1,x_1}$ und $K_{y_2,x_2}$, sowie Lemma 5.6, folgt
\begin{align}
    &(F_{y_1,x_1} \circ K_{y_2,x_2})(x) = (K_{y_2,x_2} \circ F_{y_1,x_1})(x) \nonumber\\ \Leftrightarrow \: &(F_{\tilde{y}+\lambda x_1,x_1} \circ K_{y_2,x_2})(x) = (K_{y_2,x_2} \circ F_{\tilde{y}+\lambda x_1,x_1})(x) \nonumber\\ \Leftrightarrow \: &(F_{\tilde{y},x_1} \circ K_{y_2,x_2} + F_{\lambda x_1,x_1} \circ  K_{y_2,x_2})(x) = (K_{y_2,x_2} \circ F_{\tilde{y},x_1} + K_{y_2,x_2} \circ F_{\lambda x_1,x_1})(x) 
\end{align}
Aus Lemma 5.6 folgt, dass $(F_{\tilde{y},x_1} \circ K_{y_2,x_2})(x) = (K_{y_2,x_2} \circ F_{\tilde{y},x_1})(x)$ ist. Daraus folgt nun, dass 
\begin{align}
    (F_{\lambda x_1,x_1} \circ  K_{y_2,x_2})(x) = (K_{y_2,x_2} \circ F_{\lambda x_1,x_1})(x).
\end{align}
Analog kann man $(F_{y_1,x_1} \circ K_{\lambda x_2,x_2})(x) = (K_{\lambda x_2,x_2} \circ F_{y_1,x_1})(x)$ zeigen. Diese Beobachtung werden wir in den noch kommenden Lemmata nutzen.

Als nächstes werden wir zeigen, dass es sich bei den Abbildungen $h_1$ und $h_2$ um nicht gemischte m-Morphismen handelt. Dieser Fakt wird in darauffolgenden Lemmata noch wichtig werden.

\begin{lemma}
Die Abbildungen $h_1$ und $h_2$ sind keine gemischten m-Morphismen.
\end{lemma}

\begin{proof}
Sei $x_1 \in \mathcal{H}_1$ und $\lambda \in \mathbb{C}$. Wir wissen nach Theorem 5.2, dass gilt 
\begin{align}
    &F_{\lambda x_1,x_1} = \lambda P_1^{\langle x_1 \rangle} + \lambda^* P_2^{\langle x_1 \rangle}, \\ &P_1^{\langle x_1 \rangle} \circ P_2^{\langle x_1 \rangle} = 0.
\end{align}
Per Definition folgt, dass wenn $h_1$ nicht gemischt ist, dass entweder $P_1^{\langle x_1 \rangle} = 0$ oder $P_2^{\langle x_1 \rangle} = 0$ ist. Wir nehmen nun an, $h_1$ sei gemischt, d.h. 
\begin{align}
    P_1^{\langle x_1 \rangle}[h_1(\langle x_1 \rangle)] \neq 0 \: \: \: \: \textrm{und} \: \: \: \: P_2^{\langle x_2 \rangle}[h_1(\langle x_1 \rangle)] \neq 0.
\end{align}
Sei $z \in P_1^{\langle x_1 \rangle}[h_1(\langle x_1 \rangle)]$ und $d \in P_2^{\langle x_2 \rangle}[h_1(\langle x_1 \rangle)]$. Aus Lemma 5.4 ist bekannt, dass die Abbildung $v_{\langle x_1 \rangle} : \mathcal{L}(\mathcal{H}_2) \rightarrow \mathcal{L}(h_1(\langle x_1 \rangle))$ surjektiv ist. Damit existieren Elemente $z_2, d_2 \in \mathcal{H}_2$, sodass 
\begin{align}
    &v_{\langle x_1 \rangle}(\langle z_2 \rangle) = \langle z \rangle = h_1(\langle x_1 \rangle) \sqcap h_2(\langle z_2 \rangle), \\ &v_{\langle x_1 \rangle}(\langle d_2 \rangle) = \langle d \rangle = h_1(\langle x_1 \rangle) \sqcap h_2(\langle d_2 \rangle)
\end{align}
bzw.
\begin{align}
    &z \in h_1(\langle x_1 \rangle) \sqcap h_2(\langle z_2 \rangle), \\ &d \in h_1(\langle x_1 \rangle) \sqcap h_2(\langle d_2 \rangle)
\end{align}
gilt. Wir betrachten nun die Abbildung $K_{d_2,z_2} : h_2(\langle z_2 \rangle) \rightarrow h_2(\langle d_2 \rangle)$. Dann sei $K_{d_2,z_2}(z) = s \in h_2(\langle d_2 \rangle)$ und wegen $d \in h_1(\langle x_1 \rangle) \sqcap h_2(\langle d_2 \rangle)$ insbesondere $s \in \langle d \rangle$. Dann gilt 
\begin{align}
    K_{d_2,z_2}(F_{\lambda x_1,x_1}(z)) =& \: K_{d_2,z_2}(\lambda P_1^{\langle x_1 \rangle}(z) + \lambda^* P_2^{\langle x_1 \rangle}(z)) \nonumber\\ =& \: K_{d_2,z_2}(\lambda z) \nonumber\\ =& \: \lambda s.
\end{align}
Dabei haben wir ausgenutzt, dass aus $z \in P_1^{\langle x_1 \rangle}[h_1(\langle x_1 \rangle)]$ folgt, dass $P_1^{\langle x_1 \rangle}(z) = z$ ist, und das die Projektoren $P_1^{\langle x_1 \rangle}$ und $P_2^{\langle x_1 \rangle}$ orthogonal zueinander sind. 

Andererseits gilt wegen der Kommutativität von $K_{d_2,z_2}$ und $F_{\lambda x_1, x_1}$ (Lemma 5.6 und Bemerkung 5.1): 
\begin{align}
    F_{\lambda x_1,x_1}(K_{d_2,z_2}(z)) =& \: F_{\lambda x_1, x_1}(s) \nonumber\\ =& \: \lambda P_1^{\langle x_1 \rangle}(s) + \lambda^* P_2^{\langle x_1 \rangle}(s) \nonumber\\ =& \: \lambda^* s,
\end{align}
da $s \in \langle d \rangle$ und damit auch $s \in P_2^{\langle x_2 \rangle}[h_1(\langle x_1 \rangle)]$ und $P_1^{\langle x_1 \rangle}$ sowie $P_2^{\langle x_1 \rangle}$ orthogonale Projektoren sind. Nun gilt aber im Allgemeinen $\lambda s \neq \lambda^* s$ für nicht verschwindenden Imaginärteil von $\lambda$, d.h $\mathbf{Im}(\lambda) \neq 0$. Aus diesem Wiederspruch folgt nun aber, dass $h_1$ nicht gemischt sein kann. Demzufolge ist $h_1$ entweder linear oder antilinear. Auf analogem Weg beweist man, dass $h_2$ nicht gemischt sein kann.
\end{proof}

Für die kommenden Lemmata benötigen wir nun noch den Begriff der unitären und antiunitären Vektorraum-Homomorphismen. Dabei handelt es sich um spezielle Abbildungen, die die von den Skalarprodukten induzierten Normen erhalten. Da eine Norm einen Abstandsbegriff definiert, handelt es sich also bei diesen Abbildungen um abstandserhaltende Abbildungen. 

\begin{definition}
Sei $f: \mathcal{H}_1 \rightarrow \mathcal{H}_2$ eine Abbildung zwischen zwei komplexen Hilberträumen. Die Abbildung $f$ heißt unitärer Vektorraum-Homomorphismus, falls $f$ linear und surjektiv ist und gilt, dass 
\begin{align}
    \|f(v)\|_{\mathcal{H}_2} = \|v\|_{\mathcal{H}_1}
\end{align}
für alle $v \in \mathcal{H}_1$, wobei $\|\cdot\|_{\mathcal{H}_1}$ und $\|\cdot\|_{\mathcal{H}_2}$ die von den Skalarprodukten auf $\mathcal{H}_1$ und $\mathcal{H}_2$ induzierten Normen sind. Die Abbildung $f$ heißt antiunitär, falls $f$ antilinear und surjektiv ist und gilt, dass 
\begin{align}
    \|f(v)\|_{\mathcal{H}_2} = \|v\|_{\mathcal{H}_1}
\end{align}
für alle $v \in \mathcal{H}_1$. Dabei heißt antilinear (als Vektorraum-Homomorphismus), dass für alle $v,w \in \mathcal{H}_1$ und alle $\lambda \in \mathbb{C}$ gilt, dass
\begin{align}
    f(\lambda v + w) = \lambda^* f(v) + f(w)
\end{align}
ist. 
\end{definition}

Im folgenden wollen wir solche unitären und antiunitären Abbildungen konstruieren. Der Grund dafür liegt in den Eigenschaften derartiger Abbildungen. Diese sind nämlich nicht nur surjektiv, sondern auch injektiv, denn wegen der positiven Definitheit des Skalarproduktes folgt aus
\begin{align}
    \|f(v)\|_{\mathcal{H}_2} = \|v\|_{\mathcal{H}_1} = 0
\end{align}
das $v = 0$ ist, d.h. der Kern von $f$ ist trivial (enthält also nur den Nullvektor), woraus die Injektivität folgt \cite{fischer2011lernbuch}. Damit sind solche $f$ im wesentlichen strukturerhaltende bijektive Abbildungen. 

Die unitären und antiunitären Abbildungen, die wir im folgenden Lemma konstruieren, eignen sich daher aufgrund dieser Eigenschaft zur Konstruktion von Isomorphien. Unser Ziel wird es daher sein, mittels dieser Abbildungen später die Isomorphie zwischen $\mathcal{L}(\mathcal{H})$ und $\mathcal{L}(\mathcal{H}_1 \otimes \mathcal{H}_2)$ zu konstruieren. 

Im Folgenden bezeichnen die immer wieder auftauchenden Mengen $I, J, K, L \subseteq \mathbb{N}$ (endliche oder unendliche) Indexmengen. 

\begin{lemma}
Seien $z_1,x_1 \in \mathcal{H}_1$, $z_2,x_2 \in \mathcal{H}_2$ und $z \in h_1(\langle z_1 \rangle) \sqcap h_2(\langle z_2 \rangle)$, sodass $z_1,z_2$ und $z$ jeweils nicht den Nullvektoren in ihren jeweiligen Räumen entsprechen. Wir definieren die Konstante
\begin{align}
    \alpha = \frac{\|z_1\|_{\mathcal{H}_1} \|z_2\|_{\mathcal{H}_2}}{\|z\|_\mathcal{H}},
\end{align}
wobei $\|\cdot\|_{\mathcal{H}_1}, \|\cdot\|_{\mathcal{H}_2}$ und $\|\cdot\|_\mathcal{H}$ jeweils die von den in $\mathcal{H}_1, \mathcal{H}_2$ und $\mathcal{H}$ definierten Skalarprodukten induzierten Normen sind. Wir definieren weiter die Abbildungen 
\begin{align}
     U_{\langle x_2 \rangle} : \mathcal{H}_1 \longrightarrow& \: \: h_2(\langle x_2 \rangle) \nonumber\\
     x_1 \longmapsto& \: \: \frac{\alpha}{\|x_2\|_{\mathcal{H}_2}} (F_{x_1,z_1} \circ K_{x_2,z_2})(z)
\end{align}
und 
\begin{align}
     V_{\langle x_1 \rangle} : \mathcal{H}_2 \longrightarrow& \: \: h_1(\langle x_1 \rangle) \nonumber\\
     x_2 \longmapsto& \: \: \frac{\alpha}{\|x_1\|_{\mathcal{H}_1}} (K_{x_2,z_2} \circ F_{x_1,z_1})(z).
\end{align}
Dann sind alle $U_{\langle x_2 \rangle}$ unitäre oder antiunitäre Abbildungen, je nachdem ob $h_1$ linear oder antilinear ist. Weiter generiert $U_{\langle x_2 \rangle}$ die Abbildung $u_{\langle x_2 \rangle}$. Analoges gilt für $V_{\langle x_1 \rangle}$. 
\end{lemma}

\begin{proof}
Wir bemerken zuerst, dass die eben definierten Abbildungen nach Lemma 5.6 wohldefiniert sind. Wir zeigen nun zuerst, das $U_{\langle x_2 \rangle}$ linear ist, wenn $h_1$ ein linearer m-Morphismus ist, und das $U_{\langle x_2 \rangle}$ antilinear ist, wenn $h_1$ ein antilinearer m-Morphismus ist. Sei also zunächst angenommen, dass $h_1$ ein linearer m-Morphismus ist, d.h. es gilt 
\begin{align}
    F_{\lambda x_1,x_1} = \lambda id_{h_1(\langle x_1 \rangle)}.
\end{align}
Wir stellen zuerst fest, dass aufgrund der Eigenschaften von $F_{x_1,z_1}$, unabhängig von der Linearität oder Antilinearität von $h_1$, trivialerweise die Additivität folgt, denn seien $x_1, \tilde{x}_1 \in \mathcal{H}_1$, so gilt 
\begin{align}
    U_{\langle x_2 \rangle}(x_1 + \tilde{x}_1) =& \: \frac{\alpha}{\|x_2\|_{\mathcal{H}_2}} (F_{x_1 + \tilde{x}_1,z_1} \circ K_{x_2,z_2})(z) \nonumber\\ =& \: \frac{\alpha}{\|x_2\|_{\mathcal{H}_2}} (F_{x_1,z_1} \circ K_{x_2,z_2} + F_{\tilde{x}_1,z_1} \circ K_{x_2,z_2})(z) \nonumber\\ =& \: \frac{\alpha}{\|x_2\|_{\mathcal{H}_2}} (F_{x_1,z_1} \circ K_{x_2,z_2})(z) + \frac{\alpha}{\|x_2\|_{\mathcal{H}_2}} (F_{\tilde{x}_1,z_1} \circ K_{x_2,z_2})(z) \nonumber\\ =& \: U_{\langle x_2 \rangle}(x_1) + U_{\langle x_2 \rangle}(\tilde{x}_1)
\end{align}
Sei $\lambda \in \mathbb{C}$. Wegen 
\begin{align}
    F_{\lambda x_1,z_1} =& \: F_{\lambda x_1,x_1} \circ F_{x_1,z_1} \nonumber\\ =& \: \lambda id_{h_1(\langle x_1 \rangle)} \circ F_{x_1,z_1} \nonumber\\ =& \: \lambda F_{x_1,z_1}
\end{align}
folgt die Homogenität:
\begin{align}
    U_{\langle x_2 \rangle}(\lambda x_1) =& \: \frac{\alpha}{\|x_2\|_{\mathcal{H}_2}} (F_{\lambda x_1,z_1} \circ K_{x_2,z_2})(z) \nonumber\\ =& \: \lambda \frac{\alpha}{\|x_2\|_{\mathcal{H}_2}} (F_{x_1,z_1} \circ K_{x_2,z_2})(z) \nonumber\\ =& \: \lambda U_{\langle x_2 \rangle}(x_1)
\end{align}
Sei nun $h_1$ ein antilinearer m-Morphismus, dann gilt 
\begin{align}
    F_{\lambda x_1,x_1} = \lambda^* id_{h_1(\langle x_1 \rangle)}.
\end{align}
Es gilt analog zu oben 
\begin{align}
    F_{\lambda x_1,z_1} =& \: F_{\lambda x_1,x_1} \circ F_{x_1,z_1} \nonumber\\ =& \: \lambda^* id_{h_1(\langle x_1 \rangle)} \circ F_{x_1,z_1} \nonumber\\ =& \: \lambda^* F_{x_1,z_1}
\end{align}
und damit
\begin{align}
     U_{\langle x_2 \rangle}(\lambda x_1) =& \: \frac{\alpha}{\|x_2\|_{\mathcal{H}_2}} (F_{\lambda x_1,z_1} \circ K_{x_2,z_2})(z) \nonumber\\ =& \: \lambda^* \frac{\alpha}{\|x_2\|_{\mathcal{H}_2}} (F_{x_1,z_1} \circ K_{x_2,z_2})(z) \nonumber\\ =& \: \lambda^* U_{\langle x_2 \rangle}(x_1)
\end{align}
Analog zeigt man, dass $V_{\langle x_1 \rangle}$ linear ist, falls $h_2$ ein linearer m-Morphismus ist, und das $V_{\langle x_1 \rangle}$ antilinear ist, falls $h_2$ ein antilinearer m-Morphismus ist. Seien nun $d_1 \in \langle x_1 \rangle$ und $d_2 \in \langle x_2 \rangle$ mit $\|d_1\|_{\mathcal{H}_1} = \|z_1\|_{\mathcal{H}_1}$ und $\|d_2\|_{\mathcal{H}_2} = \|z_2\|_{\mathcal{H}_2}$ gegeben. Wegen $d_1 \in \langle x_1 \rangle$ folgt, dass ein $\lambda \in \mathbb{C}$ existiert, sodass $\lambda d_1 = x_1$ ist. Es gilt 
\begin{align}
    \|x_1\|_{\mathcal{H}_1} =& \: \|\lambda d_1\|_{\mathcal{H}_1} = |\lambda|\|d_1\|_{\mathcal{H}_1} \nonumber\\ \Rightarrow& \: |\lambda| = \frac{\|x_1\|_{\mathcal{H}_1}}{\|d_1\|_{\mathcal{H}_1}} = \frac{\|x_1\|_{\mathcal{H}_1}}{\|z_1\|_{\mathcal{H}_1}}
\end{align}
Analog existiert ein $\nu \in \mathbb{C}$ mit $\nu d_2 = x_2$, sodass $|\nu| = \frac{\|x_2\|_{\mathcal{H}_2}}{\|z_2\|_{\mathcal{H}_2}}$ ist. Damit folgt
\begin{align}
     \|U_{\langle x_2 \rangle}(x_1)\|_{\mathcal{H}} =& \: \frac{\alpha}{\|x_2\|_{\mathcal{H}_2}} \|(F_{x_1,z_1} \circ K_{x_2,z_2})(z)\|_{\mathcal{H}} \nonumber\\ =& \: \frac{\alpha}{\|x_2\|_{\mathcal{H}_2}} \frac{\|x_1\|_{\mathcal{H}_1}}{\|z_1\|_{\mathcal{H}_1}} \frac{\|x_2\|_{\mathcal{H}_2}}{\|z_2\|_{\mathcal{H}_2}} \|(F_{d_1,z_1} \circ K_{d_2,z_2})(z)\|_\mathcal{H} \nonumber\\ =& \: \|x_1\|_{\mathcal{H}_1}
\end{align}
Dabei haben wir ausgenutzt, dass $h_1$ und $h_2$ lineare oder antilineare m-Morphismen sind, dass $\alpha = (\|z_1\|_{\mathcal{H}_1} \|z_2\|_{\mathcal{H}_2})/\|z\|_\mathcal{H}$ ist und das nach Theorem 5.2 gilt, dass für $\|d_1\|_{\mathcal{H}_1} = \|z_1\|_{\mathcal{H}_1}$ und $\|d_2\|_{\mathcal{H}_2} = \|z_2\|_{\mathcal{H}_2}$ die Abbildungen $F_{d_1,z_1}$ und $K_{d_2,z_2}$ Isometrien sind. Analog zeigt man $\|V_{\langle x_1 \rangle}(x_2)\|_\mathcal{H} = \|x_2\|_{\mathcal{H}_2}$.
Aus Lemma 5.6 folgt weiter 
\begin{align}
    \langle U_{\langle x_2 \rangle}(x_1) \rangle = h_1(\langle x_1 \rangle) \sqcap h_2(\langle x_2 \rangle) = \langle V_{\langle x_1 \rangle}(x_2) \rangle
\end{align}
Wir definieren nun die Abbildungen 
\begin{align}
     \tilde{u}_{\langle x_2 \rangle} : \mathcal{L}(\mathcal{H}_1) \longrightarrow& \: \: \mathcal{L}(h_2(\langle x_2 \rangle)) \nonumber\\
    G_1 \longmapsto& \: \: \{U_{\langle x_2 \rangle}(x_1) : x_1 \in G_1\}
\end{align}
und 
\begin{align}
     \tilde{v}_{\langle x_1 \rangle} : \mathcal{L}(\mathcal{H}_2) \longrightarrow& \: \: \mathcal{L}(h_1(\langle x_1 \rangle)) \nonumber\\
    G_2 \longmapsto& \: \: \{V_{\langle x_1 \rangle}(x_2) : x_2 \in G_2\}.
\end{align}
Es lässt sich zeigen, dass es sich bei $\tilde{u}_{\langle x_2 \rangle}$ und $\tilde{v}_{\langle x_1 \rangle}$ um c-Morphismen handelt, die Atome auf Atome abbilden. Letzters folgt aus der Linearität oder Antlinearität der Abbildungen $U_{\langle x_2 \rangle}$ und $V_{\langle x_1 \rangle}$ und der Tatsache , das Isometrien injektiv sind, d.h. der Kern ist trivial. Sei nun $I$ eine Indexmenge und die Menge $\{e_i \in G_1 : i \in I\}$ eine Orthonormalbasis von $G_1$, d.h. $span_\mathbb{C}(\{e_i \in G_1 : i \in I\}) = G_1$ und $\langle e_i , e_j \rangle_{\mathcal{H}_1} = \delta_{i,j}$ mit $\delta_{i,i} = 1$ und $\delta_{i,j} = 0$ für $i \neq j$. Dann definiert die Menge $\{U_{\langle x_2 \rangle}(e_i) : i \in I\}$ eine Basis in $\tilde{u}_{\langle x_2 \rangle}[G_1]$ und es gilt
\begin{align}
    u_{\langle x_2 \rangle}(G_1) =& \: u_{\langle x_2 \rangle}(\sqcup_{i \in I} \langle e_i \rangle) \nonumber\\ =& \: \sqcup_{i \in I} u_{\langle x_2 \rangle} (\langle e_i \rangle) \nonumber\\ =& \: \sqcup_{i \in I} (h_1(\langle e_i \rangle) \sqcap h_2(\langle x_2 \rangle)) \nonumber\\ =& \: \sqcup_{i \in I} \langle U_{\langle x_2 \rangle} (e_i) \rangle \nonumber\\ =& \: \tilde{u}_{\langle x_2 \rangle}(G_1)
\end{align}
Daraus folgt, dass die Abbildung $U_{\langle x_2 \rangle}$ die Abbildung $ u_{\langle x_2 \rangle}$ erzeugt und das $u_{\langle x_2 \rangle} = \tilde{u}_{\langle x_2 \rangle}$ ist, womit $\tilde{u}_{\langle x_2 \rangle}$ ein unitärer c-Morphismus ist und damit ein Isomorphismus. Aus der Surjektivität von $\tilde{u}_{\langle x_2 \rangle}$ folgt damit die Surjektivität von $U_{\langle x_2 \rangle}$, womit $U_{\langle x_2 \rangle}$ ein unitärer Vektorraum-Homomorphismus ist, sofern $h_1$ einen linearen m-Morphismus darstellt, und $U_{\langle x_2 \rangle}$ ist antiunitär, sofern $h_1$ ein antilinearer m-Morphismus ist. Analog argumentiert man für $V_{\langle x_1 \rangle}$.
\end{proof}

Mittels der $U_{\langle x_2 \rangle}$ und $V_{\langle x_1 \rangle}$ können wir, wie wir im Folgenden sehen werden, Orthonormalbasen in $\mathcal{H}$ erklären. Dies wird uns später beispielsweise dabei helfen, lineare und antilineare Bijektionen zwischen $\mathcal{H}$ und $\mathcal{H}_1 \otimes \mathcal{H}_2$ zu erklären, die wir dann wiederrum nutzen können, um Isomorphien zwischen $\mathcal{L}(\mathcal{H})$ und $\mathcal{L}(\mathcal{H}_1 \otimes \mathcal{H}_2)$ zu konstruieren. 

\begin{lemma}
Seien $x_1 \in \mathcal{H}_1$ und $x_2 \in \mathcal{H}_2$ mit $\|x_1\|_{\mathcal{H}_1} = \|x_2\|_{\mathcal{H}_2} > 0$ gegeben. Dann gilt, dass 
\begin{align}
    U_{\langle x_2 \rangle}(x_1) = V_{\langle x_1 \rangle}(x_2)
\end{align}
ist. Seien $I$ und $J$ Indexmengen und $\{e_i \in \mathcal{H}_1 : i \in I\}$ eine Orthonormalbasis von $\mathcal{H}_1$ und $\{f_j \in \mathcal{H}_2 : j \in J\}$ eine Orthonormalbasis von $\mathcal{H}_2$, so folgt, dass $\{U_{\langle f_j \rangle}(e_i) \in \mathcal{H} : i \in I, \: j \in J\}$ eine Orthonormalbasis von $\mathcal{H}$ ist. 
\end{lemma}

\begin{proof}
Die Identität 
\begin{align}
    U_{\langle x_2 \rangle}(x_1) = V_{\langle x_1 \rangle}(x_2)
\end{align}
ist eine direkte Konsequenz aus der Definition der beiden Abbildungen mit $\|x_1\|_{\mathcal{H}_1} = \|x_2\|_{\mathcal{H}_2} \neq 0$ (gemäß Lemma 5.8) und Lemma 5.6:
\begin{align}
    U_{\langle x_2 \rangle}(x_1) =& \: \frac{\alpha}{\|x_2\|_{\mathcal{H}_2}} (F_{x_1,z_1} \circ K_{x_2,z_2})(z) \nonumber\\ =& \: \frac{\alpha}{\|x_1\|_{\mathcal{H}_1}} (F_{x_1,z_1} \circ K_{x_2,z_2})(z) \nonumber\\ =& \: \frac{\alpha}{\|x_1\|_{\mathcal{H}_1}} (K_{x_2,z_2} \circ F_{x_1,z_1})(z) \nonumber\\ =& \: V_{\langle x_1 \rangle}(x_2)
\end{align}
Zeigen noch, dass $\{U_{\langle f_j \rangle}(e_i) \in \mathcal{H} : i \in I, \: j \in J\}$ eine Orthonormalbasis (ONB) von $\mathcal{H}$ ist: \\ \\ Sei im folgendem $j,j' \in J$ mit $j \neq j'$. Da $\langle f_j, f_{j'} \rangle_{\mathcal{H}_2} = 0$ oder anders ausgedrückt $f_j \perp f_{j'}$ ist, folgt, dass $\langle f_j \rangle \perp \langle f_{j'} \rangle$ und damit auch $\langle f_j \rangle \sqcap \langle f_{j'} \rangle = 0$ ist. Es gilt $0 = h_2(\langle f_j \rangle \sqcap \langle f_{j'} \rangle) = h_2(\langle f_j \rangle) \sqcap h_2(\langle f_{j'} \rangle)$. Wegen $\langle f_j \rangle \subseteq \langle f_{j'} \rangle^\perp$ folgt $h_2(\langle f_j \rangle) \subseteq h_2(\langle f_{j'} \rangle^\perp) = h_2(\langle f_{j'} \rangle)^\perp$ und damit $h_2(\langle f_j \rangle) \perp h_2(\langle f_{j'} \rangle)$. Daraus folgt unmittelbar, dass $U_{\langle f_j \rangle}(x_1) \perp U_{\langle f_{j'} \rangle}(x_1)$ für alle $x_1 \in \mathcal{H}_1$ ist, denn $U_{\langle f_j \rangle}(x_1) \in h_2(\langle f_j \rangle)$ und $U_{\langle f_{j'} \rangle}(x_1) \in h_2(\langle f_{j'} \rangle)$. Es gilt damit 
\begin{align}
    \langle U_{\langle f_j\rangle}(x_1),U_{\langle f_{j'} \rangle}(x_1) \rangle_\mathcal{H} =& \: \langle U_{\langle f_j\rangle}(\sum_{i \in I}a^i_1 e_i),U_{\langle f_{j'} \rangle}(\sum_{i' \in I} a^{i'}_1 e_{i'}) \rangle_\mathcal{H} \nonumber\\ =& \: \sum_{i,i' \in I} (a^{i}_1)^* a^{i'}_1 \langle U_{\langle f_j \rangle}(e_i),U_{\langle f_{j'} \rangle}(e_{i'}) \rangle_\mathcal{H} \nonumber\\ =& \: 0.
\end{align}
Wir zeigen mittels Induktion über die Anzahl $n = |\{i_1,...,i_n : i_1,...,i_n \in I\}|$ der nicht verschwindenden Komponenten des Vektors $x_1 = \sum_{i \in I} a^{i}_1e_i$, dass für $j \neq j^{'}$ folgt, dass $U_{\langle f_j \rangle}(e_i) \perp U_{\langle f_{j'} \rangle}(e_{i'})$ ist:
\begin{description}
\item Induktionsanfang: \\
Da n = 1 trivial ist, wählen wir n = 2. Es gilt 
\begin{align}
   \langle U_{\langle f_j \rangle}(x_1),U_{\langle f_{j'} \rangle}(x_1) \rangle_\mathcal{H} =& \: \langle U_{\langle f_j \rangle}(a^{i_1}_1 e_{i_1} + a^{i_2}_1 e_{i_2}),U_{\langle f_{j'} \rangle}(a^{i_1}_1 e_{i_1} + a^{i_2}_1 e_{i_2}) \rangle_\mathcal{H} \nonumber\\ =& \: (a^{i_1}_1)^*a^{i_1}_1 \langle U_{\langle f_j \rangle}(e_{i_1}),U_{\langle f_{j'} \rangle}(e_{i_1}) \rangle_\mathcal{H} \nonumber\\ &+ \: (a^{i_1}_1)^*a^{i_2}_1 \langle U_{\langle f_j \rangle}(e_{i_1}),U_{\langle f_{j'} \rangle}(e_{i_2}) \rangle_\mathcal{H} \nonumber\\ &+ \: (a^{i_2}_1)^*a^{i_1}_1 \langle U_{\langle f_j \rangle}(e_{i_2}),U_{\langle f_{j'} \rangle}(e_{i_1}) \rangle_\mathcal{H} \nonumber\\ &+ \: (a^{i_2}_1)^*a^{i_2}_1 \langle U_{\langle f_j \rangle}(e_{i_2}),U_{\langle f_{j'} \rangle}(e_{i_2}) \rangle_\mathcal{H} \nonumber\\ =& \: (a^{i_1}_1)^*a^{i_2}_1 \langle U_{\langle f_j \rangle}(e_{i_1}),U_{\langle f_{j'} \rangle}(e_{i_2}) \rangle_\mathcal{H} \nonumber\\ &+ \: (a^{i_2}_1)^*a^{i_1}_1 \langle U_{\langle f_j \rangle}(e_{i_2}),U_{\langle f_{j'} \rangle}(e_{i_1}) \rangle_\mathcal{H} \nonumber\\ =& \: 0
\end{align}
Wählen wir feste $a^{i_1}_1,a^{i_2}_1 \in \mathbb{R}$, so folgt 
\begin{align}
    &\langle U_{\langle f_j \rangle}(e_{i_1}),U_{\langle f_{j'} \rangle}(e_{i_2}) \rangle_\mathcal{H} = \alpha, \\
    &\langle U_{\langle f_j \rangle}(e_{i_2}),U_{\langle f_{j'} \rangle}(e_{i_1}) \rangle_\mathcal{H} = -\alpha,
\end{align}
da $(a^{i_1}_1)^*a^{i_2}_1 = (a^{i_2}_1)^*a^{i_1}_1$ ist. Wählen wir nun aber feste $a^{i_1}_1,a^{i_2}_1 \in \mathbb{C}$, sodass $(a^{i_1}_1)^*a^{i_2}_1 \neq (a^{i_2}_1)^*a^{i_1}_1$ ist, folgt
\begin{align}
    &\langle U_{\langle f_j \rangle}(e_{i_1}),U_{\langle f_{j'} \rangle}(e_{i_2}) \rangle_\mathcal{H} = \alpha, \\
    &\langle U_{\langle f_j \rangle}(e_{i_2}),U_{\langle f_{j'} \rangle}(e_{i_1}) \rangle_\mathcal{H} = \beta \neq -\alpha,
\end{align}
falls $\alpha \neq 0$. Somit erhalten wir 
\begin{align}
    \langle U_{\langle f_j \rangle}(e_{i_1}),U_{\langle f_{j'} \rangle}(e_{i_2}) \rangle_\mathcal{H} = 0
\end{align}
\item Induktionsbehauptung:\\ 
Die Aussage sei für ein $n = |\{i_1,...,i_n : i_1,...,i_n \in I\}| \geq 2$ bereits gezeigt, d.h. $U_{\langle f_j \rangle}(e_i) \perp U_{\langle f_{j'} \rangle}(e_{i'})$ für alle $i,i' \in \{i_1,...,i_n : i_1,...,i_n \in I\}$.
\item Induktionsschritt:\\ 
Seien nun $n+1$ Komponenten von $x_1$ nicht Null. Es gilt
\begin{align}
     \langle U_{\langle f_j \rangle}(x_1),U_{\langle f_{j'} \rangle}(x_1) \rangle_\mathcal{H} =& \:  \langle U_{\langle f_j \rangle}(\sum_{k=1}^{n+1} a^{i_k}_1 e_{i_k}),U_{\langle f_{j'} \rangle}(\sum_{k'=1}^{n+1} a^{i_{k'}}_1 e_{i_{k'}}) \rangle_\mathcal{H} \nonumber\\ =& \: \sum_{k,k'=1}^{n+1} (a^{i_k}_1)^*a^{i_{k'}}_1 \langle U_{\langle f_j \rangle}(e_{i_k}),U_{\langle f_{j'} \rangle}(e_{i_{k'}}) \rangle_\mathcal{H} \nonumber\\ =& \:  \sum_{k,k'=1}^{n} (a^{i_k}_1)^*a^{i_{k'}}_1 \langle U_{\langle f_j \rangle}(e_{i_k}),U_{\langle f_{j'} \rangle}(e_{i_{k'}}) \rangle_\mathcal{H} \nonumber\\ &+ \: (a^{i_1}_1)^*a^{i_{n+1}}_1 \langle U_{\langle f_j \rangle}(e_{i_1}),U_{\langle f_{j'} \rangle}(e_{i_{n+1}}) \rangle_\mathcal{H} + ... \nonumber\\ &+ \: (a^{i_n}_1)^*a^{i_{n+1}}_1 \langle U_{\langle f_j \rangle}(e_{i_n}),U_{\langle f_{j'} \rangle}(e_{i_{n+1}}) \rangle_\mathcal{H} \nonumber\\ &+ \: (a^{i_{n+1}}_1)^*a^{i_{1}}_1 \langle U_{\langle f_j \rangle}(e_{i_{n+1}}),U_{\langle f_{j'} \rangle}(e_{i_{1}}) \rangle_\mathcal{H} + ... \nonumber\\ &+ \: (a^{i_{n+1}}_1)^*a^{i_{n}}_1 \langle U_{\langle f_j \rangle}(e_{i_{n+1}}),U_{\langle f_{j'} \rangle}(e_{i_{n}}) \rangle_\mathcal{H} \nonumber\\ =& \: 0
\end{align}
Anwenden der Induktionsbehauptung liefert
\begin{align}
       &(a^{i_1}_1)^*a^{i_{n+1}}_1 \underbrace{\langle U_{\langle f_j \rangle}(e_{i_1}),U_{\langle f_{j'} \rangle}(e_{i_{n+1}}) \rangle_\mathcal{H}}_{\alpha_1} + ... \nonumber\\ +& \: (a^{i_n}_1)^*a^{i_{n+1}}_1 \underbrace{\langle U_{\langle f_j \rangle}(e_{i_n}),U_{\langle f_{j'} \rangle}(e_{i_{n+1}}) \rangle_\mathcal{H}}_{\alpha_n} \nonumber\\ +& \: (a^{i_{n+1}}_1)^*a^{i_{1}}_1 \underbrace{\langle U_{\langle f_j \rangle}(e_{i_{n+1}}),U_{\langle f_{j'} \rangle}(e_{i_{1}}) \rangle_\mathcal{H}}_{\beta_1} + ... \nonumber\\ +& \: (a^{i_{n+1}}_1)^*a^{i_{n}}_1 \underbrace{\langle U_{\langle f_j \rangle}(e_{i_{n+1}}),U_{\langle f_{j'} \rangle}(e_{i_{n}}) \rangle_\mathcal{H}}_{\beta_n} = 0
\end{align}
Wählen wir feste $a^{i_1}_1,...,a^{i_{n+1}}_1 \in \mathbb{R}$ mit $a^{i_l}_1a^{i_{n+1}}_1 \neq a^{i_{l'}}_1a^{i_{n+1}}_1$ und $l,l' \in \{i_1,...,i_n : i_1,...,i_n \in I\}$ mit $l \neq l'$, so folgt 
\begin{align}
    \beta_1 = - \alpha_1, \: ... \: , \: \beta_n = - \alpha_n.
\end{align}
Wählen wir hingegen feste $a^{i_1}_1,...,a^{i_{n+1}}_1 \in \mathbb{C}$ mit $a^{i_l}_1a^{i_{n+1}}_1 \neq a^{i_{l'}}_1a^{i_{n+1}}_1$ für $l,l' \in \{i_1,...,i_n : i_1,...,i_n \in I\}$ und $l \neq l'$, sowie $(a^{i_l}_1)^*a^{i_{n+1}}_1 \neq (a^{i_{n+1}}_1)^*a^{i_{l}}_1$ für alle $l \in \{i_1,...,i_n : i_1,...,i_n \in I\}$, so folgt 
\begin{align}
    \beta_1 \neq - \alpha_1, \: ... \: , \: \beta_n \neq - \alpha_n,
\end{align}
falls $\alpha_1, ..., \alpha_n \neq 0$. Daraus folgt $\alpha_1 = ... = \alpha_n = 0$. Damit folgt für $j \neq j^{'}$, dass $U_{\langle f_j \rangle}(e_i) \perp U_{\langle f_{j'} \rangle}(e_{i'})$ ist.
\end{description}
Da weiter $\langle \cdot , \cdot \rangle_{\mathcal{H}_1}$ und $\langle \cdot , \cdot \rangle_{\mathcal{H}}$ Sesquilinearformen sind, gilt für diese die sogenannte Polarisationsformel (siehe Anhang für den Beweis der Formel und wichtiger Eigenschaften, die im folgenden genutzt werden):
\begin{align}
    &\langle x , y \rangle_{\mathcal{H}_1} = \frac{1}{4}(\|x + y\|^2_{\mathcal{H}_1} - \|x - y\|^2_{\mathcal{H}_1}) - \frac{i}{4}(\|x + iy\|^2_{\mathcal{H}_1} - \|x - iy\|^2_{\mathcal{H}_1}), \\ &\langle x , y \rangle_{\mathcal{H}} = \frac{1}{4}(\|x + y\|^2_{\mathcal{H}} - \|x - y\|^2_{\mathcal{H}}) - \frac{i}{4}(\|x + iy\|^2_{\mathcal{H}} - \|x - iy\|^2_{\mathcal{H}}) 
\end{align}
Da $U_{\langle f_j \rangle}$ ein unitärer oder antiunitärer Operator ist, folgt mittels der Polarisationsformel, dass $U_{\langle f_j \rangle}(e_i) \perp U_{\langle f_{j} \rangle}(e_{i'})$ bzw. $\langle U_{\langle f_j \rangle}(e_i) , U_{\langle f_{j} \rangle}(e_{i'}) \rangle_\mathcal{H} = 0$ für alle $i,i' \in I$ mit $i \neq i'$ ist und für $i = i'$ folgt $\|U_{\langle f_j \rangle}(e_i)\|_\mathcal{H} = 1$, denn \\
\begin{align}
    \langle e_i , e_{i'} \rangle_{\mathcal{H}_1} =& \: \frac{1}{4}(\|e_i + e_{i'}\|^2_{\mathcal{H}_1} - \|e_i - e_{i'}\|^2_{\mathcal{H}_1}) - \frac{i}{4}(\|e_i + ie_{i'}\|^2_{\mathcal{H}_1} - \|e_i - ie_{i'}\|^2_{\mathcal{H}_1}) \nonumber\\ =& \:  \frac{1}{4}(\|U_{\langle f_j \rangle}(e_i + e_{i'})\|^2_{\mathcal{H}} - \|U_{\langle f_j \rangle}(e_i - e_{i'})\|^2_{\mathcal{H}}) \nonumber\\ -& \: \frac{i}{4}(\|U_{\langle f_j \rangle}(e_i + ie_{i'})\|^2_{\mathcal{H}} - \|U_{\langle f_j \rangle}(e_i - ie_{i'})\|^2_{\mathcal{H}}) \nonumber\\ =& \: \frac{1}{4}(\|U_{\langle f_j \rangle}e_i + U_{\langle f_j \rangle}e_{i'}\|^2_{\mathcal{H}} - \|U_{\langle f_j \rangle}e_i - U_{\langle f_j \rangle}e_{i'}\|^2_{\mathcal{H}}) \nonumber\\ -& \: \frac{i}{4}(\|U_{\langle f_j \rangle}e_i + iU_{\langle f_j \rangle}e_{i'}\|^2_{\mathcal{H}} - \|U_{\langle f_j \rangle}e_i - iU_{\langle f_j \rangle}e_{i'}\|^2_{\mathcal{H}}) \nonumber\\ =& \:  \langle U_{\langle f_j \rangle}e_i , U_{\langle f_j \rangle}e_{i'} \rangle_{\mathcal{H}}
\end{align} 
falls $U_{\langle f_j \rangle}$ unitär und 
\begin{align}
     \langle e_i , e_{i'} \rangle_{\mathcal{H}_1} =& \: \frac{1}{4}(\|e_i + e_{i'}\|^2_{\mathcal{H}_1} - \|e_i - e_{i'}\|^2_{\mathcal{H}_1}) - \frac{i}{4}(\|e_i + ie_{i'}\|^2_{\mathcal{H}_1} - \|e_i - ie_{i'}\|^2_{\mathcal{H}_1}) \nonumber\\ =& \:  \frac{1}{4}(\|U_{\langle f_j \rangle}(e_i + e_{i'})\|^2_{\mathcal{H}} - \|U_{\langle f_j \rangle}(e_i - e_{i'})\|^2_{\mathcal{H}}) \nonumber\\ -& \: \frac{i}{4}(\|U_{\langle f_j \rangle}(e_i + ie_{i'})\|^2_{\mathcal{H}} - \|U_{\langle f_j \rangle}(e_i - ie_{i'})\|^2_{\mathcal{H}}) \nonumber\\ =& \: \frac{1}{4}(\|U_{\langle f_j \rangle}e_i + U_{\langle f_j \rangle}e_{i'}\|^2_{\mathcal{H}} - \|U_{\langle f_j \rangle}e_i - U_{\langle f_j \rangle}e_{i'}\|^2_{\mathcal{H}}) \nonumber\\ -& \: \frac{i}{4}(\|U_{\langle f_j \rangle}e_i - iU_{\langle f_j \rangle}e_{i'}\|^2_{\mathcal{H}} - \|U_{\langle f_j \rangle}e_i + iU_{\langle f_j \rangle}e_{i'}\|^2_{\mathcal{H}}) \nonumber\\ =& \:  \langle U_{\langle f_j \rangle}e_{i'} , U_{\langle f_j \rangle}e_i \rangle_{\mathcal{H}} \nonumber\\ =& \: \langle U_{\langle f_j \rangle}e_i , U_{\langle f_j \rangle}e_{i'} \rangle_{\mathcal{H}}^*
\end{align}
falls $U_{\langle f_j \rangle}$ antiunitär ist. Insgesamt folgt somit $U_{\langle f_j \rangle}(e_i) \perp U_{\langle f_{j'} \rangle}(e_{i'})$ für $i \neq i'$ und $j \neq j'$. Damit ist $\{U_{\langle f_j \rangle}(e_i) \in \mathcal{H} : i \in I, \: j \in J\}$ eine orthonormale Menge in $\mathcal{H}$. Es bleibt zu zeigen, dass $\{U_{\langle f_j \rangle}(e_i) \in \mathcal{H} : i \in I, \: j \in J\}$ den Raum $\mathcal{H}$ aufspannt. Dazu betrachten wir 
\begin{align}
    \sqcup_{i \in I, j \in J} \langle U_{\langle f_j \rangle}(e_i) \rangle =& \:  \sqcup_{i \in I, j \in J} (h_1(\langle e_i \rangle) \sqcap h_2(\langle f_j \rangle)) \nonumber\\ =& \: h_1(\mathcal{H}_1) \sqcap h_2(\mathcal{H}_2) \nonumber\\ =& \: \mathcal{H}.
\end{align}
Daraus folgt, dass die Menge $\{U_{\langle f_j \rangle}(e_i) \in \mathcal{H} : i \in I, \: j \in J\}$ eine ONB von $\mathcal{H}$ ist.
\end{proof}

Im Folgenden Unterabschnitt wollen wir uns noch mit dem Begriff des Tensors und der Tensorprodukträume auseinandersetzen, um dann schlussendlich zu zeigen, dass $\mathcal{L}(\mathcal{H})$ tatsächlich isomorph zu $\mathcal{L}(\mathcal{H}_1 \otimes \mathcal{H}_2)$ ist. 


\subsubsection{Tensoren und Tensorprodukträume}

Zuerst benötigen wir einige grundlegende Definitionen, um den Tensorbegriff einzuführen:

\begin{definition}
Seien $\mathcal{H}$ und $\mathcal{G}$ zwei Hilberträume. Das Skalarprodukt auf $\mathcal{H}$ bezeichnen wir mit $\langle \cdot,\cdot \rangle_\mathcal{H}$ und das Skalarprodukt auf $\mathcal{G}$ mit $\langle \cdot,\cdot \rangle_\mathcal{G}$. Beide Skalarprodukte induzieren eine Norm, die wir mit $\|\cdot\|_\mathcal{H}$ und $\|\cdot\|_\mathcal{G}$ bezeichnen wollen. Eine Abbildung
\begin{align}
    F : \mathcal{H} \longrightarrow \mathcal{G}
\end{align}
heißt dann stetig, wenn ein $M>0$ existiert, sodass für alle $x \in \mathcal{H}$ 
\begin{align}
    \|F(x)\|_\mathcal{G} \leq M \|x\|_\mathcal{H}
\end{align}
ist.
\end{definition}

Wir betrachten nun eine spezielle Menge von stetigen Funktionen auf dem Hilbertraum $\mathcal{H}$:

\begin{definition}
Sei $\mathcal{H}$ ein komplexer Hilbertraum. Mit $\mathcal{H}^*$ wollen wir den sogenannten (topologischen) Dualraum von $\mathcal{H}$ bezeichnen, der als die Menge aller stetiger, linearer Abbildung von $\mathcal{H}$ in den Körper $\mathbb{C}$ definiert ist, d.h. 
\begin{align}
    \mathcal{H}^* = \{ f: \mathcal{H} \rightarrow \mathbb{C} \: : \: f \: \textrm{linear und stetig} \}.
\end{align}
Die Stetigkeit der Abbildungen $f$ ist dabei bezüglich der von dem Skalarprodukt induzierten Norm auf $\mathcal{H}$ und dem Betrag auf $\mathbb{C}$, der ja ebenfalls eine Norm darstellt, zu verstehen. Man nennt $f \in \mathcal{H}^*$ auch lineares, stetiges Funktional. Weiter nennt man $\mathcal{H}^{**} = (\mathcal{H}^*)^*$ topologischer Bidualraum von $\mathcal{H}$.
\end{definition}

\begin{definition}
Sei $\mathcal{H}$ ein komplexer Hilbertraum. Wir erklären die kanonische Einbettung $\mathbb{i}_\mathcal{H} : \mathcal{H} \rightarrow \mathcal{H}^{**}$ von $\mathcal{H}$ in seinen Bidualraum $\mathcal{H}^{**}$ durch 
\begin{align}
    (\mathbb{i}_\mathcal{H}(x))(f) := f(x),
\end{align}
wobei $x \in \mathcal{H}$ und $f \in \mathcal{H}^*$ ist. Aus der Definition wird sofort ersichtlich, dass es sich bei $\mathbb{i}_\mathcal{H}$ um einen linearen Operator handelt.
\end{definition}

\begin{definition}
Sei $\mathcal{H}$ ein komplexer Hilbertraum. Wir nennen $\mathcal{H}$ reflexiv, wenn die kanonische Einbettung $\mathbb{i}_\mathcal{H}$ bijektiv ist.
\end{definition}

Im folgenden nutzen wir folgendes wichtiges Resultat aus der Funktionalanalysis (für den Beweis siehe \cite{kreyszig1991introductory}), welches später für die Theorie wichtig wird.

\begin{lemma}
Jeder Hilbertraum $\mathcal{H}$ ist ein reflexiver Raum.
\end{lemma}

\begin{remark}
Eine direkte Folgerung aus Lemma 5.10 ist, dass jeder Hilbertraum $\mathcal{H}$ isometrisch isomorph zu seinem Bidualraum $\mathcal{H}^{**}$ ist. Dabei wird die Isomorphie über die kanonische Einbettung realisiert.
\end{remark}

Ein weiteres wichtiges Resultat, was wir im folgenden nutzen wollen, ist der Rieszsche Darstellungssatz (Beweis siehe \cite{conway2019course}):

\begin{theorem}
Sei $\mathcal{H}$ ein komplexer Hilbertraum, wobei $\langle \cdot , \cdot \rangle_\mathcal{H}$ das Skalarprodukt von $\mathcal{H}$ ist, welches im ersten Argument antilinear und im zweiten Argument linear ist. Für jedes Element $f \in \mathcal{H}^*$ aus dem topologischen Dualraum von $\mathcal{H}$ existiert ein eindeutiges Element $\phi_{f} \in \mathcal{H}$, so dass 
\begin{align}
    f(x) = \langle \phi_{f}, x \rangle_\mathcal{H} 
\end{align}
für alle $x \in \mathcal{H}$. Desweiteren gilt, dass 
\begin{align}
    \|\phi_{f}\|_\mathcal{H} = \|f\|_{\mathcal{O}}
\end{align}
ist, wobei $\|f\|_{\mathcal{O}} = sup_{\|x\|_\mathcal{H} \leq 1}|f(x)|$ die Operatornorm ist.
\end{theorem}

Nach Theorem 5.4 existiert demnach ein kanonischer Antiisomorphismus zwischen dem komplexen Hilbertraum $\mathcal{H}$ und seinem topologischen Dualraum $\mathcal{H}^*$:
\begin{align}
     \mathbb{k} : \mathcal{H} \longrightarrow& \: \: \mathcal{H}^* \nonumber\\
    x \longmapsto& \: \: \mathbb{k}(x) := x^*, \: \: \: x^*(y) = \langle x , y \rangle_\mathcal{H} \: \: \forall y \in \mathcal{H}
\end{align}
Wir nennen die Abbildung $\mathbb{k}$ Riesz-Abbildung.\\

Wie bereits oft angesprochen, wollen wir zeigen, dass $\mathcal{L}(\mathcal{H})$ und $\mathcal{L}(\mathcal{H}_1 \otimes \mathcal{H}_2)$ isomorph sind. D.h., das wir insbesondere den bisher noch nicht eingeführten Tensorproduktraum $\mathcal{H}_1 \otimes \mathcal{H}_2$ als Zustandsvektorraum des zusammengesetzten quantenmechanishen Systems $\Sigma$ verstehen wollen. Demzufolge erwarten wir gemäß Abschnitt 4.1, dass es sich bei $\mathcal{H}_1 \otimes \mathcal{H}_2$ um einen Hilbertraum handelt. 

Um später in diesem Abschnitt die Tensorprodukträume tatsächlich zu Hilberträumen zu machen, müssen wir zuerst ein Skalarprodukt auf $\mathcal{H}^*$ einführen.   


Dazu betrachten wir wieder die Abbildung $\mathbb{k}$ aus dem Rieszschen Darstellungssatz. Mit dieser ist es nun möglich, den Dualraum $\mathcal{H}^*$ des komplexen Hilbertraumes $\mathcal{H}$ zu einem Hilbertraum zu machen. Dazu definieren wir in $\mathcal{H}^*$ das innere Produkt 
\begin{align}
    [f,g]_{\mathcal{H}^*} := \langle \mathbb{k}^{-1}(g) , \mathbb{k}^{-1}(f) \rangle_\mathcal{H}, \qquad \forall f,g \in \mathcal{H}^*. \label{eq:Dual}
\end{align}
Wegen der Additivität der Abbildung $\mathbb{k}$ folgt die Additivität in den beiden Argumenten von $[ \cdot , \cdot ]_{\mathcal{H}^*}$. Aus der Bijektivität von $\mathbb{k}$ folgt die positive Definitheit von $[ \cdot , \cdot ]_{\mathcal{H}^*}$. Da  $\mathbb{k}$ ein Antiisomorphismus ist, folgt für alle $\lambda \in \mathbb{C}$ im ersten Argument, dass
\begin{align}
    [\lambda f , g]_{\mathcal{H}^*} =& \: \langle \mathbb{k}^{-1}(g) , \mathbb{k}^{-1}(\lambda f) \rangle_\mathcal{H} \nonumber\\ =& \:  \langle \mathbb{k}^{-1}(g) , \lambda^* \mathbb{k}^{-1}(f) \rangle_\mathcal{H} \nonumber\\ =& \: \lambda^* \langle \mathbb{k}^{-1}(g) , \mathbb{k}^{-1}(f) \rangle_\mathcal{H} \nonumber\\ =& \: \lambda^* [f,g]_{\mathcal{H}^*},
\end{align}
ist und im zweiten Argument gilt
\begin{align}
    [f , \lambda g]_{\mathcal{H}^*} =& \: \langle \mathbb{k}^{-1}(\lambda g) , \mathbb{k}^{-1}(f) \rangle_\mathcal{H} \nonumber\\ =& \:  \langle \lambda^* \mathbb{k}^{-1}(g) , \mathbb{k}^{-1}(f) \rangle_\mathcal{H} \nonumber\\ =& \: \lambda \langle \mathbb{k}^{-1}(g) , \mathbb{k}^{-1}(f) \rangle_\mathcal{H} \nonumber\\ =& \: \lambda [f,g]_{\mathcal{H}^*}.
\end{align}
Die Hermitizität von $[\cdot, \cdot]_{\mathcal{H}^*}$ ergibt sich definitionsgemäß aus der Tatsache, dass $\langle \cdot,\cdot \rangle_\mathcal{H}$ ein komplexes Skalarprodukt ist, womit ingesamt folgt, dass $[\cdot, \cdot]_{\mathcal{H}^*}$ eine positiv definite, hermitische Sesquilinearform ist, die im ersten Argument antilinear ist. Wegen 
\begin{align}
    \|f\|_{\mathcal{H}^*} :=& \: \sqrt{[f,f]_{\mathcal{H}^*}} \nonumber\\ =& \: \sqrt{\langle \mathbb{k}^{-1}(f) , \mathbb{k}^{-1}(f) \rangle_\mathcal{H}} \nonumber\\ =& \: \|\mathbb{k}^{-1}(f)\|_\mathcal{H} \nonumber\\ =& \: \|f\|_\mathcal{O}
\end{align}
und der Tatsache, dass $(\mathcal{H}^*, \| \cdot \|_\mathcal{O})$ ein Banachraum ist \cite{rudin1991functional}, d.h. ein vollständiger normierter Raum, folgt, dass $(\mathcal{H}^*, [\cdot,\cdot]_{\mathcal{H}^*})$ ein Hilbertraum ist. 

\begin{remark}
Mittels der Riesz-Abbildung ist es nun auch möglich, eine einfache ONB in $\mathcal{H}^*$ zu konstruieren. Sei dazu im Folgenden $I$ eine Indexmengen und $\mathcal{A} = \{e_i \in \mathcal{H} : i \in I\}$ eine beliebige ONB von $\mathcal{H}$. Dann definiert die Menge $\{\mathbb{k}(e_i) = (e_i)^{*} := e^{i} : e_i \in \mathcal{A}\}$ eine Basis in $\mathcal{H}^*$, die gemäß \eqref{eq:Dual} eine ONB in $\mathcal{H}^*$ darstellt. Mittels einer ONB in $\mathcal{H}$ und der eben angegeben ONB in $\mathcal{H}^*$ werden wir später in der Lage sein, eine einfache ONB in den Tensorprodukträumen einzuführen, was es uns später erleichtern wird, lineare und antilineare Bijektionen zwischen $\mathcal{H}$ und $\mathcal{H}_1 \otimes \mathcal{H}_2$ zu konstruieren. 
\end{remark}

Nun lässt sich der Begriff des Tensors (in unserem funktionalanalytischen Kontext) einführen: 

\begin{definition}
Seien $r,s \in \mathbb{N}_0 = \mathbb{N} \cup \{0\}$ und $\mathcal{H}$ ein komplexer Hilbertraum. Unter einem $(r,s)$-Tensor $T$ verstehen wir die stetige Multilinearform
\begin{align}
    T: \underbrace{\mathcal{\mathcal{H}} \times ... \times \mathcal{H}}_{s} \times \underbrace{\mathcal{H}^* \times ... \times \mathcal{H}^*}_{r} \longrightarrow \mathbb{C},
\end{align}
wobei Multilinearform meint, dass die Abbildung in den Grundkörper $\mathbb{C}$ abbildet und linear in jedem Argument ist. Die Stetigkeit der Multilinearform ist dabei wie folgt zu verstehen: Es existiert eine Konstante $M>0$, sodass für alle $x_1,...,x_s \in \mathcal{H}$ und alle $\omega_1,...,\omega_r \in \mathcal{H}^*$ gilt, dass
\begin{align}
    |T(x_1,...,x_s,\omega_1,...,\omega_r)| \leq M \|x_1\|_\mathcal{H}...\|x_s\|_\mathcal{H} \|\omega_1\|_{\mathcal{H}^*}...\|\omega_r\|_{\mathcal{H}^*}.
\end{align}

Seien $I$ und $J$ Indexmengen und $\mathcal{A} = \{e_i \in \mathcal{H} : i \in I\}$ eine Basis in $\mathcal{H}$ und $\mathcal{B} = \{\epsilon^j \in \mathcal{H}^* : j \in J\}$ eine Basis in $\mathcal{H}^*$. Dann sind die Komponenten vom $(r,s)$-Tensor $T$ bzgl. der Basen $\mathcal{A}$ und $\mathcal{B}$ gegeben durch 
\begin{align}
    T(e_{i_1}, ... , e_{i_s}, \epsilon^{j_1}, ..., \epsilon^{j_r}) =  T^{j_1,...,j_r}_{i_1,...,i_s}
\end{align}
für $i_1,...,i_s \in I$ und $j_1,...,j_r \in J$.
\end{definition}

Aus der Definition eines $(r,s)$-Tensors auf $\mathcal{\mathcal{H}} \times ... \times \mathcal{H} \times \mathcal{H}^* \times ... \times \mathcal{H}^*$ ergibt sich, dass $(0,1)$-Tensoren gerade die linearen, stetigen Funktionale sind, d.h. für einen $(0,1)$-Tensor $f$ gilt $f \in \mathcal{H}^*$. Weiter entsprechen die $(1,0)$-Tensoren Elementen aus dem Bidualraum $\mathcal{H}^{**}$. Da $\mathcal{H}$ und $\mathcal{H}^{**}$ nach Bemerkung 5.1 isometrisch isomorph sind, können wir die $(1,0)$-Tensoren auch äquivalent als Elemente aus dem komplexen Hilbertraum $\mathcal{H}$ auffassen. Dabei gilt für ein Element $\eta \in \mathcal{H}^{**}$, welches über die Abbildung $\mathbb{i}^{-1}_\mathcal{H}$ auf $e \in \mathcal{H}$ abgebildet wird, dass $\eta(\epsilon) = (\mathbb{i}_\mathcal{H}(e))(\epsilon) = \epsilon(\mathbb{i}^{-1}_\mathcal{H}(\eta)) = \epsilon(e)$ für alle $\epsilon \in \mathcal{H}^*$ ist. Wir setzen daher $e(\epsilon) := \epsilon(e)$. 

\begin{definition}
Seien $r,s,m,n \in \mathbb{N}_0$. Sei $T$ ein $(r,s)$-Tensor und $C$ ein $(m,n)$-Tensor. Dann ist das Tensorprodukt $\otimes$ zwischen $T$ und $C$, bezeichnet mit $T \otimes C$, definiert durch 
\begin{align}
    (T \otimes C)&(e_1,...,e_{s-1},e_s,e_{s+1},...,e_{s+n},\epsilon^1,...,\epsilon^{r-1},\epsilon^r,\epsilon^{r+1},...,\epsilon^{r+m}) \nonumber\\ =& \: T(e_1,...,e_{s-1},e_s,\epsilon^1,...,\epsilon^{r-1},\epsilon^r)C(e_{s+1},...,e_{s+n},\epsilon^{r+1},...,\epsilon^{r+m}),
\end{align}
wobei $e_1,...,e_{s-1},e_s,e_{s+1},...,e_{s+n} \in \mathcal{H}$ und $\epsilon^1,...,\epsilon^{r-1},\epsilon^r,\epsilon^{r+1},...,\epsilon^{r+m} \in \mathcal{H}^*$. Aus der Definition wird sofort klar, dass $T \otimes C$ ein $(r+m,s+n)$-Tensor ist.
\end{definition}

Seien $I$ und $J$ wieder zwei Indexmengen und $\mathcal{A} = \{e_i \in \mathcal{H} : i \in I\}$ eine Basis vom komplexen Hilbertraum $\mathcal{H}$ und $\mathcal{B} = \{\epsilon^j \in \mathcal{H}^* : j \in J\}$ eine Basis von $\mathcal{H}^*$. Mittels des Tensorproduktes $\otimes$ und Bemerkung 5.1 sehen wir, dass Elemente der Form $\epsilon^{j_1} \otimes ... \otimes \epsilon^{j_r} \otimes e_{i_1} \otimes ... \otimes e_{i_s}$ mit $i_1,...,i_s \in I$ und $j_1,...,j_r \in J$ $(r,s)$-Tensoren darstellen, die wir $(r,s)$-Elementartensoren nennen.

\begin{definition}
Sei $\mathcal{H}$ ein komplexer Hilbertraum und sei $\mathcal{A} = \{e_i \in \mathcal{H} : i \in I\}$ eine Basis in $\mathcal{H}$ und $\mathcal{B} = \{\epsilon^j \in \mathcal{H}^* : j \in J\}$ eine Basis in $\mathcal{H}^*$. Wir erklären die Menge 
\begin{align}
    \underbrace{\mathcal{H}^* \odot ... \odot \mathcal{H}^*}_{r} \odot \underbrace{\mathcal{H} \odot ... \odot \mathcal{H}}_{s} 
\end{align}
als die komplexe lineare Hülle aller $(r,s)$-Elementartensoren, d.h. die Menge $\mathcal{H}^* \odot ... \odot \mathcal{H}^* \odot \mathcal{H} \odot ... \odot \mathcal{H}$ wird identifiziert mit der Menge
\begin{align}
    span_\mathbb{C}(\{\epsilon^{j_1} \otimes ... \otimes \epsilon^{j_r} \otimes e_{i_1} \otimes ... \otimes e_{i_s} : i_1,...,i_s \in I, \: j_1,...,j_r \in J\}).
\end{align}
Wir bezeichnen $\mathcal{H}^* \odot ... \odot \mathcal{H}^* \odot \mathcal{H} \odot ... \odot \mathcal{H}$ als $(r,s)$-Tensorproduktraum.
\end{definition}

Es lässt sich zeigen, dass die Menge 
\begin{align}
    \{\epsilon^{j_1} \otimes ... \otimes \epsilon^{j_r} \otimes e_{i_1} \otimes ... \otimes e_{i_s} : i_1,...,i_s \in I, \: j_1,...,j_r \in J\}
\end{align}
eine Basis in  $\mathcal{H}^* \odot ... \odot \mathcal{H}^* \odot \mathcal{H} \odot ... \odot \mathcal{H}$ ist (\cite{folland1995course}, Proposition 7,14). \\

Betrachten wir den $(r,s)$-Tensorproduktraum $\mathcal{H}^* \odot ... \odot \mathcal{H}^* \odot \mathcal{H} \odot ... \odot \mathcal{H}$ etwas genauer, so stellen wir fest, dass jeder Faktor einen Hilbertraum darstellt, sofern wir $\mathcal{H}$ mit $\langle \cdot, \cdot \rangle_\mathcal{H}$ und $\mathcal{H}^*$ mit $[\cdot, \cdot]_{\mathcal{H}^*}$ ausstatten. Es stellt sich nun, wie bereits angekündigt, die Frage, ob wir den $(r,s)$-Tensorproduktraum mit einem inneren Produkt so derart ausstatten können, dass damit auch dieser wieder ein Hilbertraum ist, denn unser Ziel ist es ja, den Tensorproduktraum als quantenmechanischen Zustandsvektorraum zu verstehen. Und Zustandsvektorräume sind nach Abschnitt 4.1 Hilberträume. Eine natürliche Wahl für ein inneres Produkt auf $\mathcal{H}^* \odot ... \odot \mathcal{H}^* \odot \mathcal{H} \odot ... \odot \mathcal{H}$ ist dabei die lineare Fortsetzung von
\begin{align}
    \langle \epsilon^{j_1} \otimes ... \otimes \epsilon^{j_r} \otimes e_{i_1} &\otimes ... \otimes e_{i_s} , \epsilon^{k_1} \otimes ... \otimes \epsilon^{k_r} \otimes e_{l_1} \otimes ... \otimes e_{l_s} \rangle \nonumber\\ :=& [\epsilon^{j_1},\epsilon^{k_1}]_{\mathcal{H}^*}...[\epsilon^{j_r},\epsilon^{k_r}]_{\mathcal{H}^*} \langle e_{i_1},e_{l_1} \rangle_\mathcal{H}...\langle e_{i_s},e_{l_s} \rangle_\mathcal{H}. \label{eq:TenHil}
\end{align}
Es lässt sich leicht nachzuprüfen, dass $\langle \cdot,\cdot \rangle$ auf $\mathcal{H}^* \odot ... \odot \mathcal{H}^* \odot \mathcal{H} \odot ... \odot \mathcal{H}$ eine positiv definite, hermetische Sesquilinearform definiert. Im allgemeinen folgt aber nicht, das $(\mathcal{H}^* \odot ... \odot \mathcal{H}^* \odot \mathcal{H} \odot ... \odot \mathcal{H}, \langle \cdot,\cdot \rangle)$ ein Hilbertraum ist. Mit Blick auf die Tatsache, dass wir zeigen wollen, dass sich ein zusammengesetztes quantenmechanisches System über einen $(r,s)$-Tensorproduktraum, bestehend aus den Hilberträumen der Teilsysteme, modellieren lässt, ist diese letzte Beobachtung natürlich problematisch, da Zustandsvektorräumen nach obigen Ausführungen vollständig bezüglich ihres Skalarproduktes sein sollten. Um dieses Problem zu beheben, vervollständigen wir einfach den Raum $\mathcal{H}^* \odot ... \odot \mathcal{H}^* \odot \mathcal{H} \odot ... \odot \mathcal{H}$ bzgl. $\langle \cdot,\cdot \rangle$. Das führt auf folgende Definition:

\begin{definition}
Sei $\mathcal{H}$ ein komplexer Hilbertraum und sei $\mathcal{H}^* \odot ... \odot \mathcal{H}^* \odot \mathcal{H} \odot ... \odot \mathcal{H}$ der zugehörige $(r,s)$-Tensorproduktraum, den wir mit dem oben definierten Skalarprodukt $\langle \cdot,\cdot \rangle$ ausstatten. Wir bezeichnen mit 
\begin{align}
    \underbrace{\mathcal{H}^* \otimes ... \otimes \mathcal{H}^*}_{r} \otimes \underbrace{\mathcal{H} \otimes ... \otimes \mathcal{H}}_{s}
\end{align}
die Vervollständigung von $\mathcal{H}^* \odot ... \odot \mathcal{H}^* \odot \mathcal{H} \odot ... \odot \mathcal{H}$ bezgl. $\langle \cdot,\cdot \rangle$. Wir nennen $\mathcal{H}^* \otimes ... \otimes \mathcal{H}^* \otimes \mathcal{H} \otimes ... \otimes \mathcal{H}$ den $(r,s)$-Hilbert-Tensorproduktraum.
\end{definition}


\subsubsection{Isomorphe Propositionensysteme zusammengesetzter quantenmechanischer Systeme}

Wir betrachten nun wieder unsere zwei quantenmechanischen Systeme $\Sigma_1$ und $\Sigma_2$ dir wir als die beiden Teilsysteme des zusammengesetzten Systems $\Sigma$ verstehen wollen, wobei $\Sigma_1$ über den komplexen Hilbertraum $\mathcal{H}_1$, $\Sigma_2$ über den komplexen Hilbertraum $\mathcal{H}_2$ und $\Sigma$ über den komplexen Hilbertraum $\mathcal{H}$ modelliert werden.

Fassen wir noch einmal kurz zusammen, was wir bisher getan haben: In den vorangegangenen Abschnitten haben wir die Abbildungen $h_1: \mathcal{L}(\mathcal{H}_1) \rightarrow \mathcal{L}(\mathcal{H})$ und $h_2: \mathcal{L}(\mathcal{H}_2) \rightarrow \mathcal{L}(\mathcal{H})$ untersucht und festgestellt, dass es sich bei diesen um nicht gemischte m-Morphismen handelt, was dazu führte, dass gemäß Theorem 5.2 für beide Abbildungen eine Familie von speziellen Operatoren existiert, mit deren Hilfe wir in der Lage waren, die unitären oder antiunitären Abbildungen $U_{\langle x_2 \rangle}$ bzw. $V_{\langle x_1 \rangle}$ zu konstruieren. Diese bilden jeweils von den Hilberträumen $\mathcal{H}_1$ bzw. $\mathcal{H}_2$ der Subsysteme in den Hilbertraum $\mathcal{H}$ des Gesamtsystems ab. Wählt man eine ONB $\{e_i \in \mathcal{H}_1 : i \in I\}$ in $\mathcal{H}_1$ und eine ONB $\{f_j \in \mathcal{H}_2 : j \in J\}$ in $\mathcal{H}_2$, so ließ sich mittels der Operatoren $U_{\langle f_j \rangle}$ über $U_{\langle f_j \rangle}(e_i)$ eine ONB in $\mathcal{H}$ definieren. Im letzten Unterabschnitt haben wir dann den Begriff der Tensorprodukträume kennengelernt. Unser Ziel wird es nun unter anderem sein, bijektive Abbildungen zwischen $\mathcal{H}$ und $\mathcal{H}_1 \otimes \mathcal{H}_2$ zu erklären, die im wesentlichen strukturerhaltend sind. Im wesentlichen strukturerhaltend meint dabei, dass die Abbildungen entweder linear oder antilinear sind.   

Die Bedeutung dieser zu konstruierenden Abbildungen liegt darin, dass diese sich aufgrund ihrer Eigenschaften dazu eignen, Abbildungen zu konstruieren, die die Isomorphie zwischen den Propositionensystemen $\mathcal{L}(\mathcal{H})$ und $\mathcal{L}(\mathcal{H}_1 \otimes \mathcal{H}_2)$ vermitteln. Immerhin müssen derartige Abbildungen jedem abgeschlossenen Unterraum aus $\mathcal{H}$ eineindeutig einen abgeschlossenen Unterraum aus $\mathcal{H}_1 \otimes \mathcal{H}_2$ zuordnen. 

Im folgenden sei daher $\mathcal{A} = \{e_i \in \mathcal{H}_1 : i \in I\}$ eine ONB in $\mathcal{H}_1$ und $\mathcal{B} = \{f_j \in \mathcal{H}_2 : j \in J\}$ sei eine ONB in $\mathcal{H}_2$. Es gilt nach obigem, dass die Menge $\{\mathbb{k}_1(e_i) := e^i : e_i \in \mathcal{A}\}$ mit der Riesz-Abbildung $\mathbb{k}_1 : \mathcal{H}_1 \rightarrow \mathcal{H}^*_1$ eine ONB in $\mathcal{H}^*_1$ (Bemerkung 5.2), die Menge $\{e_i \otimes f_j : i \in I, j \in J\}$ eine ONB in $\mathcal{H}_1 \otimes \mathcal{H}_2$ und die Menge $\{e^i \otimes f_j : i \in I, j \in J\}$ eine ONB in $\mathcal{H}^*_1 \otimes \mathcal{H}_2$ ist (siehe dazu \eqref{eq:TenHil}). Wir betrachten die folgenden Abbildungen:
\begin{itemize}
    \item [\textit{i)}] Sind $h_1$ und $h_2$ lineare m-Morphismen, so definieren wir
    \begin{align}
         \phi_{e,f} : \mathcal{H}_1 \otimes \mathcal{H}_2 \longrightarrow& \: \: \mathcal{H}, \nonumber\\
        \sum_{i \in I, j \in J} x^{ij} (e_i \otimes f_j) \longmapsto& \: \: \sum_{i \in I, j \in J} x^{ij} U_{\langle f_j \rangle}(e_i)
    \end{align}
    \item [\textit{ii)}] Sind $h_1$ und $h_2$ antilineare m-Morphismen, so definieren wir
    \begin{align}
          \psi_{e,f} : \mathcal{H}_1 \otimes \mathcal{H}_2 \longrightarrow& \: \: \mathcal{H}, \nonumber\\
        \sum_{i \in I, j \in J} x^{ij} (e_i \otimes f_j) \longmapsto& \: \: \sum_{i \in I, j \in J} (x^{ij})^* U_{\langle f_j \rangle}(e_i)
    \end{align}
    \item [\textit{iii)}] Ist $h_1$ ein antilinearer und $h_2$ ein linearer m-Morphismus, so definieren wir
    \begin{align}
          \mu_{e,f} : \mathcal{H}^*_1 \otimes \mathcal{H}_2 \longrightarrow& \: \: \mathcal{H}, \nonumber\\
        \sum_{i \in I, j \in J} x_{i}^{j} (e^i \otimes f_j) \longmapsto& \: \: \sum_{i \in I, j \in J} x_{i}^{j} U_{\langle f_j \rangle}(e_i)
    \end{align}
    \item [\textit{iv}] Ist $h_1$ ein linearer und $h_2$ ein antilinearer m-Morphismus, so definieren wir
    \begin{align}
          \nu_{e,f} : \mathcal{H}^*_1 \otimes \mathcal{H}_2 \longrightarrow& \: \: \mathcal{H}, \nonumber\\
        \sum_{i \in I, j \in J} x_{i}^{j} (e^i \otimes f_j) \longmapsto& \: \: \sum_{i \in I, j \in J} (x_{i}^{j})^* U_{\langle f_j \rangle}(e_i)
    \end{align}
\end{itemize}
Per Konstruktion folgt, dass diese Abbildungen das Skalarprodukt erhalten und Additiv sind. Die Abbildungen $\phi_{e,f}$ und $\mu_{e,f}$ sind dabei unitäre Vektorraumhomomorphismen und die Abbildungen $\psi_{e,f}$ und $\nu_{e,f}$ antiunitäre Abbildungen, den es gilt:
\begin{align}
    \phi_{e,f}(\sum_{i \in I, j \in J} x^{ij} (e_i \otimes f_j)) =& \: \phi_{e,f}(\sum_{i \in I, j \in J} (x^{ij} e_i) \otimes f_j) = \sum_{i \in I, j \in J} U_{\langle f_j \rangle}(x^{ij} e_i) \nonumber\\ =& \: \sum_{i \in I, j \in J} x^{ij} U_{\langle f_j \rangle}(e_i) = \sum_{i \in I, j \in J} x^{ij} \phi_{e,f}(e_i \otimes f_j) \nonumber\\ =& \: \sum_{i \in I, j \in J} x^{ij} V_{\langle e_i \rangle}(f_j) = \sum_{i \in I, j \in J} V_{\langle e_i \rangle}(x^{ij} f_j) \nonumber\\ =& \: \phi_{e,f}(\sum_{i \in I, j \in J}  e_i \otimes (x^{ij} f_j)),
\end{align}
\begin{align}
    \phi_{e,f}(\sum_{i \in I, j \in J} x^{ij} (e_i \otimes f_j)) =& \: \psi_{e,f}(\sum_{i \in I, j \in J} (x^{ij} e_i) \otimes f_j) = \sum_{i \in I, j \in J} U_{\langle f_j \rangle}(x^{ij} e_i) \nonumber\\ =& \: \sum_{i \in I, j \in J} (x^{ij})^* U_{\langle f_j \rangle}(e_i) = \sum_{i \in I, j \in J} (x^{ij})^* \psi_{e,f}(e_i \otimes f_j) \nonumber\\ =& \: \sum_{i \in I, j \in J} (x^{ij})^* V_{\langle e_i \rangle}(f_j) = \sum_{i \in I, j \in J} V_{\langle e_i \rangle}(x^{ij} f_j) \nonumber\\ =& \: \psi_{e,f}(\sum_{i \in I, j \in J}  e_i \otimes (x^{ij} f_j)),
\end{align}
\begin{align}
    \mu_{e,f}(\sum_{i \in I, j \in J} x^{j}_i (e^i \otimes f_j)) =& \: \mu_{e,f}(\sum_{i \in I, j \in J} ((x^{j}_i)^* e_i)^* \otimes f_j) = \sum_{i \in I, j \in J} U_{\langle f_j \rangle}((x^{j}_i)^* e_i) \nonumber\\ =& \: \sum_{i \in I, j \in J} x^{j}_i U_{\langle f_j \rangle}(e_i) = \sum_{i \in I, j \in J} x^{j}_i \mu_{e,f}(e^i \otimes f_j) \nonumber\\ =& \: \sum_{i \in I, j \in J} x^{j}_i V_{\langle e_i \rangle}(f_j) = \sum_{i \in I, j \in J} V_{\langle e_i \rangle}(x^{j}_i f_j) \nonumber\\ =& \: \mu_{e,f}(\sum_{i \in I, j \in J}  e^i \otimes (x^{j}_i f_j)),
\end{align}
\begin{align}
    \nu_{e,f}(\sum_{i \in I, j \in J} x^{j}_i (e^i \otimes f_j)) =& \: \nu_{e,f}(\sum_{i \in I, j \in J} ((x^{j}_i)^* e_i)^* \otimes f_j) = \sum_{i \in I, j \in J} U_{\langle f_j \rangle}((x^{j}_i)^* e_i) \nonumber\\ =& \: \sum_{i \in I, j \in J} (x^{j}_i)^* U_{\langle f_j \rangle}(e_i) = \sum_{i \in I, j \in J} (x^{j}_i)^* \nu_{e,f}(e^i \otimes f_j) \nonumber\\ =& \: \sum_{i \in I, j \in J} (x^{j}_i)^* V_{\langle e_i \rangle}(f_j) = \sum_{i \in I, j \in J} V_{\langle e_i \rangle}(x^{j}_i f_j) \nonumber\\ =& \: \nu_{e,f}(\sum_{i \in I, j \in J}  e^i \otimes (x^{j}_i f_j))
\end{align}

Als nächstes werden wir zeigen, dass diese Abbildungen wohldefiniert sind:

\begin{lemma}
Die Abbildungen $\phi_{e,f}, \psi_{e,f}, \mu_{e,f}$ und $\nu_{e,f}$ sind unabhängig von den gewählten Basen $\{e_i \in \mathcal{H}_1 : i \in I\}$ und $\{f_j \in \mathcal{H}_2 : j \in J\}$ definiert und damit wohldefiniert.
\end{lemma}

\begin{proof}
Wir zeigen die Basisunabhängigkeit explizit für die Abbildung $\mu_{e,f}$. Der Beweis für die restlichen drei Abbildungen erfolgt auf analoge Weise. Wir wählen nun eine weitere ONB aus $\mathcal{H}_1$, $\{p_k \in \mathcal{H}_1 : k \in K\}$, und eine weitere ONB aus $\mathcal{H}_2$, $\{q_l \in \mathcal{H}_2 : l \in L\}$. Wir zeigen, dass $\mu_{e,f} = \mu_{p,q}$ ist. Wir wählen ein $x \in \mathcal{H}^*_1 \otimes \mathcal{H}_2$. Es gilt 
\begin{align}
    x = \sum_{i \in I, j \in J} x_i^j (e^i \otimes f_j) = \sum_{k \in K, l \in L} y_k^l (p^k \otimes q_l)
\end{align}
Da $\{e_i \in \mathcal{H}_1 : i \in I\}$ bzw. $\{f_j \in \mathcal{H}_2 : j \in J\}$ ONB's in $\mathcal{H}_1$ bzw. $\mathcal{H}_2$ darstellen, gilt
\begin{align}
    p_k = \sum_{i \in I} \langle e_i,p_k \rangle_{\mathcal{H}_1}e_i, \\
    q_l = \sum_{j \in J} \langle f_j,q_l \rangle_{\mathcal{H}_2}f_j.
\end{align}
Daraus folgt 
\begin{align}
    x_i^j = \sum_{k \in K, l \in L} y_k^l \langle p_k,e_i \rangle_{\mathcal{H}_1} \langle f_j,q_l \rangle_{\mathcal{H}_2}.
\end{align}

Es gilt damit
\begin{align}
    \mu_{p,q}(x) =& \: \sum_{k \in K, l \in L} y_k^l U_{\langle q_l \rangle}(p_k) \nonumber\\ =& \: \sum_{k \in K, l \in L} y_k^l U_{\langle q_l \rangle}(\sum_{i \in I} \langle e_i,p_k \rangle_{\mathcal{H}_1}e_i) \nonumber\\ =& \: \sum_{k \in K, l \in L, i \in I} y_k^l U_{\langle q_l \rangle}(\langle e_i,p_k \rangle_{\mathcal{H}_1}e_i) \nonumber\\ =& \: \sum_{k \in K, l \in L, i \in I} y_k^l \langle p_k,e_i \rangle_{\mathcal{H}_1} U_{\langle q_l \rangle}(e_i) \nonumber\\ =& \:  \sum_{k \in K, l \in L, i \in I} y_k^l \langle p_k,e_i \rangle_{\mathcal{H}_1} V_{\langle e_i \rangle}(q_l) \nonumber\\ =& \: \sum_{k \in K, l \in L, i \in I} y_k^l \langle p_k,e_i \rangle_{\mathcal{H}_1} V_{\langle e_i \rangle}(\sum_{j \in J} \langle f_j,q_l \rangle_{\mathcal{H}_2}f_j) \nonumber\\ =& \: \sum_{k \in K, l \in L, i \in I, j \in J} y_k^l \langle p_k,e_i \rangle_{\mathcal{H}_1} \langle f_j,q_l \rangle_{\mathcal{H}_2} V_{\langle e_i \rangle}(f_j) \nonumber\\ =& \: \sum_{k \in K, l \in L, i \in I, j \in J} y_k^l \langle p_k,e_i \rangle_{\mathcal{H}_1} \langle f_j,q_l \rangle_{\mathcal{H}_2} U_{\langle f_j \rangle}(e_i) \nonumber\\ =& \: \sum_{i \in I, j \in J} x_i^j U_{\langle f_j \rangle}(e_i) \nonumber\\ =& \: \mu_{e,f}(x).
\end{align}
Dabei haben wir ausgenutzt, dass wegen 
\begin{align}
    \langle f_s,f_j \rangle_{\mathcal{H}_2} = \delta_{s,j}, \qquad
    \langle q_l,q_r \rangle_{\mathcal{H}_2} = \delta_{l,r}
\end{align}
gilt, dass
\begin{align}
    \langle f_j,q_r \rangle_{\mathcal{H}_2} =& \: \langle \sum_{l \in L} \langle q_l,f_j \rangle_{\mathcal{H}_2}q_l,q_r \rangle_{\mathcal{H}_2} \nonumber\\ =& \: \sum_{l \in L} \langle q_l,f_j \rangle_{\mathcal{H}_2} \langle q_l,q_r \rangle_{\mathcal{H}_2} \nonumber\\ =& \: \sum_{l \in L} \langle q_l,f_j \rangle_{\mathcal{H}_2} \delta_{l,r} \nonumber\\ =& \: \langle q_r,f_j \rangle_{\mathcal{H}_2}.
\end{align}
\end{proof}

Nun lässt sich unter anderem zeigen, dass die Propositionensysteme $\mathcal{L}(\mathcal{H})$ und $\mathcal{L}(\mathcal{H}_1 \otimes \mathcal{H}_2)$ isomorph sind:

\begin{theorem}
Seien $\mathcal{H}_1, \mathcal{H}_2$ und $\mathcal{H}$ jeweils komplexe Hilberträume, wobei $dim(\mathcal{H}_1), dim(\mathcal{H}_2) \geq 3$ ist. Weiter seien $\mathcal{L}(\mathcal{H}_1), \mathcal{L}(\mathcal{H}_2)$ und $\mathcal{L}(\mathcal{H})$ die zu den komplexen Hilberträumen gehörigen (quantenmechanischen) Propositionensysteme. Sind Abbildungen $h_1:\mathcal{L}(\mathcal{H}_1) \rightarrow \mathcal{L}(\mathcal{H})$ und $h_2:\mathcal{L}(\mathcal{H}_2) \rightarrow \mathcal{L}(\mathcal{H})$ gegeben, die den Bedingungen 
\begin{itemize}
    \item [\textit{i)}] $h_1$ und $h_2$ sind unitäre c-Morphismen.
    \item [\textit{ii)}] Für alle $A \in \mathcal{L}(\mathcal{H}_1)$ und für alle $B \in \mathcal{L}(\mathcal{H}_2)$ folgt, dass $h_1(A) \leftrightarrow h_2(B)$ ist.
    \item [\textit{iii)}] Für alle Atome $p \in \mathcal{L}(\mathcal{H}_1)$ und alle Atome $q \in \mathcal{L}(\mathcal{H}_2)$ folgt, dass der Ausdruck $h_1(p) \sqcap h_2(q) \in \mathcal{L}(\mathcal{H})$ ein Atom ist.
\end{itemize}
genügen, so folgt, das $\mathcal{L}(\mathcal{H})$ isomorph zu $\mathcal{L}(\mathcal{H}_1 \otimes \mathcal{H}_2)$ oder $\mathcal{L}(\mathcal{H}^*_1 \otimes \mathcal{H}_2)$ ist.
\end{theorem}

\begin{proof}
Aus dem Lemma 5.5 und Lemma 5.7 wissen wir, dass die Abbildungen $h_1$ und $h_2$ nicht gemischte m-Morphismen sind. Mittels der Abbildungen $\phi_{e,f}, \psi_{e,f}, \mu_{e,f}$ und $\nu_{e,f}$ und Lemma 5.11 folgt:
\begin{itemize}
    \item [\textit{i)}] Sind $h_1$ und $h_2$ linear, so wird ein Isomorphismus zwischen $\mathcal{L}(\mathcal{H})$ und $\mathcal{L}(\mathcal{H}_1 \otimes \mathcal{H}_2)$ über $\phi_{e,f}$ generiert:
    \begin{align}
          \phi : \mathcal{L}(\mathcal{H}_1 \otimes \mathcal{H}_2) \longrightarrow& \: \mathcal{L}(\mathcal{H}), \nonumber\\
         G \longmapsto& \: \: \{\phi_{e,f}(x) : x \in G\}
    \end{align}
    \item [\textit{ii)}] Sind $h_1$ und $h_2$ antilinear, so wird ein Isomorphismus zwischen $\mathcal{L}(\mathcal{H})$ und $\mathcal{L}(\mathcal{H}_1 \otimes \mathcal{H}_2)$ über $\psi_{e,f}$ generiert:
    \begin{align}
         \psi : \mathcal{L}(\mathcal{H}_1 \otimes \mathcal{H}_2) \longrightarrow& \: \mathcal{L}(\mathcal{H}), \nonumber\\
         G \longmapsto& \: \: \{\psi_{e,f}(x) : x \in G\}
    \end{align}
    \item [\textit{iii)}] Ist $h_1$ linear und $h_2$ antilinear, so wird ein Isomorphismus zwischen $\mathcal{L}(\mathcal{H})$ und $\mathcal{L}(\mathcal{H}^*_1 \otimes \mathcal{H}_2)$ über $\mu_{e,f}$ generiert:
    \begin{align}
         \mu : \mathcal{L}(\mathcal{H}^*_1 \otimes \mathcal{H}_2) \longrightarrow& \: \mathcal{L}(\mathcal{H}), \nonumber\\
         G \longmapsto& \: \: \{\mu_{e,f}(x) : x \in G\}
    \end{align}
    \item [\textit{iv)}] Ist $h_1$ antilinear und $h_2$ linear, so wird ein Isomorphismus zwischen $\mathcal{L}(\mathcal{H})$ und $\mathcal{L}(\mathcal{H}^*_1 \otimes \mathcal{H}_2)$ über $\nu_{e,f}$ generiert:
    \begin{align}
         \nu : \mathcal{L}(\mathcal{H}^*_1 \otimes \mathcal{H}_2) \longrightarrow& \: \mathcal{L}(\mathcal{H}), \nonumber\\
         G \longmapsto& \: \: \{\nu_{e,f}(x) : x \in G\}
    \end{align}
\end{itemize}
Alle diese Abbildungen sind Isomorpismen zwischen Propositionensystemen, denn es gilt z.B. für die Abbildung $\phi$, dass 
\begin{align}
    \phi(\sqcup_{s \in S} G_s) =& \: \{\phi_{e,f}(x) : x \in \sqcup_{s \in S} G_s\} \nonumber\\ =& \: \sqcup_{s \in S} \{\phi_{e,f}(x) : x \in G_s\} \nonumber\\ =& \: \sqcup_{s \in S} \phi(G_s)
\end{align}
ist, wobei $S$ eine beliebige Indexmenge ist. Weiter gilt, da $\phi_{e,f}$ unitär ist, dass
\begin{align}
    \phi(G^{\perp}) =& \: \{\phi_{e,f}(x) : x \in G^{\perp}\} \nonumber\\ =& \: \{\phi_{e,f}(x) : x \in G\}^{\perp} \nonumber\\ =& \: \phi(G)^{\perp}
\end{align}
ist. Analoges gilt für die Abbildungen $\psi, \mu$ und $\nu$, womit das Theorem gezeigt ist.
\end{proof}

Am Ende wollen wir dieses Resultat noch kurz diskutieren. Man kann sich nun beispielsweise die Frage stellen, ob dieses Ergebnis ein zu Erwartendes war oder nicht. Dazu betrachten wir folgendes Beispiel:

\begin{example}
Wir betrachten ein zusammengesetztes quantenmechanisches System, aufgebaut aus zwei Einteilchensystemen. Wollen wir nur die einzelnen Einteilchensysteme beschreiben, so modellieren wir den Zustand dieser wie in Abschnitt 4.1 über eine (normierte) Wellenfunktion. Bezeichne $\psi_1(r_1)$ die Wellenfunktion des ersten Teilchens zu einer festen Zeit $t$ und $\psi_2(r_2)$ die Wellenfunktion des zweiten Teilchens zur selben festen Zeit. Die Orte $r_i$ referieren dabei an die möglichen Orte des $i$-ten Teilchens mit $i=1,2$. 

Für die beiden Wellenfunktionen gilt
\begin{align}
    \int_{\mathbb{R}^3} |\psi_1(r_1)|^2 \,d \mu (r_1) = 1
\end{align}
und 
\begin{align}
    \int_{\mathbb{R}^3} |\psi_2(r_2)|^2 \,d \mu (r_2) = 1.
\end{align}
Ein natürlicher Ansatz für eine Wellenfunktion, welche das Gesamtsystem zur Zeit $t$ modelliert, wäre gegeben durch $\psi(r_1,r_2)$, sodass
\begin{align}
    \int_{\mathbb{R}^3} |\psi(r_1,r_2)|^2 \,d \mu (r_1) \,d \mu (r_2) = 1.
\end{align}
Ein besonders einfaches Beispiel, welches diesen Ansatz veranschaulichen soll, liegt vor, wenn die beiden obigen Teilchen unabhängig zueinander sind. In diesem Fall setzen wir die Gesamtwellenfunktion als 
\begin{align}
    \psi(r_1,r_2) = \psi_1(r_1)\psi_2(r_2)
\end{align}
an. Das $\psi$ tatsächlich ein System zweier unabhängiger Teilchen beschreibt, ergibt sich daraus, dass sich die Wahrscheinlichkeit $\mathbb{P}^{\psi}(B)$, mit $B = \{(r_1,r_2) \in B_1 \times B_2 \: : \: B_1, B_2 \subseteq \mathbb{R}^3\} \subseteq \mathbb{R}^6$ messbar, dass das eine Teilchen im (messbarem) Raumvolumen $B_1 \subseteq \mathbb{R}^3$ und das andere Teilchen im (messbarem) Raumvolumen $B_2 \subseteq \mathbb{R}^3$ zu finden ist, sich zu 
\begin{align}
    \mathbb{P}^{\psi}(B) = \mathbb{P}^{\psi_1}(B_1)\cdot\mathbb{P}^{\psi_2}(B_2),
\end{align}
ergibt. 

Wir modellieren also das zusammengesetzte System über eine quadratintegrable Funktion aus $L^2(\mathbb{R}^3 \times \mathbb{R}^3)$.

Der Grund, warum das Ergebnis von Theorem 5.5 zum Teil erwartbar war, liegt nun darin begründet, dass $L^2(\mathbb{R}^3 \times \mathbb{R}^3)$ isomorph zu $L^2(\mathbb{R}^3) \otimes L^2(\mathbb{R}^3)$ ist (\cite{reed2012methods}, Seite 51).
\end{example}

Desweiteren kann man nach der Bedeutung des Resultats von Theorem 5.5 fragen. Dazu folgende einfache Beobachtung:

Seien $\mathcal{H}_1$ und $\mathcal{H}_2$ zwei Hilberträume. Wir bilden den Hilbertraum $\mathcal{H}_1 \otimes \mathcal{H}_2$. Es gibt nun grob zwei Typen von Elementen in $\mathcal{H}_1 \otimes \mathcal{H}_2$: Separable Elemente, welche von der Form 
\begin{align}
    \Psi = \phi_1 \otimes \psi_2, \qquad \phi_1 \in \mathcal{H}_1, \: \psi_2 \in \mathcal{H}_2 \label{eq:separa}
\end{align}
sind, und nicht separable Elemente, welche sich nicht auf die Form \eqref{eq:separa} bringen lassen. Ein nicht separables Element ist z.B. gegeben durch  
\begin{align}
    \Phi = \phi_1 \otimes \psi_2 + \phi_2 \otimes \psi_1, \qquad \phi_1,\phi_1 \in \mathcal{H}_1, \: \psi_1,\psi_2 \in \mathcal{H}_2,
\end{align}
wobei $\phi_1 \neq \phi_2$ und $\psi_1 \neq \psi_2$ ist. Elemente wie diese ließen sich mittels eines kartesischen Produkts von $\mathcal{H}_1$ und $\mathcal{H}_2$ nicht darstellen. 

Wenn man also ein zusammengesetztes System über $\mathcal{H}_1 \otimes \mathcal{H}_2$ modelliert, lohnt es sich daher, nach der physikalischen Interpretation derartiger Zustände zu fragen. Wie sich zeigt, gibt es in der Quantenphysik ein bestimmtes Phänomen, welches nicht in klassischen Systemen auftreten kann und welches über derartige Zustände beschrieben wird \cite{nielsen2000quantum}. Dieses Phänomen nennt sich Verschränkung. Als verschränkt gelten immer dann zwei Teilchen, wenn beiden Teilchen zusammen ein wohldefinierter Zustand zugeordnet werden kann, den einzelnen Teilchen selbst aber nicht.

Demzufolge liegt die Bedeutung von Theorem 5.5 unter anderem darin, dass mittels dieses Resultats gewisse quantenmechanische Effekte zwischen den Teilsystemen eines zusammengesetzten Systems, wie die der Verschränkung, berücksichtigt werden.




\newpage

\section{Zusammenfassung und Ausblick}

In dieser Arbeit wurden Logiksysteme der klassischen Mechanik und der Quantenmechanik präsentiert und von ihrer mathematischen Struktur her verglichen. Während das klassische Logiksystem mathematisch durch einen  vollständigen, orthokomplementierten, distributiven und atomaren Verband beschrieben wird, wird die Quantenlogik erfasst durch einen vollständigen, orthomodularen, irreduziblen, schwach modularen, atomaren Verband, für den das Überdeckungsgesetz gilt und für den, im Gegensatz zum Logikmodell der klassischen Mechanik, das Distributivgesetz im Allgemeinen nicht mehr erfüllt sein muss, was einen wesentlichen Unterschied beider mathematischer Modelle darstellt. 

Es wurde darüber hinaus gezeigt, inwiefern sich beide Systeme bei der Beschreibung spezieller zusammengesetzter physikalischer Systeme unterscheiden. Während im klassischen Logikmodell zur Beschreibung zusammengesetzer Systeme das kartesische Produkt gebraucht wird, benötigt man in der Quantenlogik zur Beschreibung derartiger Systeme das Tensorprodukt.

Wichtig anzumerken ist dabei, dass dieses Resultat streng genommen nur für zusammengesetzte Systeme gezeigt wurde, die auf klassische Weise aneinander koppeln, d.h. das wir keine zusammengesetzten Systeme in die Betrachtung aufgenommen haben, bei denen quantenmechanische Effekte wie die der Verschränkung zwischen den einzelnen Teilsystemen entstehen können.

Zu beachten ist aber, dass rein von den Postulaten der Quantenmechanik, die gezeigten Resultate auch für zusammengesetzte Systeme gültig sind, bei denen quantenmechanische Effekte zwischen den einzelnen Teilsystemen zugelassen sind. Die Plausibilität dieses Postulates lässt sich damit begründen, dass Elemente im Tensorprodukt $\mathcal{H}_1 \otimes \mathcal{H}_2$ wie z.B. $\frac{1}{\sqrt{2}}(\phi_1 \otimes \psi_2 - \phi_2 \otimes \psi_1)$, wobei $\phi_1, \phi_2$ Basisvektoren aus $\mathcal{H}_1$ und $\psi_1,\psi_2$ Basisvektoren aus $\mathcal{H}_2$ sind, verschränkte Zustände beschreiben \cite{nielsen2000quantum}. Es ist daher anzunehmen, dass sich die obigen Resultate verallgemeinern lassen auf Systeme, die nicht auf klassische Weise aneinander gekoppelt werden.\\

Zum Abschluss wollen wir uns noch mit der Frage beschäftigen, wie ein quantenmechanisches Propositionensystem im obigen Sinne aussehen könnte, bei welchem auch Aussagen die sich z.B. auf den Ort oder den Impuls des Teilchens beziehen, aussehen könnte. Eine Idee wäre, den obigen Quantenlogikformalismus mittels des sogenannten Gelfand-Tripels \cite{de2001quantum} aus der Funktionalanalysis zu erweitern, um wirklich alle physikalischen Aussagen an das System zu berücksichtigen. Dabei geht man (im eindimensionalen) grob folgendermaßen vor: Man schränkt alle Observablenoperatoren auf eine dichte Teilmenge im $L^2(\mu)$, dem sogenannten Schwartzraum $\mathcal{S}$, ein und bildet den topologischen Dualraum des Schwartzraumes. Dieser Dualraum ist der Raum der temperierten Distributionen. Da der $L^2(\mu)$ isometrisch isomorph zu seinem topologischen Dualraum ist, lässt sich der $L^2(\mu)$ kanonisch in den Raum der temperierten Distributionen einbetten. Unter dieser Einbettung liegt der Schwartzraum dicht in seinem eigenen topologischen Dualraum, weshalb sich alle Observablenoperatoren, welche auf $\mathcal{S}$ definiert sind, stetig auf dem Raum der temperierten Distributionen hochheben lassen: Sei $\mathcal{O}_\mathcal{V}$ ein Observablenoperator der Observable $\mathcal{V}$. Dann ist die Hebung $\hat{\mathcal{O}}_\mathcal{V}$ von $\mathcal{O}_\mathcal{V}$,
\begin{align}
    (\mathcal{O}_\mathcal{V} \psi )(x) := \sum\limits_{m \in M \subset \mathbb{N}_0} D_m(x) \psi^{(m)}(x), \qquad D_m \: \: \textit{Polynom $\forall m \in M$}, \psi \in \mathcal{S},
\end{align}
auf den Raum der temperierten Distributionen definiert als
\begin{align}
    (\hat{\mathcal{O}}_\mathcal{V}(\Phi))(\psi) := \Phi\biggl( \sum\limits_{m \in M} (-1)^m \psi^{(m)}(x) \biggr), \qquad \Phi \in \mathcal{S}^*.
\end{align}
Das Spektrum $spec(\mathcal{O}_\mathcal{V})$ des Operators $\mathcal{O}_\mathcal{V}$ ergibt sich nun als 
\begin{align}
    spec(\mathcal{O}_\mathcal{V}) = \{\lambda \in \mathbb{R} \: \: : \: \: \exists \Phi \in \mathcal{S}^*\setminus\{0\} \: \: \textit{mit} \: \: \hat{\mathcal{O}}_\mathcal{V} \Phi = \lambda \Phi\}.
\end{align}
Nun können wir unsere Philosophie von oben bemühen und nun jede Proposition, die Aussagen über den Zustand von $\Sigma$ machen, als Unterraum in $\mathcal{S}^*$ auffassen, nur das nun die assoziierten Unterräume derjenigen Propositionen, die zu Ja/Nein-Experimenten korrespondieren, aus Eigendistributionen konstruiert werden! Offene Fragen bzgl. dieser Konstruktion, die ich mir aufgrund Zeitmangels selbst nicht beantworten konnte, sind: Welche Topologie wählt man auf $\mathcal{S}^*$, um über abgeschlossene Unterräume in $\mathcal{S}^*$ sprechen zu können und wie führt man ohne zur Verfügung stehendes Skalarprodukt eine sinnvolle Komplementbildung in $\mathcal{S}^*$ ein, die kompatibel mit dem Komplementbegriff in $\mathcal{L}$ ist?

\newpage




\addcontentsline{toc}{section}{Anhang}
\appendix       
\section*{Anhang}

Bezogen auf den Beweis von Lemma 5.9 beweisen wir die sogenannte Polarisationsformel.

\begin{lemma}
Sei $\mathcal{H}$ ein komplexer Hilbertraum und bezeichne $\langle \cdot,\cdot \rangle$ sein Skalarprodukt, welches im ersten Argument antilinear ist, und sei $\|\cdot\| := \sqrt{\langle \cdot,\cdot \rangle}$ die vom Skalarprodukt induzierte Norm. Wir definieren die Größe 
\begin{align}
    R(x,y) := \frac{1}{4}(\|x+y\|^2 - \|x-y\|^2)
\end{align}
Dann gilt die sogenannte Polarisationsformel
\begin{align}
    \langle x,y \rangle = R(x,y) - iR(x,iy)
\end{align}
und 
\begin{align}
    &R(x,y) = R(y,x), \\ &R(x,iy) = -R(ix,y).
\end{align}
\end{lemma}

\begin{proof}
Es gilt 
\begin{align}
    R(x,y) - iR(x,iy) =& \: \frac{1}{4}(\|x+y\|^2 - \|x-y\|^2) - \frac{i}{4}(\|x+iy\|^2 - \|x-iy\|^2) \nonumber\\ =& \: \frac{1}{4}(\|x+y\|^2 - \|x-y\|^2 - i\|x+iy\|^2 + i\|x-iy\|^2) \nonumber\\ =& \: \frac{1}{4}(\langle x+y,x+y \rangle - \langle x-y,x-y \rangle \nonumber\\ &- \: i\langle x+iy,x+iy \rangle + i\langle x-iy,x-iy \rangle) \nonumber\\ =& \: \frac{1}{4}(\langle x,x \rangle + \langle y,y \rangle + \langle x,y \rangle + \langle y,x \rangle \nonumber\\ &- \: \langle x,x \rangle - \langle y,y \rangle +  \langle x,y \rangle + \langle y,x \rangle \nonumber\\ &- \: i(\langle x,x \rangle + \langle y,y \rangle + \langle x,iy \rangle + \langle iy,x \rangle) \nonumber\\ &+ \: i( \langle x,x \rangle + \langle y,y \rangle - \langle x,iy \rangle - \langle iy,x \rangle)) \nonumber\\ =& \: \frac{1}{4}(\langle x,x \rangle + \langle y,y \rangle + \langle x,y \rangle + \langle y,x \rangle \nonumber\\ &- \: \langle x,x \rangle - \langle y,y \rangle +  \langle x,y \rangle + \langle y,x \rangle \nonumber\\ &- \: i\langle x,x \rangle - i\langle y,y \rangle + \langle x,y \rangle - \langle y,x \rangle \nonumber\\ &+ \: i\langle x,x \rangle + i\langle y,y \rangle +  \langle x,y \rangle - \langle y,x \rangle) \nonumber\\ =& \: \frac{1}{4}(2\langle x,y \rangle + 2\langle y,x \rangle + 2\langle x,y \rangle - 2\langle y,x \rangle) \nonumber\\ =& \: \frac{1}{4}(4\langle x,y \rangle) \nonumber\\ =& \: \langle x,y \rangle
\end{align}
Um die letzten beiden Identitäten zu zeigen, rechnen wir explizit nach:
\begin{align}
    4R(x,y) =& \: \|x+y\|^2 - \|x-y\|^2 \nonumber\\ =& \: \langle x+y,x+y \rangle - \langle x-y,x-y\rangle \nonumber\\ =& \: \langle x,x \rangle + \langle y,y \rangle + \langle x,y \rangle + \langle y,x \rangle \nonumber\\ &- \: \langle x,x \rangle - \langle y,y \rangle + \langle x,y \rangle + \langle y,x \rangle \nonumber\\ =& \: \langle y,y \rangle + \langle x,x \rangle + \langle y,x \rangle + \langle x,y \rangle \nonumber\\ &- \: \langle y,y \rangle - \langle x,x \rangle + \langle y,x \rangle + \langle x,y \rangle \nonumber\\ =& \: \langle y+x,y+x \rangle - \langle y-x,y-x \rangle \nonumber\\ =& \: \|y+x\|^2 - \|y-x\|^2 \nonumber\\ =& \: 4R(y,x)
\end{align}
\begin{align}
    4R(x,iy) =& \: \|x+iy\|^2 - \|x-iy\|^2 \nonumber\\ =& \: \langle x+iy,x+iy \rangle - \langle x-iy,x-iy\rangle \nonumber\\ =& \: \langle x,x \rangle + \langle y,y \rangle + \langle x,iy \rangle + \langle iy,x \rangle \nonumber\\ &- \: \langle x,x \rangle - \langle y,y \rangle + \langle x,iy \rangle + \langle iy,x \rangle \nonumber\\ =& \: \langle x,x \rangle + \langle y,y \rangle + i\langle x,y \rangle - i\langle y,x \rangle \nonumber\\ &- \: \langle x,x \rangle - \langle y,y \rangle + i\langle x,y \rangle - i\langle y,x \rangle \nonumber\\ =& \: -\langle x,x \rangle - \langle y,y \rangle + i\langle x,y \rangle - i\langle y,x \rangle \nonumber\\ &+ \: \langle x,x \rangle + \langle y,y \rangle + i\langle x,y \rangle - i\langle y,x \rangle \nonumber\\ =& \: -\langle ix,ix \rangle - \langle y,y \rangle - \langle ix,y \rangle - \langle y,ix \rangle \nonumber\\ &+ \: \langle ix,ix \rangle + \langle y,y \rangle - \langle ix,y \rangle - \langle y,ix \rangle \nonumber\\ =& \: -\langle ix+y,ix+y \rangle + \langle ix-y,ix-y \rangle \nonumber\\ =& \: -(\langle ix+y,ix+y \rangle - \langle ix-y,ix-y \rangle) \nonumber\\ =& \: -(\|ix+y\|^2 - \|ix-y\|^2) \nonumber\\ =& \: -4R(ix,y)
\end{align}
\end{proof}

\newpage

\addcontentsline{toc}{section}{Literatur}

\bibliographystyle{plain}

\bibliography{bibliography.bib}
\end{document}